\documentclass[a4paper,12pt]{article}
\usepackage{./jheppub}
\usepackage[utf8]{inputenc}
\usepackage[T1]{fontenc}
\usepackage[italian,english]{babel}
\usepackage{amsmath,amssymb,amsthm}
\usepackage{graphicx,float}
\usepackage{subfig}
\usepackage{tikz-cd}
\usepackage{slashed}
\usepackage{hyperref}
\usepackage{bbm}
\usepackage[page]{appendix}
\usepackage{tikz}
\usetikzlibrary{arrows,decorations.markings,positioning,shapes,calc,patterns,tqft,intersections,shapes.misc}
\tikzset{
    myarrow/.style={
        draw=black!100,
        fill=red!100,
        single arrow,
        minimum height=3.5ex,
        single arrow head extend=1ex
    }
}
\tikzset{
    mydoublearrow/.style={
        draw=black!100,
        fill=red!100,
        double arrow,
        minimum height=6ex,
        double arrow head extend=.8ex
    }
}

\tikzset{->-/.style={decoration={markings, mark=at position 0.5 with {\arrow{stealth}}},postaction={decorate}},}
\tikzset{GaugeNode/.style={circle,draw,inner sep=0pt,minimum size=10mm}}
\tikzset{FrameNode/.style={rectangle,draw,inner sep=0pt,minimum size=9mm}}
\tikzset{token/.style={circle,double,draw=black!70,fill=black!50,inner sep=0pt,minimum size=3mm}}

\newdimen\shadedBaseline\shadedBaseline=-4mm
\newcount\tableauRow\newcount\tableauCol
\newcommand\ShadedTableau[2][\relax]{%
  \begin{tikzpicture}[scale=0.4,draw/.append style={thick,black},baseline=\shadedBaseline]
    \ifx\relax#1\relax%
    \else 
      \foreach\bx in {#1} { \filldraw[blue!30]\bx+(-.5,-.5)rectangle++(.5,.5); }
    \fi
    \tableauRow=0
    \foreach \Row in {#2} {
       \tableauCol=1
       \foreach\k in \Row {
          \draw(\the\tableauCol,\the\tableauRow)+(-.5,-.5)rectangle++(.5,.5);
          \draw(\the\tableauCol,\the\tableauRow)node{\k};
          \global\advance\tableauCol by 1
       }
       \global\advance\tableauRow by -1
    }
  \end{tikzpicture}%
}

\newdimen\shadedBaseline\shadedBaseline=-4mm
\newcount\tableauRow\newcount\tableauCol
\newcommand\ShadedTableauS[2][\relax]{%
  \begin{tikzpicture}[scale=0.15,draw/.append style={thick,black},baseline=\shadedBaseline]
    \ifx\relax#1\relax%
    \else 
      \foreach\bx in {#1} { \filldraw[blue!30]\bx+(-.5,-.5)rectangle++(.5,.5); }
    \fi
    \tableauRow=0
    \foreach \Row in {#2} {
       \tableauCol=1
       \foreach\k in \Row {
          \draw(\the\tableauCol,\the\tableauRow)+(-.5,-.5)rectangle++(.5,.5);
          \draw(\the\tableauCol,\the\tableauRow)node{\k};
          \global\advance\tableauCol by 1
       }
       \global\advance\tableauRow by -1
    }
  \end{tikzpicture}%
}

\newdimen\shadedBaseline\shadedBaseline=-4mm
\newcount\tableauRow\newcount\tableauCol
\newcommand\ShadedTableauRed[2][\relax]{%
  \begin{tikzpicture}[scale=0.4,draw/.append style={thick,black},baseline=\shadedBaseline]
    \ifx\relax#1\relax%
    \else 
      \foreach\bx in {#1} { \filldraw[red!30]\bx+(-.5,-.5)rectangle++(.5,.5); }
    \fi
    \tableauRow=0
    \foreach \Row in {#2} {
       \tableauCol=1
       \foreach\k in \Row {
          \draw(\the\tableauCol,\the\tableauRow)+(-.5,-.5)rectangle++(.5,.5);
          \draw(\the\tableauCol,\the\tableauRow)node{\k};
          \global\advance\tableauCol by 1
       }
       \global\advance\tableauRow by -1
    }
  \end{tikzpicture}%
}

\newcommand\ShadedTableauR[3][\relax]{%
  \begin{tikzpicture}[scale=0.4,draw/.append style={thick,black},baseline=\shadedBaseline]
    \ifx\relax#1\relax%
    \else 
      \foreach\bx in {#1} { \filldraw[blue!30]\bx+(-.5,-.5)rectangle++(.5,.5); }
    \fi
    \tableauRow=0
    \foreach \Row in {#3} {
       \tableauCol=1
       \foreach\k in \Row {
          \draw(\the\tableauCol,\the\tableauRow)+(-.5,-.5)rectangle++(.5,.5);
          \draw(\the\tableauCol,\the\tableauRow)node{\k};
          \global\advance\tableauCol by 1
       }
       \global\advance\tableauRow by -1
    }
    \ifx\relax#2\relax%
    \else 
      \foreach\bx in {#2} { \draw[pattern=north east lines,pattern color=black]\bx+(-.5,-.5)rectangle++(.5,.5); }
    \fi
    \tableauRow=0
    \foreach \Row in {#3} {
       \tableauCol=1
       \foreach\k in \Row {
          \draw(\the\tableauCol,\the\tableauRow)+(-.5,-.5)rectangle++(.5,.5);
          \draw(\the\tableauCol,\the\tableauRow)node{\k};
          \global\advance\tableauCol by 1
       }
       \global\advance\tableauRow by -1
    }
  \end{tikzpicture}%
}

\DeclareMathOperator{\Tr}{Tr}

\DeclareMathOperator{\rk}{rk}

\DeclareMathOperator{\ch}{ch}
\DeclareMathOperator{\Hilb}{Hilb}

\DeclareMathOperator{\PL}{PL}

\DeclareMathOperator{\Ind}{Ind}

\DeclareMathOperator{\Det}{Det}
\DeclareMathOperator{\Td}{Td}

\DeclareMathOperator{\vol}{vol}
\newcommand{\ie}{\textit{i.e. }}
\newcommand{\eg}{\textit{e.g. }}
\newcommand{\git}{/\!\!/}
\newcommand{\de}{{\text d}}
\newcommand{\pd}{{\partial}}
\newcommand{\iu}{{\text i}}
\newcommand{\eu}{{\text e}}
\newtheorem{theorem}{Theorem}

\newtheorem{proposition}{Proposition}

\newtheorem{conjecture}{Conjecture}

\newtheorem{example}{Example}
\DeclareMathOperator{\Hom}{Hom}
\DeclareMathOperator{\End}{End}

\title{Defects, nested instantons and \\ comet shaped quivers}
\date{}
\author{G.~Bonelli, N.~Fasola and A.~Tanzini}
\affiliation{S.I.S.S.A. - Scuola Internazionale di Studi Scientifici Avanzati\\ via Bonomea, 265 - 34136 Trieste ITALY}
\affiliation{I.N.F.N. -- Sezione di Trieste}
\affiliation{I.G.A.P. -- Institute for Geometry and Physics\\ via Beirut, 4 -34100 Trieste ITALY}
\abstract{We introduce and study a surface defect in four dimensional gauge theories supporting nested instantons with respect to the parabolic reduction of the gauge group at the defect.
This is engineered from a D3/D7-branes system on a non compact Calabi-Yau threefold $X$. For $X=T^2\times T^*{\mathcal C}_{g,k}$, the product of a two torus $T^2$ times the 
cotangent bundle over a Riemann surface ${\mathcal C}_{g,k}$ with marked points, we propose an effective theory in the limit of small volume of ${\mathcal C}_{g,k}$ given as a 
comet shaped quiver gauge theory on $T^2$, 
the tail of the comet being made of a flag quiver for each marked point
and 
the head describing the degrees of freedom due to the genus $g$. 
Mathematically, we obtain for a single D7-brane
conjectural  explicit formulae for the virtual equivariant elliptic genus of a certain bundle over the moduli space of the nested Hilbert scheme of points on the affine plane.
A connection with elliptic cohomology of character varieties and an elliptic version of modified Macdonald polynomials naturally arises.}
\begin{document}
\maketitle
\section{Introduction and discussion}

The study of defects can be used to characterise the behaviour of physical theories and the related mathematical structures. In this paper we are interested in {\it surface defects} in four-dimensional supersymmetric gauge theories,
namely real codimension two submanifolds were a specific reduction of the gauge connection takes place. This kind of defects has been widely investigated in many contexts from various different perspectives. 
The study of the r\^ole of defects in the classification of the phases of gauge theories was pioneered by 't Hooft \cite{thooft}. 
Surface defects were introduced by Kronheimer and Mrowka \cite{KM1,KM2} in the study of Donaldson invariants, while their r\^ole in the context of Geometric Langlands correspondence was emphasized in \cite{Gukov:2006jk}.
The correspondence with two-dimensional conformal field theories \cite{AGT} prompted a systematic analysis of surface defects and highlighted their relevance for quantum integrable systems \cite{Alday:2009fs,Nekrasov:2017rqy} and for 
the study of isomonodromic deformations and Painlev\'e equations \cite{Gamayun:2012ma,Bonelli:2016qwg,Bonelli:2019boe}. In this paper we introduce and study surface defects supporting {\it nested instantons} with respect to the parabolic reduction of the gauge group at the defect.
These defects are engineered from a D7/D3 brane system on a local compact complex surface $S$. The brane engineering naturally leads to a description of these defects and their effective supersymmetric field theories in terms of  moduli spaces of representations of 
quivers in the category of vector spaces, the objects being the branes and morphisms the open strings suspended among them.  
Supersymmetric partition functions of these systems provide conjectural formulae for topological invariants of these moduli spaces, more precisely, since these are generically not smooth, for {\it virtual} invariants of them.

The brane system we consider is D7/D3 on a local four-fold embedded in the ten dimensional IIB superstring supersymmetric background, the D3 branes wrapping the non-contractible cycle $S$.
The D3 branes effective theory is the  topologically twisted Vafa-Witten (VW) theory 
\cite{Vafa:1994tf}
with two extra chiral multiplets in the fundamental describing the D7/D3 open string sector.
The D7 branes gauge theory is related to (equivariant) Donaldson-Thomas theory \cite{Donaldson:1996kp} on the fourfold.
Actually, we consider these theories in a non-trivial $\Omega$-background corresponding to the 
equivariant parameters associated to rotations along the non-compact directions of the fourfold. This lead to a refinement of the above mentioned gauge theories. We focus on the case $S=T^2\times \mathcal{C}$, the last being a Riemann surface with punctures $\{p_i\}$.
Surface operators of this four-dimensional gauge theory are real codimension two defects located at $T^2\times\{p_i\}$.

The effective theory describing the dynamics of such surface defects is obtained in the limit of small area of $\mathcal{C}$ and turns out to be a quiver gauged linear sigma model which flows in the infrared to a non-linear sigma model of maps
from $T^2$ to the moduli space of {\it nested instantons}. This is a generalisation of the usual ADHM instanton moduli space, structured on  the decomposition of the gauge connection at the surface defect.
It is obtained from the usual ADHM instanton moduli space by implementing a suitable orbifold action which generates the fractional fluxes of the gauge field at the defect.
The partition function of the D7/D3 effective theory computes the equivariant (virtual) elliptic genus of this moduli space in presence of matter content dictated by the topology of $\mathcal{C}$, which, for genus $g$ amounts to 
$g$-hypermultiplets in the adjoint representation. Their contribution is encoded in a bundle ${\mathcal V}_g$ over the moduli space of nested instantons.
The general formula for the elliptic genus is \eqref{216} which, in the particular case $r=1$ and $k=1$, calculates 
the virtual elliptic genus of the bundle ${\mathcal V}_g$ over the nested Hilbert scheme of points on ${\mathbb C}^2$. The explicit combinatorial expression of \eqref{216}  is given by \eqref{ZT2-g-k-full} in terms of nested partitions. 

We also study the circle reduction of this system, which leads to a T-dual D6/D2 quantum mechanics. In this case, we find that the generating function of the defects, obtained by summing over all possible decompositions of the connection
at the puncture, or in other terms over all possible tails of the quiver, displays a very nice polynomial structure in the equivariant parameters.

The method we used to compute the partition function of the D-brane system is twofold. 
One, worked out in section 3.1,
makes use of  superlocalisation formulae \cite{Bruzzo:2002xf} directly leading to a sum over fixed points with weights computed from the character
of the torus action on the nested instanton moduli space.
An alternative derivation is performed in section 3.2, where
the $T^2$ partition function is evaluated via a higher dimensional contour integral {\it \`a la} \cite{Moore:1998et}.
This can
be also prescribed via Jeffrey-Kirwan residue method
\cite{Benini:2013nda,Benini:2013xpa}, as it was used in the study of D1/D7 BPS  bound state counting on $\mathbb{C}^3$  in  \cite{Benini:2018hjy}. We remark that although
the residue method is computationally more demanding,
 it has the advantage of allowing the study of wall-crossing among spaces 
with different stability conditions by changing the integration contour 
\cite{Bonelli:2013mma,Ashok:2018zxp}.

%

When one considers a single D7 brane, the nested instanton moduli space reduces to the nested Hilbert scheme of points on $\mathbb{C}^2$. Our brane construction provides a conjectural description of this space as the moduli space of representations of the quiver considered in Section 2.6.  Moreover, in this case the summation over the tails of the quiver gives rise to polynomials related to the modified Macdonald polynomials, and the whole partition function is related to the generating function 
introduced in \cite{hausel2011} to describe the cohomology of character varieties. 
The analog result for the full $T^2$ partition function gives rise to 
special combinations of elliptic functions which can be regarded 
as an elliptic lift of these polynomials. We display few examples in equations
\eqref{duepall},\eqref{trepall},\eqref{quattropall}. These formulae 
should encode the elliptic cohomology of character varieties and
can be viewed as an elliptic virtual 
refinement of the generating function of \cite{hausel2011}. 
We remark that the D6/D2 quantum mechanical system and its relation with \cite{hausel2011}  was studied in \cite{Chuang:2013wpa} via a different approach based on topological string amplitudes on orbifold Calabi-Yau. 

The relation with character varieties can be understood from the fact that Vafa-Witten theory on $S=T^2\times \mathcal{C}$ is known to reduce in the small $\mathcal{C}$ limit to a GLSM on $T^2$ with target space the Hitchin
moduli space over $\mathcal{C}$ \cite{Bershadsky:1995vm}. This in turn is homeomorphic \cite{Simpson} to the character variety of    $\mathcal{C}$, namely the moduli space of representations of the first fundamental group of  $\mathcal{C}\backslash \{p_i\}$  into $GL_n(\mathbb{C})$
with fixed semi-simple conjugacy classes at the punctures.

There are some open questions to be discussed about the above construction.
Actually, the 2d $(2,0)$ D3/D7 quiver gauge theory that we consider is
{\it anomalous}, the D3/D7 open string modes breaking $(2,2)$ to $(2,0)$ and generating an R-symmetry anomaly.
Indeed  instantons in the D7 brane gauge theory are sourced from D3 branes.
The mathematical counterpart is that Donaldson-Thomas (DT) theory on fourfolds has positive virtual dimension and requires the insertion of observables to produce the appropriate measure on the moduli space \cite{cao,Cao:2017swr}.
To cure this, we introduce new fields with opposite representations with respect to the gauge group and global symmetries. These are sources of the insertion of suitable observables which compensate the R-symmetry anomaly.
Actually, the extra fields we consider can be thought as arising from coupling of D3 branes to  $\overline{\rm D7}$-branes. It was recently conjectured \cite{Sen-slides}, that  
$\overline{\rm D7}$/D7 system undergoes tachyon condensation leaving behind 
D3-branes. This proposal is a generalisation of the known condensation \cite{Sen:1998sm} of $\overline{\rm D5}$/D5 into D3s.
Indeed in our calculations we find that, at special values of the equivariant parameters,
the contribution of the D3/$\overline{\rm D7}$ and D3/D7 modes to the elliptic genus cancels out, in line with the above expectations.
It would be extremely interesting to further analyse a possible application of our technique to the string field theoretic description of D-branes/anti D-branes annihilation.

The mathematical implication of all this is that DT theory on the local surface four-fold should reduce to VW theory on the complex surface $S$ itself,
the corresponding partition function providing conjectural formulae for VW invariants on $S$ in presence of surface defects. We aim to further investigate this reduction in the future and to analyse the elliptic genus
of the nested instanton moduli space and in particular of the nested Hilbert scheme of points on toric surfaces. This can be obtained via gluing the contributions of the local patches \cite{localising,Bershtein:2015xfa,Bershtein:2016mxz}.
Let underline that our computations concern a {\it refined} version of VW theory,
a refinement being given by the mass $m$ of the adjoint hypermultiplet.
Therefore, by studying the limit at $m\to\infty$, with appropriately rescaled gauge coupling, we reduce to pure twisted ${\mathcal N}=2$ gauge theory computing higher rank equivariant Donaldson invariants of $S$.
Moreover, while in this paper we considered $S=T^2\times {\mathcal C}$, non trivial elliptic fibrations or other product geometries can be studied. In this way our approach could be used to generalise the results on 
Donaldson invariants of \cite{Marino:1998bm,Lozano:1999us}.
The general modular properties of these generating functions are worth to be analysed
\cite{Vafa:1994tf,Manschot:2017xcr}.

The surface defects considered in this paper are directly related to Hitchin system with regular singularities. It would be obviously interesting to consider the case of irregular singularities, in particular the ones related to
Argyres-Douglas points of gauge theories \cite{Bonelli:2011aa,Bonelli:2016qwg}, and investigate their r\^ole and contribution to the above mentioned differential invariants
\cite{Moore:2017cmm}.

Moreover, the
relation of our results to representation theory and quantum integrable systems
should be explored, in particular investigating whether the cohomology of the
nested instanton moduli space hosts representations of suitable infinite dimensional Lie algebrae, generalising the results of \cite{nakajima1994,Schiffmann2013,maulik-okounkov}.
Also the
characterisation of the polynomials appearing in the quantum mechanical limit is to be worked out, by studying recurrence relations and/or difference equations
they satisfy. 
This would possibly open a window on the relation with quantum integrable systems. For example, in \cite{Bonelli:2014iza,Bonelli:2015kpa}, the relation between 
D1/D5 systems on ${\mathbb P}^1$ and quantum Intermediate Long Wave hydrodynamics was studied, finding that the mirror of the associated GLSM provides the Bethe ansatz equations of the latter. Analogous relations between the mirror of the 2d comet-shaped quiver gauge theories and suitable integrable systems are worth to be explored.
Finally, the 
F-theory uplift of our construction would help to study dualities of these defect gauge theories and to generalise them to other gauge groups.

The rest of the paper is structured in two main Sections. In the first one we provide the general brane set-up and a detailed derivation of the comet-shaped quiver from D-branes on orbifolds. We then discuss the reduction to quiver quantum mechanics and  
the relation to character varieties. In the second, we perform explicit computations of the relevant partition functions and the relation with modified Macdonald polynomials of the reduced quantum mechanical quiver theory.

{\bf Acknowledgments}

We thank 
U. Bruzzo,
M. Cirafici,
E. Diaconescu, 
O. Foda,
M.L. Frau, A. Lerda, C. Maccaferri, 
M. Manabe,
M. Marino, F. Rodriguez-Villegas and R. Thomas
for useful discussions. 
The work of G.B. is partially supported by INFN - ST\&FI.
The work of N.F. and A.T. is partially supported by INFN - GAST.
The work of A.T. is supported by PRIN project "Geometria delle variet\`a algebriche". The work of G.B. is supported by the PRIN project "Non-perturbative Aspects Of Gauge Theories And Strings".

\section{D-branes, geometry and quivers}\label{sec:branes}
\subsection{Preliminaries}

Let us start by discussing the geometric D-branes set-up.

We consider a Type IIB supersymmetric general background built as the total space of a rank three complex vector bundle $V^3_S$ on a complex surface $S$
\begin{equation} 
X_5 =  {\rm tot}\left(V^3_S\right)
\end{equation}
where supersymmetry requires the Calabi-Yau condition ${\rm det}V^3_S={\mathcal K}_S$, where ${\mathcal K}_S$ is 
the canonical bundle over $S$.
To place a D3-D7 system in such a background, we assume that
$V^3_S$ has the following reduced structure
$$
V^3_S={\mathcal K}_S\otimes {\rm det}^{-1}V^2_S\oplus V^2_S
$$
where the rank two bundle $V^2_S$ is otherwise unconstrained.

Let us therefore consider the theory of 
$N$ D3-branes wrapping the complex surface $S$
in the background of
$r$ D7-branes along the local surface fourfold ${\rm tot}\left(V^2_S\right)$.

The low-energy dynamics of the $N$ D3-branes can be obtained as usual by dimensional reduction of the 
${\mathcal N}=1$ $D=10$ supersymmetric Yang-Mills theory on $X_5$ down to their world-volume.
This produces a topologically twisted version of the ${\mathcal N}=4$ $D=4$ 
theory on $S$ 
\cite{Bershadsky:1995qy}
whose boson content is given by the gauge connection ${\mathcal A}$, a section $\Phi_{\mathcal L}$ of the line bundle ${\mathcal L}={\mathcal K}_S\otimes {\rm det}^{-1}V^2_S$ and a doublet $\Phi_{V^2}$ which is a section of $V^2_S$, the latter describing the transverse motion of the D3-branes in the ambient $X_5$.
All these fields are in the adjoint representation of the $U(N)$ gauge group.
The above set-up reduces to the Vafa-Witten topologically twisted ${\mathcal N}=4$ $D=4$
on $S$ if the rank two vector bundle $V^2_S={\mathbb C}^2$ is trivial and therefore
$X_5={\rm tot}\left({\mathcal K}_S\right)\oplus{\mathbb C}^2$.
In this case, the above construction indeed gives 
the gauge connection ${\mathcal A}$ on $S$, a complex $(2,0)$-form
$\Phi_{S}$ valued in the fiber of ${\mathcal K}_S$ describing the transverse D3-branes motion within the local surface $X_3={\rm tot}\left({\mathcal K}_S\right)$, 
while the motion along the remaining
${\mathbb C}^2$ transverse directions is described by two other complex scalars
$B_i$, with $i=1,2$. 

The effect of the additional $r$ background D7-branes on the D3-branes is kept into account by a further set
of two complex scalars $I$ and $J$ in the 
bifundamental $N\times \bar r$ and $r\times \bar N$ of the gauge symmetry group
$U(N)$ and flavour global $U(r)$ group. These are sections respectively of ${\mathcal O}_S$ and ${\rm det} V^2_S$ in general. This follows from the fact that these fields 
are in the positive chiral spinorial representation of the transverse $SO(4)$ and are therefore sections of  ${\mathcal S}_+ \sim \Lambda^{(even,0)}(V^2_S)$, for $S$ a Kahler surface. 

The continuous symmetries of this geometric set-up in the transverse directions to the D3-branes are the $\left({\mathbb C}^*\right)^3$-action on the ${\mathbb C}^3$ 
fiber of $V_S^3$ with respective weights $(\epsilon_1,\epsilon_2,m)$. 
These are the global symmetries of the 
gauge theory on $S$ which can be uses to define the relevant $\Omega$-background 
after turning on the relative background gauge fields. The parameter $m$ introduces a mass for the adjoint hypermultiplet of the four dimensional theory inducing the supersymmetry breaking from the ${\mathcal N}=4$ Vafa-Witten theory to its ${\mathcal N}=2^*$ refined version.

In the following, we will study the above general system in the case in which
the complex surface is in the product form $S=T^2\times {\mathcal C}$, where ${\mathcal C}$ is a Riemann surface and $V^2_S$ is trivial.
In this case, the canonical bundle over $S$ reduces to the holomorphic cotangent bundle  over ${\mathcal C}$
and 
$$X_5={\rm tot}\left(T^*{\mathcal C}\right)\times T^2\times {\mathbb C}^2\, .$$

In order to introduce surface defects in the gauge theory, we're going to 
generalise the above set-up to the case in which ${\mathcal C}$ is 
punctured at the points where the defects are located. 
More precisely, the parabolic reduction of the gauge bundle at the punctures 
is encoded in an orbibundle structure. The effective two dimensional field theory 
describing the dynamics of the defect is obtained from the above set-up
in the chamber of small ${\mathcal C}$ volume leading to a quiver gauged linear sigma model describing the relevant open string modes. In the IR this reduces to a non-linear sigma model of maps from $T^2$ to the moduli space of representation of the quiver above.

\subsection{D-branes on the orbicurve and defects}

Let us now generalise the above setup to the case in which ${\mathcal C}$
is an orbicurve, that is a Riemann surface with elliptic singular points.
This means that the local geometry at some marked points 
$\{P_\alpha\}$ of ${\mathcal C}$
is that of the ${\mathbb Z}_{s_\alpha}$ quotient of a disk $D$ acted by
$z_\alpha\to\omega_\alpha z_\alpha$ with $\omega_\alpha^{s_\alpha}=1$.

Placing D-branes on an orbicurve consists in excising a regular cylinder out of the 
total space of the corresponding regular vector bundle 
and prescribing new local transition functions 
defining the lift of the discrete group action
to the total space of the vector bundle. 
This operation extends the vector bundle to an orbibundle.

Let us therefore consider the geometry of the D-branes in the vicinity of a marked point $P$ of order $s$ with local coordinate $z$.
The action on the D-brane Chan-Paton factors induces a modification of the gauge symmetry due to D-branes fractionalisation \cite{Douglas:1996sw}. Let 
$\gamma_\ell$
be the number of D-branes in the $\ell^{th}$ sector, namely the one corresponding to the charge $\ell$ representation $z^\ell$ of ${\mathbb Z}_{s}$. 
This corresponds to
prescribe the new transition function at the excised disk as 
$$g_P=\bigoplus_{\ell=0}^{s-1} z^\ell {\bf 1}_{\gamma_\ell}$$
and, correspondingly the local behaviour of the gauge connection
as 
$${\mathcal A}_P= g_P^{-1} d g_P=\left(\frac{d\, z}{z}\right)\bigoplus_{\ell=0}^{s-1} \ell \,\,{\bf 1}_{\gamma_\ell}
=
\left(\frac{d\, \tilde z}{\tilde z}\right)\bigoplus_{\ell=0}^{s-1} \frac{\ell}{s} {\bf 1}_{\gamma_\ell}\, ,
$$
where $\tilde z=z^{s}$.
This finally induces the local prescription on the curvature ${\mathcal F}=d{\mathcal A}$ as
$${\mathcal F}_P/(2\pi)= \sqrt{-1}\delta(\tilde z)d{\tilde z}\wedge d\bar {\tilde z}\bigoplus_{\ell=0}^{s-1} \left(\frac{\ell}{s}\right) {\bf 1}_{\gamma_\ell}$$
which implements the realisation of the real co-dimension two defect in the four-dimensional gauge theory. Let us remark that from the algebraic geometry viewpoint this corresponds to study sheaves
on root stacks, which is a natural framework were fractional Chern classes appear \cite{Bruzzo:2009uc}.

One can better describe the resulting gauge theory structure of the local D-brane configuration from the viewpoint of the geometry of the covering disk with local coordinate $\tilde z=z^{s}$. 

\begin{figure}[H]
\centering
\begin{tikzpicture}[x=.3cm,y=.3cm]
\draw[line width=1] (10 ,0) arc (0:100:10 );
\draw[dashed,line width=1] ({10 *cos(100)},{10 *sin(100)}) arc (100:260:10 );
\draw[line width=1] ({10 *cos(260)},{10 *sin(260)}) arc (260:310:10 );
\draw[line width=1] ({10 *cos(310)},{10 *sin(310)}) arc (310:360:10 );
\draw[red] (9,0) arc (0:50:9);
\draw[red] (8.5,0) arc (0:50:8.5);
\draw[blue] (8,0) arc (0:100:8);
\draw[blue] (7.5,0) arc (0:100:7.5);
\draw[blue] (7,0) arc (0:100:7);
\draw[green] (3,0) arc (0:310:3);
\draw[green] (2.5,0) arc (0:310:2.5);
\draw (2 ,0) arc (0:360:2 );
\draw (1.5 ,0) arc (0:360:1.5 );
\foreach \s in {1,...,7}
{
	\draw[violet,dotted] ({7-\s/2},0) arc (0:260:{7-\s/2});
}
\draw[line width=1] (0,0) to (13 ,0);
\draw[red,dotted,line width=1] (0,0) to ({13 *cos(50)},{13 *sin(50)});
\draw[blue,dotted,line width=1] (0,0) to ({13 *cos(100)},{13 *sin(100)});
\draw[violet,dotted,line width=1] (0,0) to ({13 *cos(260)},{13 *sin(260)});
\draw[green,dotted,line width=1] (0,0) to ({13 *cos(310)},{13 *sin(310)});


\draw ({12*cos(25)},{12*sin(25)}) node[right] {$\textcolor{red}{\gamma_0}+\textcolor{blue}{\gamma_1}+\textcolor{violet}{\cdots}+\textcolor{green}{\gamma_{s-2}}+\gamma_{s-1}=n_0\quad \boldsymbol{(=n)}$};
\draw ({12.5*cos(60)},{12.5*sin(60)}) node[above,fill=white] {$\textcolor{blue}{\gamma_1}+\textcolor{violet}{\cdots}+\textcolor{green}{\gamma_{s-2}}+\gamma_{s-1}=n_1$};
\draw ({12*cos(290)},{12*sin(290)}) node[below] {$\textcolor{green}{\gamma_{s-2}}+\gamma_{s-1}=n_{s-2}$};
\draw ({12*cos(335)},{12*sin(335)}) node[right] {$\gamma_{s-1}=n_{s-1}$};
\end{tikzpicture}\caption{The ``brane cake'' describing the covering structure the D3 branes on the local orbifold disk.}\label{fig:cake}
\end{figure}
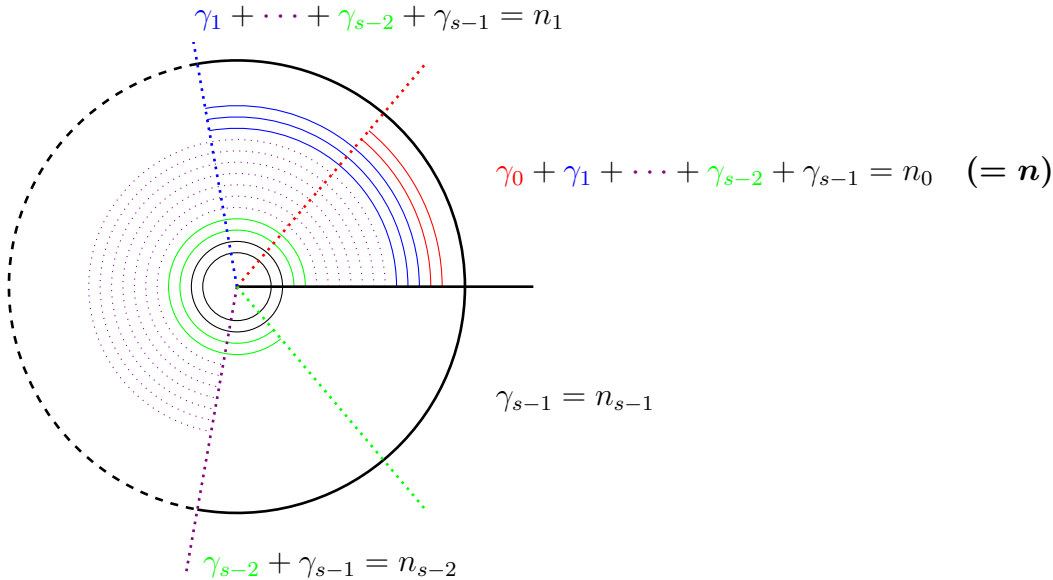

This is the $s$-covering of the quotient disk, that is given by the collection of 
$s$ consecutive Riemann sheets. 
The $\gamma_\ell$ D-branes in the $\ell^{th}$ sector and their images span $\ell$
Riemann sheets. As a consequence the $\ell^{th}$ Riemann sheet is spanned by an overall number of $n_\ell=\sum_{\ell'=\ell}^{s-1} \gamma_{\ell'}$ D-branes.
Let us notice that the outward of the quotient disk is joined to the rest of the Riemann surface by the first Riemann sheet which is consistently covered by all the 
$n_0=\sum_{\ell'=0}^{s-1}\gamma_{\ell'}=N$ D-branes. 

\subsection{Two dimensional quiver GLSM from the reduction to small ${\mathcal C}$ volume: bulk part}\label{sec:2d_bulk}

Let us consider now the reduction to small ${\mathcal C}$ volume of the system above.
This leaves behind a gauge theory on the leftover $T^2$ world volume whose spectrum can be computed by harmonic analysis. We denote by $g$ the genus of ${\mathcal C}$.

Let us first discuss the reduction on a regular Riemann surface and then the more general situation in which ${\mathcal C}$ is an orbicurve.

The complex scalars $I$ and $J$ get simple dimensional reduction and stay scalars in the bifundamental,
the gauge connection ${\mathcal A}$ on $S={\mathcal C}\times T^2$ leaves behind 
the gauge connection $A$ on $T^2$ and $g$ complex scalars in the adjoint, while other $g$ complex scalars in the adjoint arise from the reduction of the transverse field 
$\Phi_S$. These will be denoted as $B_3^{(i)}$ and $B_4^{(i)}$, where $i=1,\ldots, g$.

The other two complex scalar fields in the adjoint, namely $B_1$ and $B_2$, get simply dimensionally reduced.

This field content results in the quiver of figure \ref{fig:quiver_bulk}.

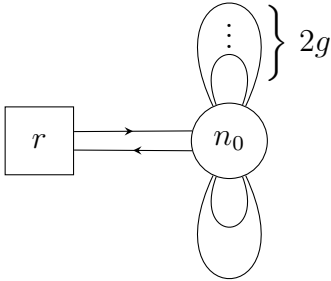
\begin{figure}[H]
\centering
\begin{tikzpicture}\vspace{-2mm}
\node[FrameNode](F0) at (10,0){$r$};
\node[GaugeNode](Gk0) at (2.5+10,0){$n_0$};
\node[](dots2) at (2.5+10,1.5){$\vdots$};
\node[] (g) at (3+10.4,1.3){$\bigg\}\ 2g$};
\draw[->-](F0.15) to (Gk0.165);
\draw[->-](Gk0.195) to (F0.345);
\draw[-](Gk0) to[out=65,in=115,looseness=12] (Gk0);
\draw[-](Gk0) to[out=70,in=110,looseness=7] (Gk0);
\draw[-](Gk0) to[out=65+180,in=115+180,looseness=12] (Gk0);
\draw[-](Gk0) to[out=70+180,in=110+180,looseness=7] (Gk0);
\end{tikzpicture}\vspace{-2mm}\caption{Quiver gauge theory arising from the compactification on $\mathcal C_{g,0}$.}\label{fig:quiver_bulk}
\end{figure}

The relations of this quiver can be read from the reduction of the F-term equations in the Appendix A \eqref{forse} and \eqref{si} by expanding in harmonic modes along the curve $\mathcal{C}$. More explicitely, the $\Phi_S$ field and the component
of the gauge connection $A_\mathcal{C}$ along $\mathcal{C}$  give rise to the $g$ hypers in the adjoint representation $(B_3^{(i)}, B_4^{(i)})$, where $i=1,\ldots, g$, obeying the BPS equations
\begin{eqnarray}\label{hypers}
&& [B_1,B_2] + IJ = 0\, , \, [B_3^i,B_4^j] = 0\\
&& [B_1,B_3^i] = 0 \, , \, [B_1,B_4^i] = 0 \, , \, [B_2,B_3^i] = 0 \, , \,  [B_2,B_4^i] =0 \nonumber \\
&& B_3^iI = 0 \, , \,  JB_3^i = 0 \, , \, B_4^iI = 0 \, , \,  JB_4^i = 0 \nonumber 
\end{eqnarray}
The above equations are equivalent to $g$ commuting copies of the ADHM equations for gauge theory with one adjoint hypermultiplet \cite{Bruzzo:2002xf}, as it can be shown by a simple squaring argument. 

In the general $\Omega$-background the supersymmetry of the D3-brane system reduced on $T^2$ is $(2,2)$ while the combined D3/D7-brane system reduced on $T^2$ has $(0,2)$ supersymmetry due to the presence of the chiral fields $I$ and $J$ and the above field content, augmented by the relevant fermions, form the corresponding  multiplets. 

Let us underline that this theory itself suffers of a $U(1)_R$-symmetry anomaly due to its chiral unbalanced field content. This can be immediately understood from the fact that the D3-branes profile produces an instanton background in the D7-brane gauge theory 
inducing chiral symmetry breaking. From the mathematical viewpoint it is known that the Donaldson-Thomas theory on fourfolds has positive virtual dimension which implies that one has to introduce observables matching the dimension counting.
We propose that the suitable set of observables is given by a compensating
sector of opposite charges -- given by $\bar I$, $\bar J$ and other fields associated to the $g$-hypers to be specified later -- which cancels the anomaly.
This sector may be interpreted as a background antiD7-brane system.

\subsection{Two dimensional GLSM of the defect: the nested instanton quiver}\label{sec:2d_GLSM}

When the curve ${\mathcal C}$ is extended to an orbicurve, at each orbifold point the gauge symmetry is reduced and further 2D degrees of freedom are present.
These correspond to the open strings stretching between the twisted D-branes and, from the gauge theory viewpoint, to the degrees of freedom defining the codimension two defect prescribed by the singular behaviour of the gauge curvature at the orbifold points. 

To obtain the effective low energy quiver description, we excise a disk around each puncture of ${\mathcal C}$ and discuss the 
the local behaviour of the D-branes system at the orbifold points
computing the associated low energy quiver gauge theory. 
We then glue back the disks to the bulk Riemann surface obtaining the full description of the gauge theory with defects reduced to two dimensions by the small ${\mathcal C}$-volume limit. This procedure is pictorially described in figure \ref{fig:disk_excision}.

\begin{figure}[H]
\centering
\includegraphics[page=1]{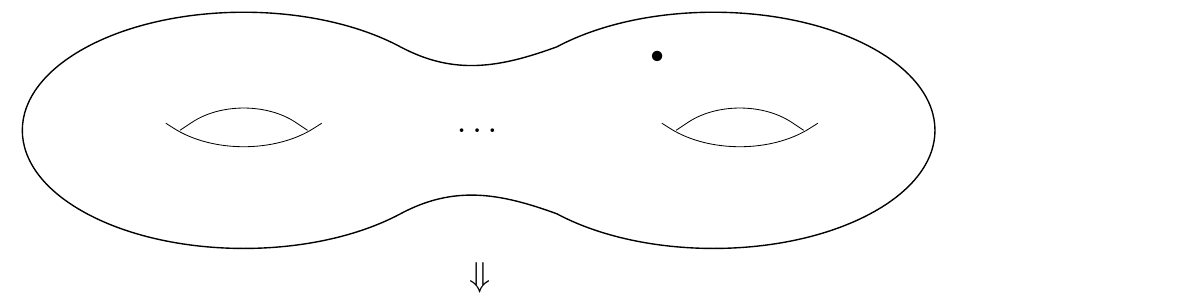}
\includegraphics[page=1]{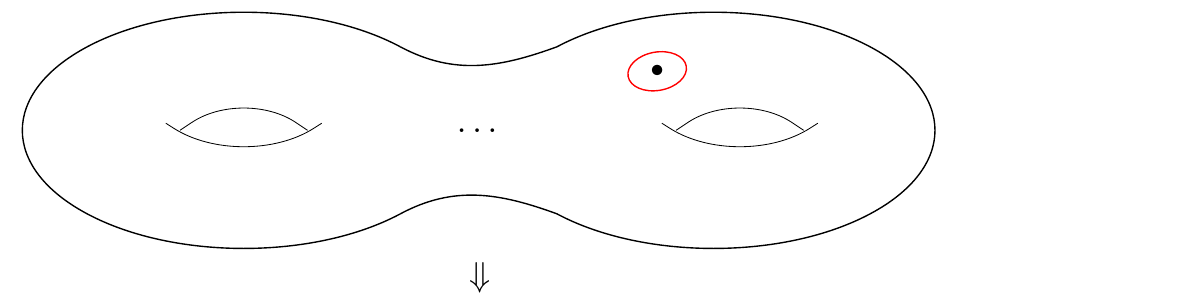}
\includegraphics[page=1]{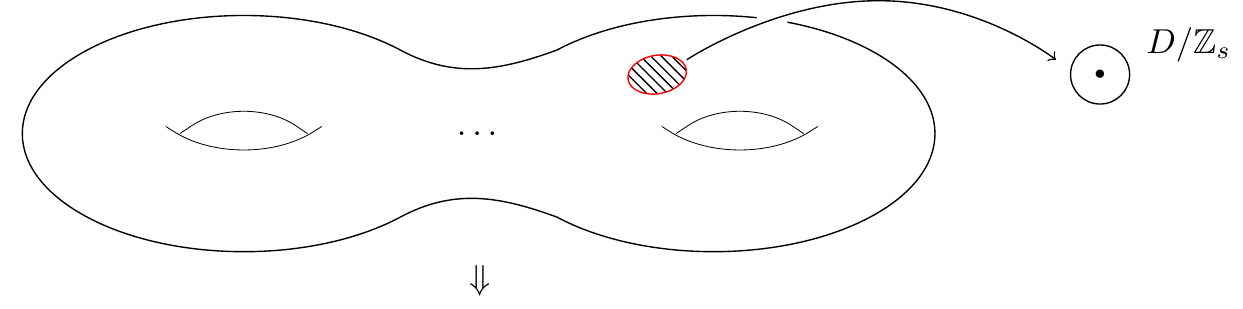}
\includegraphics[page=1]{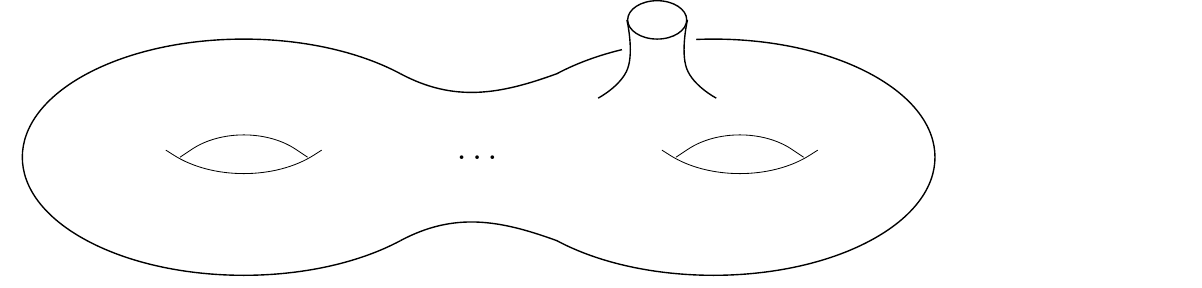}
\caption{Disk excision and gluing}\label{fig:disk_excision}
\end{figure}

The relevant open strings degrees of freedom can be inferred from the D-branes distribution as in the above figure \ref{fig:cake}. 
More precisely, see for example \cite{Fucito:2001ha,Tanzini:2002gp}, the Chan-Paton space of the D-brane system decomposes into irreducible representations $R_\ell$ of the local discrete group ${\mathbb Z}_s$ as
\begin{eqnarray}
{\mathbb V}&=&\sum_{\ell=0}^{s-1} {\mathbb V}_\ell\otimes R_\ell \label {deco3}\\
{\mathbb W}&=&\sum_{\ell=0}^{s-1} {\mathbb W}_\ell\otimes R_\ell \label {deco7}
\end{eqnarray}
where 
each of the D3 and D7 -brane charged sectors is denoted 
as 
\begin{equation}\label{vector37}
{\mathbb V}_\ell={\mathbb C}^{\gamma_\ell} \, , \, 
{\mathbb W}_\ell={\mathbb C}^{\beta_\ell}\, .
\end{equation}
As depicted in figure \ref{fig:cake}, the $\ell$-th Riemann sheet of the covering hosts a net number of $n_j\equiv\sum_{\ell=j}^{s-1}\gamma_{\ell}$ D3-branes and of 
$r_j\equiv\sum_{\ell=j}^{s-1}\beta_{\ell}$ D7-branes so that the open string degrees of freedom are represented as linear maps among the spaces
\begin{eqnarray}
{ V}_j&=&\sum_{\ell=j}^{s-1} {\mathbb V}_\ell \label {deco3'}\\
{W}_j&=&\sum_{\ell=j}^{s-1} {\mathbb W}_\ell \label {deco7'}
\end{eqnarray}

Let us now discuss the corresponding quiver gauge theory. 
This consists 
of a $(0,2)$ quiver gauge theory on $T^2$ with 
gauge group $\otimes_{j=0,\ldots, s-1} U(n_j)$, each node 
being coupled to two chiral multiplets in the adjoint 
$B_1^j,B_2^j\in {\rm End} V_j$
and each pair of successive nodes 
by a chiral in the bifundamental $F^j\in {\rm Hom}\left(V_j,V_{j+1}\right)$ for $j=0,\ldots, s-1$. 
The D3-D7 open strings modes are described by the linear maps $I^j\in{\rm Hom}\left(V_j,W_j\right)$ and $J^j\in{\rm Hom}\left(W_j,V_j\right)$.
Summarizing, the local D3-D7 system is effectively described by 
\begin{eqnarray}\label{campi37}
&&B_1^j,B_2^j\in {\rm End} V_j \, ,
F^j\in {\rm Hom}\left(V_j,V_{j+1}\right)\\
&&I^j\in{\rm Hom}\left(V_j,W_j\right)\, {\rm and} \,J^j\in{\rm Hom}\left(W_j,V_j\right)\nonumber
\end{eqnarray}
As is shown in the Appendix A these fields obey the relations
\begin{equation}\label{rels}
[B_1^j,B_2^j]+I^jJ^j=0\, , \quad B_1^jF^j-F^jB_1^{j+1}=0\, \quad B_2^jF^j-F^jB_2^{j+1}=0\, , \quad J^jF^j=0.
\end{equation}
Therefore, the resulting quiver describing the local D3-D7 system at the defect is 
given in figure \ref{fig:quiver_local}.

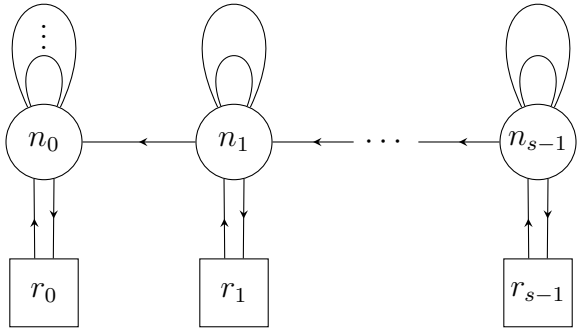
\begin{figure}[H]
\centering
\vspace{-5mm}
\begin{tikzpicture}
\node[FrameNode](F0) at (2.5+10,-2){$r_0$};
\node[FrameNode](F1) at (5+10,-2){$r_1$};
\node[FrameNode](Fn) at (9+10,-2){$r_{s-1}$};
\node[GaugeNode](Gk0) at (2.5+10,0){$n_0$};
\node[GaugeNode](Gk1) at (5+10,0){$n_1$};
\node[](label) at (7+10,0){$\cdots$};
\node[](dots2) at (2.5+10,1.5){$\vdots$};
\node[GaugeNode](Gkn) at (9+10,0){$n_{s-1}$};
\draw[->-](F0.15+90) to (Gk0.165+90);
\draw[->-](Gk0.195+90) to (F0.345+90);
\draw[->-](F1.15+90) to (Gk1.165+90);
\draw[->-](Gk1.195+90) to (F1.345+90);
\draw[->-](Fn.15+90) to (Gkn.165+90);
\draw[->-](Gkn.195+90) to (Fn.345+90);
\draw[->-](Gk1) to (Gk0);
\draw[->-](label) to (Gk1);
\draw[->-](Gkn) to (label);
\draw[-](Gk0) to[out=65,in=115,looseness=12] (Gk0);
\draw[-](Gk0) to[out=70,in=110,looseness=7] (Gk0);
\draw[-](Gk1) to[out=65,in=115,looseness=12] (Gk1);
\draw[-](Gk1) to[out=70,in=110,looseness=7] (Gk1);
\draw[-](Gkn) to[out=65,in=115,looseness=12] (Gkn);
\draw[-](Gkn) to[out=70,in=110,looseness=7] (Gkn);
\end{tikzpicture}\vspace{2mm}\caption{Quiver gauge theory arising from the compactification on $\mathcal C_{g,1}$.}\label{fig:quiver_local}
\end{figure}

The moduli space $\mathcal{N}_{r,\lambda,n,\mu}$ of its representations describes 
{\it nested instantons}. Indeed the $n$ D3 branes realise an $n$-instanton profile for the $U(r)$ D7 gauge fields, preserving the flag structure at the puncture.
The partitions $\lambda=(\lambda_1\ge\lambda_2\ge\ldots)$ of $r$ and $\mu=(\mu_1\ge\mu_2\ge\ldots)$ of $n$ describe  respectively the decomposition of the D7 and D3 Chan-Paton vector spaces into representations of the $\mathbb{Z}_s$ group.
More precisely, as shown in  fig. \ref{fig:quiver_local}, one gets the quiver of the flag manifold realised by the Chan Paton vector spaces of the D3 branes $V_{s-1}\subset V_{s-2}\subset\ldots\subset V_0$
with dimensions $n_j= n_0 - \sum_{l=1}^{j} \mu_l$ {\it framed} by the D7 branes vector spaces with dimensions $r_j=r_0 - \sum_{l=1}^{j} \lambda_l$.
The heights of the columns of each partition is obtained from an ordering of the data of the dimensions vector spaces $\beta_\ell$ and $\gamma_\ell$ of \eqref{vector37}. Indeed, these can be ordered
by using Weyl symmetry of D3 and D7 branes gauge groups such that $\gamma_0\ge\gamma_1\ge\ldots\ge\gamma_{s-1}$ and $\beta_0\ge\beta_1\ge\ldots\ge\beta_{s-1}$.

The moduli space of nested instantons has a natural projection to the standard ADHM instanton
moduli space $\mathcal{M}_{r,n}$
\begin{equation}
\pi : \mathcal{N}_{r,\lambda,n,\mu} \to \mathcal{M}_{r,n}
\end{equation}
 which is realised by setting all the open string twisted sectors to be empty, namely by setting to zero all the fields $F^j \, ,\, j=0,\ldots,s-1$ and $(I^j,J^j,B_1^j,B_2^j)$ for $ j=1,\ldots,s-1$.

\subsection{Relation to other quiver defect theories}

Some comments are in order regarding the quiver theory of the defect we obtain in our construction with respect to other quiver defect theories.
The quiver we study is derived from a Dp/Dp+4 system via an orbifold action which affects a {\it transverse} direction to both the brane types.
In this respect, it is different from the chain-saw quiver describing affine Laumon spaces \cite{Kanno:2011fw}, were the orbifold acts instead 
on the coordinates $B_1,B_2$ describing the motion of Dp branes {\it inside} the Dp+4. This induces a different quiver with a different set of relations.
A quiver which relates to the one in \cite{Kanno:2011fw}
can be obtained by considering a different specialization of the general geometric background 
for the D3/D7 system described in section 2.1. More precisely, 
one can consider $T^2\times X_6\times {\mathbb C}_{\epsilon_1}$, where
$X_6={\rm tot}\left[{\mathcal O}(p)\oplus{\mathcal O}(-p+2g-2)\right]_{{\mathcal C}_{g,k}}$ is the total space of a sum of two line bundles of the compensating degree on the orbicurve. 
In such a geometry we can consider the D3-branes along $T^2\times{\mathcal C}_{g,k}$ and the D7 say along $T^2\times Y_4\times {\mathbb C}_{\epsilon_1}$, where
$Y_4={\rm tot}\left[{\mathcal O}(p)\right]_{{\mathcal C}_{g,k}}$
and the fiber still hosts the torus action corresponding to the $\epsilon_2$-parameter of the Omega-background.
For $p>0$ in the vicinity of the orbifold points the geometry in the 
fiber direction is sensitive to the orbifold group. As a consequence, the 
corresponding modes in the open string sectors get twisted and the quiver changes by loosing an adjoint multiplet per node which gets a bifundamental,  as well as the flavoured fields $J_i$ will now point from the gauge node to the
nearby framing node. 
The resulting local quiver at the defect is then the chain-saw quiver. Correspondingly, the comet shaped quiver would in this case display tails given by chain-saw quivers.
This can be also obtained from a D1/D5 system with both D1 and D5 wrapping $\mathcal{C}_{g,k}$ via a double T-duality along transverse directions  to both. 

Since on the other hand both quivers are describing the parabolic reduction of the gauge connection on a surface defect, it is conceivable to expect that
a relation can be found between the associated partition functions  at least in some limit or suitable parametrisation. This could require non-trivial 
combinatorial identities on the partition functions themselves, similarly to what dicussed in \cite{Jeong:2018qpc} concerning the relation between orbifold and vortex-like defects.

Moreover, when decoupling the D7 branes by setting $I^j=0, J^j=0$, the description of the D3 branes at the defect lead to the quiver of fig. \ref{fig:quiver_d3local} which describes a flag manifold
with extra adjoint hypers at each node. 

\begin{figure}[H]
\centering
\vspace{-5mm}
\begin{tikzpicture}
\node[FrameNode](Gk0) at (2.5+10,0){$n_0$};
\node[GaugeNode](Gk1) at (5+10,0){$n_1$};
\node[](label) at (7+10,0){$\cdots$};
\node[GaugeNode](Gkn) at (9+10,0){$n_{s-1}$};
\draw[->-](Gk1) to (Gk0);
\draw[->-](label) to (Gk1);
\draw[->-](Gkn) to (label);
\draw[-](Gk1) to[out=65+180,in=115+180,looseness=8] (Gk1);
\draw[-](Gk1) to[out=65,in=115,looseness=8] (Gk1);
\draw[-](Gkn) to[out=65+180,in=115+180,looseness=8] (Gkn);
\draw[-](Gkn) to[out=65,in=115,looseness=8] (Gkn);
\end{tikzpicture}\caption{Quiver gauge theory for D3 branes at a single puncture.}\label{fig:quiver_d3local}
\end{figure}
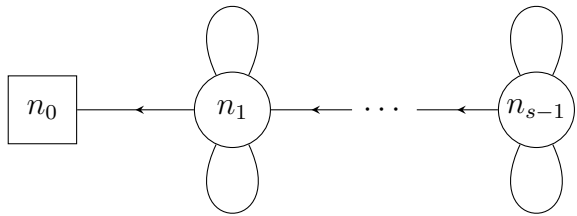

We notice that also the $TSU [N]$ quivers for defects studied  
\cite{Bruzzo:2010fk,Bonelli:2011fq,Bonelli:2011wx,Zenkevich:2017ylb,Cirafici:2018jor,Foda:2019klg} are based on flag manifold quivers but display a different field content. It should be possible to compare
the two kind of defect gauge theories in suitable limits by finding  an appropriate dictionary. 

\subsection{Nested Hilbert scheme of points}

The nested instanton moduli space is expected to reduce for a single D7-brane $r=r_0=1$ to the moduli space of the nested Hilbert scheme of points on $\mathbb{C}^2$, ${\rm Hilb}^{n,\mu}(\mathbb{C}^2)$. 
In this particular case the quiver described in the previous subsection reduces to the one of fig. \ref{fig:quiver_nhs} with relations
\begin{eqnarray}\label{rels-hilb}
&&[B_1^0,B_2^0]+I^0J^0=0\, , \quad [B_1^j,B_2^j] =0\  \, , \quad j\ge 1 \\
&& B_1^jF^j-F^jB_1^{j+1}=0\, \quad B_2^jF^j-F^jB_2^{j+1}=0\, , \quad J^0F^0=0.
\nonumber
\end{eqnarray}

\begin{figure}[H]
\centering\vspace{-5mm}
\begin{tikzpicture}
\node[FrameNode](F0) at (2.5+10,-2){$1$};
\node[GaugeNode](Gk0) at (2.5+10,0){$n_0$};
\node[GaugeNode](Gk1) at (5+10,0){$n_1$};
\node[](label) at (7+10,0){$\cdots$};
\node[GaugeNode](Gkn) at (9+10,0){$n_{s-1}$};
\draw[->-](F0.15+90) to (Gk0.165+90);
\draw[->-](Gk0.195+90) to (F0.345+90);
\draw[->-](Gk1) to (Gk0);
\draw[->-](label) to (Gk1);
\draw[->-](Gkn) to (label);
\draw[-](Gk0) to[out=65,in=115,looseness=9] (Gk0);
\draw[-](Gk0) to[out=70,in=110,looseness=7] (Gk0);
\draw[-](Gk1) to[out=65,in=115,looseness=9] (Gk1);
\draw[-](Gk1) to[out=65+180,in=115+180,looseness=9] (Gk1);
\draw[-](Gkn) to[out=65,in=115,looseness=9] (Gkn);
\draw[-](Gkn) to[out=65+180,in=115+180,looseness=9] (Gkn);
\end{tikzpicture}\caption{Nested Hilbert scheme quiver.}\label{fig:quiver_nhs}
\end{figure}
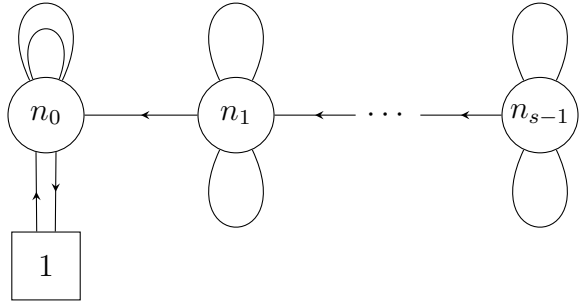

The moduli space of representations of this quiver is expected to provide an explicit description of ${\rm Hilb}^{n,\mu}(\mathbb{C}^2)$. This has been indeed proven for the particular case
of two-step nested Hilbert scheme $n_j=0$ for $j\ge2$ in \cite{flach_jardim}, were it is also shown that this variety is smooth for $n_1=1$. Indeed it is known that for $n_1>1$ the two-step nested Hilbert scheme
is singular. Moreover, nested Hilbert schemes with more than two steps are always singular, and {\it a fortiori} also the nested instanton moduli space. 
The D3/D7 partition functions we will evaluate via localisation will then compute {\it virtual} invariants
of these moduli spaces, since a perfect obstruction theory for them is expected to exist.

\subsection{Comet shaped quiver}\label{sec:comet}
Finally, the description of the D3/D7 system on the full geometry gives rise to the  {\it comet shaped quiver} of fig. \ref{fig:quiver_d3global}.

\begin{figure}[H]
\centering
\vspace{-5mm}
\begin{tikzpicture}
\node[FrameNode](F) at (2.5+10,-2){$r_0$};
\node[GaugeNode](Gk0) at (2.5+10,0){$n_0$};
\node[GaugeNode](Gk1) at (5+10,2){$n_1^{(1)}$};
\node[FrameNode](F1) at (6.5+10,1){$r_1^{(1)}$};
\node[GaugeNode](Gk12) at (5+10,-2){$n_1^{(k)}$};
\node[FrameNode](F12) at (6.5+10,-3){$r_1^{(k)}$};
\node[](label) at (7+10,2){$\cdots$};
\node[](label2) at (7+10,-2){$\cdots$};
\node[](dots1) at (11,0){$\cdots$};
\node[](vdots1) at (5+10,0){$\vdots$};
\node[](vdots2) at (9+10,0){$\vdots$};
\node[GaugeNode](Gkn) at (9+10,2){$n_{s-1}^{(1)}$};
\node[GaugeNode](Gkn2) at (9+10,-2){$n_{s-1}^{(k)}$};
\node[FrameNode](Fn) at (10.5+10,1){$r_{s-1}^{(1)}$};
\node[FrameNode](Fn2) at (10.5+10,-3){$r_{s-1}^{(k)}$};
\draw[->-](F.15+90) to (Gk0.165+90);
\draw[->-](Gk0.195+90) to (F.345+90);
\draw[->-](F1.125) to (Gk1.-20);
\draw[->-](Gk1.-40) to (F1.145);
\draw[->-](F12.125) to (Gk12.-20);
\draw[->-](Gk12.-40) to (F12.145);
\draw[->-](Gk1) to (Gk0);
\draw[->-](Gk12) to (Gk0);
\draw[->-](label) to (Gk1);
\draw[->-](label2) to (Gk12);
\draw[->-](Gkn) to (label);
\draw[->-](Gkn2) to (label2);
\draw[->-](Fn.125) to (Gkn.-20);
\draw[->-](Gkn.-40) to (Fn.145);
\draw[->-](Fn2.125) to (Gkn2.-20);
\draw[->-](Gkn2.-40) to (Fn2.145);
\draw[-](Gk0) to[out=60+90,in=120+90,looseness=12] (Gk0);
\draw[-](Gk0) to[out=70+90,in=110+90,looseness=7] (Gk0);
\draw[-](Gk1) to[out=65+180,in=115+180,looseness=8] (Gk1);
\draw[-](Gk1) to[out=65,in=115,looseness=8] (Gk1);
\draw[-](Gkn) to[out=65+180,in=115+180,looseness=8] (Gkn);
\draw[-](Gkn) to[out=65,in=115,looseness=8] (Gkn);
\draw[-](Gk12) to[out=65+180,in=115+180,looseness=8] (Gk12);
\draw[-](Gk12) to[out=65,in=115,looseness=8] (Gk12);
\draw[-](Gkn2) to[out=65+180,in=115+180,looseness=8] (Gkn2);
\draw[-](Gkn2) to[out=65,in=115,looseness=8] (Gkn2);
\end{tikzpicture}\caption{The comet-shaped quiver.}\label{fig:quiver_d3global}
\end{figure}
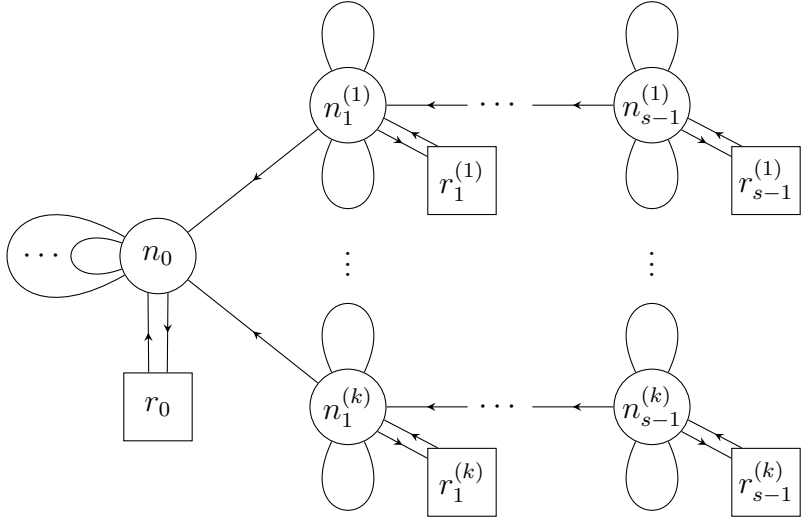

This is obtained by gluing the nested instanton moduli quivers describing the decompositions of the branes at the defects to the bulk quiver of figure \ref{fig:quiver_bulk}.
The number of tails in the comet quiver is equal to the number of punctures of the Riemann surface, while their length is related to the flag structure
due to the parabolic reduction of the connection at each puncture. All in all, the effective theory describing the D3-D7 system on $T^2$
reduces to a GLSM with target space the total space of the bundle 
\begin{equation}
\mathcal{V}_g \equiv \pi^*\left( \left(T^*\mathcal{M}_{r,n}\right)^{\oplus g} \otimes \left( \det \mathcal{T}\right)^{1-g}\right) 
\end{equation}
over the moduli space of nested instantons $\mathcal{N}_{r,\underline{\lambda},n,\underline{\mu}}$,
where the collection of partitions $\underline{\lambda} =(\lambda^1,\ldots,\lambda^k)$ and $\underline{\mu} =(\mu^1,\ldots,\mu^k)$ describe the decomposition of D7 and D3 branes
respectively under the cyclic groups $\mathbb{Z}_{s_i} \, , \, i=1,\ldots k$ acting at the punctures.
The physical interpretation of the above bundle is the following: the first factor is simply the contribution of the $g$ hypermultplets in the adjoint representation of the bulk theory described in subsection \ref{sec:2d_bulk}.
Regarding the second factor,
let us remark that the couplings of the D3/D7 brane system turns on a background line bundle describing the determinant bundle of
the Dirac zero modes in the instanton background. This is given by the determinant of the tautological bundle $\mathcal{T}$ over $\mathcal{M}_{r,n}$.
The power $(1-g)$ is due to the multiplicity of fermionic zero-modes on the Riemann surface $\mathcal{C}$. 
In the limit of degeneration of the $T^2$ to a circle this leads to a Chern-Simons
interaction term for the resulting D2/D6 system.  This term is essential in the comparison with results on character varieties and will be discussed
in detail in subsection \ref{sec:QM_reduction}, while in the next one we will briefly recall some basic definitions about character varieties that will be useful for the subsequent discussion.

%
%
%
%

\subsection{Character varieties}

Given a Riemann surface $\mathcal{C}$ of genus $g$ with  $k$ punctures $D=\sum_{i=1}^k p_i$, one defines the $GL_n(\mathbb{C})$ {\it character variety} as 
the moduli space of representations of the first fundamental group of  $\mathcal{C}\backslash D$  into $GL_n(\mathbb{C})$
\begin{equation}\label{char}
\mathcal{G}_{\sigma} = \left\{ \rho \in {\rm Hom} \left (\pi_1\left (\mathcal{C}\backslash D\right), GL_n(\mathbb{C})\right) | \rho (\gamma_i) \in C_i \right\} / / PGL_n(\mathbb{C})
\end{equation}
where $C_1,\ldots,C_k\subset GL_n(\mathbb{C})$ are semisimple conjugacy classes of type $\sigma^1,\ldots,\sigma^k$, namely the parts of the partition $\sigma^i$, 
$(\sigma^i_1\ge\sigma^i_2\ge\ldots)$, describe the multiplicities of the eigenvalues of any matrix in the  conjugacy class $C_i$.

When non empty \eqref{char} is a smooth projective variety of dimension 
$$
d_\sigma = n^2 (2g-2+k) - \sum_{i,j} \left(\sigma_{\, j}^i\right)^2 + 2
$$

To describe the cohomology of \eqref{char}, LHRV introduced the {\it k- puntures, genus $g$ Cauchy function}
\begin{equation}\label{cauchy}
\Omega (z,w) =\sum_{\sigma\in \mathcal{P}} \mathcal{H}_\sigma(z,w) \prod_{i=1}^k \tilde H^i_\sigma ({\bf x}_i; z^2,w^2) 
\end{equation}
where $\mathcal{P}$ is the set of partitions, $ \tilde H^i_\sigma ({\bf x}_i; z^2,w^2)$ are refined Macdonald polynomials  and 
\begin{equation}\label{mhook}
\mathcal H_\sigma(z,w)=\prod_{s\in\sigma}\frac{(z^{2a(s)+1}-w^{2l(s)+1})^{2g}}{(z^{2a(s)+2-w^{2l(s)}})(z^{2a(s)}-w^{2l(s)+2})}.
\end{equation}
where $a(s),l(s)$ are respectively the arm and leg length of the $s$ box of the Young diagram $\sigma$ representing the partition. 
Eq.\eqref{cauchy} turns out to be the generating function of the cohomology polynomials of $GL_n(\mathbb{C})$ character varieties, summed over $n$.

Let us now outline the connection with the brane construction described in the previous subsections. The dynamics of D3 branes on the local surface $S$ is refined Vafa-Witten theory.
When $S=T^2\times \mathcal{C}$, this reduces in the limit of small area of  $\mathcal{C}$ to a gauged linear sigma model  from $T^2$ to Hitchin's moduli space on $\mathcal{C}$ \cite{Bershadsky:1995vm}. 
On the other hand, in \cite{Simpson} it was proved that \eqref{char} is homeomorphic to the moduli space of Higgs bundles with parabolic reduction on the divisor $D=\sum_{i=1}^k p_i$. 
In presence of D7 branes, the non-perturbative effects on their dynamics are obtained by summing over the D3-branes partition functions. One then naturally obtains a generating function of the elliptic cohomology
of $GL_n(\mathbb{C})$ character varieties. Summarising the $T^2$ partition function of the D3-D7  comet shaped quiver reads
\begin{equation}\label{216}
Z_{T^2}=\sum_{n}\sum_{\underline{\mu}\in\mathcal P(n)^k}(\bold q^{\underline{\mu}})^{\bold r}
{\rm Ell}^{\textrm{vir}}\left(\mathcal N_{r,\underline{\lambda},n,\underline{\mu}}, \mathcal{V}_g
\right),
\end{equation}
with $(\bold q^{\underline{\mu}})^{\bold r}= \prod_{i=1}^k \prod_{\alpha=0}^{s_i-1} \left(q_{i,\, \alpha}^{|\mu^i_{\alpha}|}\right)^{r^i_{\alpha}}$ and 
\begin{equation}
{\rm Ell}^{\textrm{vir}}\left(\mathcal N_{r,\underline{\lambda},n,\underline{\mu}},\mathcal{V}_g
\right)=
{\it Ell}
(T^{\textrm{vir}}\mathcal N){\rm ch}\left(\mathcal{V}_g
\right)\cap\left[\mathcal N_{r,\underline{\lambda},n,\underline{\mu}}\right]^{\textrm{vir}}.
\end{equation}
For a single D7 brane $\bold r =r_0 =1$, the above formulae can be understood  as an elliptic virtual generalisation of the generating function introduced by LHRV. Indeed, we will show in the following that in the limit of degeneration of $T^2$ to a circle, one obtains LHRV formulae, or more precisely a virtual refinement of them.
%
%
%
%
%
%

\subsection{Reduction to quantum mechanics, Chern-Simons term and LHRV formulae}\label{sec:QM_reduction}

In this subsection we summarise the reduction of the D3/D7 system on $T^2$ to a quantum mechanical system in a T-dual picture. More precisely,
if the two torus factorises as $T^2=S^1\times S^1$ and one of the two circles is taken to be very small, our D-brane system can be T-dualised 
along the small circle and reduced to a corresponding D2/D6 system on ${\mathcal C}\times S^1$. This corresponds to the quantum mechanics of the comet shaped quiver
with a Chern-Simons coupling, given  by a phase factor $\eu^{i m\int CS(A,F)}=\eu^{\iu m\int\de x^\mu \mathcal A_\mu}$
so that the particle is coupled to an external vector potential. Let us briefly recall how this works in the standard ADHM case \cite{Tachikawa:2004ur} in order to then generalise it to the nested instanton moduli space.
The partition function is the equivariant index
\begin{equation}
Z_{S^1}=\sum_n q^{nr}\Ind\left(\mathcal M_{r,n},\mathcal L^{\otimes m}\right),
\end{equation}
where $\mathcal L$ is the determinant line bundle $\mathcal L=\Det\slashed D$, whose fiber on the space of connections $\mathcal A/\mathcal G$ is $(\det\ker\slashed D_A)^*\otimes(\det\ker\slashed D_A^\dagger)$. By making use the ADHM construction for the moduli space of ASD connections, the $n-$dimensional vector space $V_0$ is actually the space of fermionic zero-modes. In order to compute the Chern-Simons level, we make use of 
the Atiyah-Singer index theorem for a vector bundle $E\to M$
\begin{equation}
\Ind(M,E)=\Ind(\slashed D)=\int_M\hat A(TM)\wedge\ch E,
\end{equation}
which gives the index of the Dirac operator twisted by $E$, \textit{ie} $D:S\otimes E\to S\otimes E$, $S$ being the spin bundle over $M$. To compute the CS level  in the case at hand one has to consider  the geometric background $S^1\times T^*\mathcal C\times\mathbb C^2\times\mathbb R$.
Because of the twisting of the supersymmetric theory along  $\mathcal C$, the $\slashed D$ operator along $\mathcal C$ reduces to the  $\overline\pd$ operator and the roof genus $\hat A(TM)$ to the Todd class. Thus, when we take the effective theory obtained by shrinking the size of $\mathcal C$, $\Ind(\overline\pd)_{\mathcal C}$ gives the muliplicity of the fermionic zero modes, according to the decomposition $\Psi^{(0)}=\psi^{(0)}_{\mathcal C}\otimes \psi^{(0)}_{\mathbb C^2}$. The index theorem along $\mathcal C$ reads
\begin{equation}
\Ind(\overline\pd)_{\mathcal C}=\int_{\mathcal C}\Td(T\mathcal C)=1-g,
\end{equation}
which determines the level of the Chern-Simons interaction to be  $m=1-g$. 
Finally, the partition function is given by the following equivariant (virtual) index
\begin{equation}\label{S1_partitionfunction}
Z_{ S^1}=\sum_{n}\sum_{\mu\in\mathcal P(n)}(\bold q^{\mu})^r\Ind\left(\mathcal N_{r,\lambda,n,\mu},\Det(\slashed D)^{\otimes(1-g)}\right),
\end{equation}
where we use the notation $\bold q^\mu=q_0^{n_0(\mu)}\cdots q_{s-1}^{n_{s-1}(\mu)}$ and
\begin{equation}
\Ind\left(\mathcal N_{r,\lambda,n,\mu},\Det(\slashed D)^{\otimes(1-g)}\right)=\hat A(T^{\textrm{vir}}\mathcal N)\ch\left(\Det(\slashed D)^{\otimes(1-g)}\right)\cap\left[\mathcal N_{r,\lambda,n,\mu}\right]^{\textrm{vir}}.
\end{equation}
In the quiver representation of the nested instanton moduli space, the  $\slashed D$ operator on $\mathbb{C}^2$ appearing in the above equation is given by the pull-back of the tautological bundle $\mathcal{T}$ on the ADHM moduli space $\mathcal M_{r,n}$,
so that its determinant line bundle coincides with the one of $\mathcal {T}$, which will be used in the equivariant localisation formulae.

In the following section \ref{sec:partition_functions} we will show that the above partition function, when computed for the particular case of the nested Hilbert scheme of points on $\mathbb{C}^2$, gives a virtual generalization of LHRV formulae and reduces precisely to them when the 
nested Hilbert scheme is smooth.
Let us remark that the quantum mechanical system of the nested Hilbert scheme of points and its relation with LHRV formulae has been studied in \cite{Chuang:2013wpa} via a different approach based on topological string amplitudes on orbifold Calabi-Yaus. 
\section{Partition functions}\label{sec:partition_functions}

In this section we proceed to the evaluation of the partition function of the effective quiver gauge theories of the D3/D7-system discussed in the previous section in the limit of 
small volume of the wrapped curve ${\mathcal C}$.
This is performed by making use of supersymmetric localisation which is a
version of equivariant localisation formulae \cite{Pestun:2016zxk} for super-manifolds
which allows a generalisation to
supersymmetric path integrals in quantum field theories. 
The only configurations contributing to the latter
are the fixed loci of the supersymmetry transformations.
When these are isolated points, the path integral reduces to a sum over them
weighted by one-loop super-determinants of the tangent bundle $T$ at those points, 
that is
\begin{equation}\label{loc}
\sum_{x\in{\{FP\}}}
\frac{e^{-S(x)}}{{\rm Sdet} T_x}
\end{equation}
where $\{FP\}$ is the set of fixed points, $S(x)$ is the value of the action at $x\in\{FP\}$
and $T_x=T\vert_x$ is the restriction of $T$ at $x$.

In the following we will implement the above computational scheme by calculating 
the above data for the relevant quiver gauge theories on $T^2$. 
We will first focus on the contribution of a single defect on the sphere encoding the parabolic reduction
of the connection at a given point, which is described by a single legged quiver.
Then, we will consider the case of higher genus Riemann surface and combine all the contributions in the comet-shaped quiver theory partition function.

\subsection{Contribution of a single surface defect on the sphere}
\subsubsection{Field content and superpotential}

\begin{figure}[H]
\centering\vspace{-5mm}
\begin{tikzpicture}
\node[FrameNode](F0) at (2.5+10,-2){$r$};
\node[GaugeNode](Gk0) at (2.5+10,0){$n_0$};
\node[GaugeNode](Gk1) at (5+10,0){$n_1$};
\node[](label) at (7+10,0){$\cdots$};
\node[GaugeNode](Gkn) at (9+10,0){$n_{s-1}$};
\draw[->-](F0.15+90) to (Gk0.165+90);
\draw[->-](Gk0.195+90) to (F0.345+90);
\draw[->-](Gk1) to (Gk0);
\draw[->-](label) to (Gk1);
\draw[->-](Gkn) to (label);
\draw[-](Gk0) to[out=65,in=115,looseness=12] (Gk0);
\draw[-](Gk0) to[out=70,in=110,looseness=7] (Gk0);
\draw[-](Gk1) to[out=65,in=115,looseness=9] (Gk1);
\draw[-](Gk1) to[out=65+180,in=115+180,looseness=9] (Gk1);
\draw[-](Gkn) to[out=65,in=115,looseness=9] (Gkn);
\draw[-](Gkn) to[out=65+180,in=115+180,looseness=9] (Gkn);
\end{tikzpicture}\caption{Low energy GLSM quiver in the case of $g=0$, $k=1$.}\label{fig:quiver_sphere}
\end{figure}
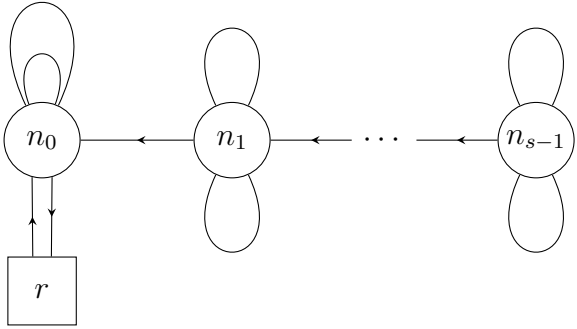

The matter content of the GLSM we are interested in is the one summarized in table \ref{table:ADHM_field_cont}, where $G=U(n_0)\times U(n_1)\times\cdots\times U(n_{s-1})$ and $\Box_i$ denotes the Young diagram corresponding to the fundamental representation of $U(n_i)$.
\begin{table}[H]
\centering
\begin{tabular}{|r|c|c|c|c|}\hline\hline
 & gauge $G$ & flavour $U(1)\times U(1)^2$ & twisted mass & $R-$charge\\ \hline\hline
$B_1^i$ & $\overline{\Box}_i\otimes\Box_i$ & $\mathbbm 1_{(1,0)}$ & $-\epsilon_1$ & $q$ \\
$B_2^i$ & $\overline{\Box}_i\otimes\Box_i$ & $\mathbbm 1_{(0,1)}$ & $-\epsilon_2$ & $q$ \\
$I$ & $\Box_0$ & $\overline\Box_{(0,0)}$ & $-a$ & $q+p$ \\
$J$ & $\overline{\Box}_0$ & $\Box_{(1,1)}$ & $a-\epsilon$ & $q-p$ \\
$F^i$ & $\overline{\Box}_i\otimes \Box_{i-1}$ & $\mathbbm 1_{(0,0)}$ & $0$ & $0$ \\
$\chi^i$ & $\overline{\Box}_i\otimes\Box_i$ & $\mathbbm 1_{(-1,-1)}$ & $\epsilon$ & $-2q$ \\
$\chi_i^{B_1}$ & $\Box_i\otimes\overline{\Box}_{i-1}$ & $\mathbbm 1_{(-1,0)}$ & $\epsilon_1$ & $-q$\\
$\chi_i^{B_2}$ & $\Box_i\otimes\overline{\Box}_{i-1}$ & $\mathbbm 1_{(0,-1)}$ & $\epsilon_2$ & $-q$\\
$\chi_{JF}$ & $\Box_1$ & $\Box_{(1,1)}$ & $\epsilon-a$ & $p-q$ \\ \hline\hline
\end{tabular}\caption{Field content for quiver \ref{fig:quiver_sphere}}\label{table:ADHM_field_cont}
\end{table}
The relations satisfied by the quiver GLSM are enforced by the superpotential $\mathcal W$ in \eqref{W_0}.
\begin{equation}\label{W_0}
\begin{split}
\mathcal W=\Tr_{0}\left[\chi_0([B_1^0,B_2^0]+IJ)\right]+\sum_{i=1}^N\Tr_{i}\left[\chi_i[B_1^i,B_2^i]+\chi_i^{B_1}(B_1^{i-1}F^i-F^iB_1^i)+\right.\\
\left.+\chi_i^{B_2}(B_2^{i-1}F^i-F^iB_2^i)\right]+\chi_{JF}JF^1.
\end{split}
\end{equation}
Let us notice that, as we already pointed out, the locus cut out by $\mathcal W$ through the $D-$term equations is overdetermined. Thus we still have to introduce $s-1$ additional chiral fields $Q_i,\ i=1,\dots,{s-1}$ taking care of the relations over the constraints. These additional fields will transform in the $\overline{\Box}_i\otimes\Box_{i-1}$ representation of $U(n_i)\times U(n_{i-1})$. We will assign them $R-$charge $2q$ and they will be charged under the $U(1)^2$ flavour symmetry with charge $(1,1)$. The relations over the constraints induced by these chirals is
\begin{equation}\label{relations_over_relations_1+}
\begin{split}
0=&[B_1^{i-1},B_2^{i-1}]F^i+B_2^{i-1}(B_1^{i-1}F^i-F^iB_1^i)-(B_1^{i-1}F^i-F^iB_1^i)B_2^i+\\
&+(B_2^{i-1}F^i-F^iB_2^i)B_1^i-B_1^{i-1}(B_2^{i-1}F^i-F^iB_2^i)-F^i[B_1^i,B_2^i],
\end{split}
\end{equation}
when $i>1$, while
\begin{equation}\label{relations_over_relations_1}
\begin{split}
0=&([B_1^{0},B_2^{0}]+IJ)F^1+B_2^{0}(B_1^{0}F^1-F^1B_1^1)-(B_1^{0}F^1-F^1B_1^1)B_2^1+\\
&+(B_2^{0}F^1-F^1B_2^1)B_1^1-B_1^{0}(B_2^{0}F^1-F^1B_2^1)-I(JF^1)-F^1[B_1^1,B_2^1]\, 
\end{split}
\end{equation}
covers the remaining case $i=1$.

The chiral supersymmetry transformations of the above fields are
\begin{eqnarray}
QI=\mu_I \, ,\quad Q\mu_I=D_AI+ \phi^0 I-Ia\\
QJ=\mu_J \, ,\quad Q\mu_J=D_AJ - J\phi^0 + a J-\epsilon J\\
QB^i_l=M^i_l \, , \quad QM^i_l = D_AB^i_l + [\phi^i,B^i_l] -\epsilon_lB^i_l\\
Q\psi_F^i=F^i\, , \quad QF^i=D_A\psi_F^i-\phi^i\psi_F^i+\psi_F^i\phi^{i+1}\\
Q\chi_i=h_i\, , \quad Qh_i=D_A\chi_i +[\phi^i, h_i] +\epsilon h_i\\
Q\chi_{JF}=h_{JF}\, , \quad Qh_{JF}=D_A\chi_{JF} + [\phi^0,\chi_{JF}] + (\epsilon-a)\chi_{JF}\\
Q\chi_i^{B_l}=h_i^{B_l}\, , \quad Qh_i^{B_l}= D_A\chi_i^{B_l}+\phi^i\chi_i^{B_l}-\chi_i^{B_l}\phi^{i+1}+\epsilon_l\chi_i^{B_l}\\
Q\chi_{Q_i}=h_{Q_i}\, , \quad Qh_{Q_i}= D_A\chi_{Q_i}+\phi^{i-1}\chi_{Q_i}-\chi_{Q_i}\phi^i
+\epsilon\chi_{Q_i}\\
Q\bar A=\eta \, , \quad Q\eta=F_A\, , \quad QA=0
\end{eqnarray}
where $(A,\bar A)$ is the connection on $T^2$ in holomorphic coordinates and $F_A$ its curvature two-form, $\epsilon_l$, $l=1,2$ are the equivariant weights of the $U(1)^2$ rotation group
acting on $\mathbb{C}^2$ and $\epsilon=\epsilon_1+\epsilon_2$. Moreover $\phi^i, \, i=0,\ldots,s-1$ are the zero modes of the $A$-connection implementing global $U(n_i)$ gauge transformations
of the  $i^{th}$-node.
The fixed points of the above supersymmetry transformation 
impose that the connection $(A,\bar A)$ is flat. Then by a standard squaring argument one can show that the other fields must be constant so that the supersymmetry fixed locus reduces to the fixed locus 
of the $U(1)^{(r+2)}$-torus action on the nested instanton moduli space, where $U(1)^r$ is the Cartan torus of the $U(r)$ gauge group with equivariant parameters $a_b, \, b=1,\ldots,r$.

\subsubsection{Anomaly and observables}

As we already discussed at the end of subsection \ref{sec:2d_bulk}, the $(0,2)$ D3/D7-branes theory displays a $U(1)_R$ anomaly whose compensation can be obtained via the insertion of suitable observables. To this end we introduce a sector of additional degrees of freedom $\bar I$ and $\bar J$ with opposite gauge global charges w.r.t. $I$ and $J$ which, once integrated out, produces the insertion of the observables. These will be properly taken into account in the following computations.

\subsubsection{Fixed points}\label{sec:fixed}

The characterization of the fixed locus of the torus action on the moduli space of nested instantons $\mathcal N(r,n_0,\dots,n_{s-1})\simeq\mathcal N_{r,[r^1],n,\mu(n)}$ is most easily understood by describing it as the moduli space of (suitably defined) stable representations of the quiver in figure \ref{fig:quiver_sphere}. In this setting we associate to the quiver \ref{fig:quiver_sphere} the vector spaces $W$ and $V_i$, in addition to the space
$$
\mathbb X=\End(V_0)^{\oplus 2}\oplus\Hom(V_0,W)\oplus\Hom(W,V_0)\oplus\left[\bigoplus_{i=1}^{s-1}\left(\End(V_i)^{\oplus 2}\oplus\Hom(V_i,V_i-1)\right)\right]
$$
of the morphisms of the quiver corresponding to the matter fields $B_{1,2}^i$, $F^i$, $I$ and $J$. In this language, the quiver in figure \ref{fig:quiver_sphere} would be represented graphically as the one in figure \ref{fig:quiver_sphere_math}.
\begin{figure}[H]
\centering
\begin{tikzcd}
V_{s-1} \arrow[out=70,in=110,loop,swap,"B_1^{s-1}"] \arrow[out=250,in=290,loop,swap,"B_2^{s-1}"] \arrow[r,"F^{s-1}"] & \cdots \arrow[r,"F^2"] & V_1 \arrow[out=70,in=110,loop,swap,"B_1^1"] \arrow[out=250,in=290,loop,swap,"B_2^1"] \arrow[r,"F^1"] & V_0 \arrow[out=70,in=110,loop,swap,"B_1^0"] \arrow[out=250,in=290,loop,swap,"B_2^0"] \arrow[r,shift left=.5ex,"J"] & W \arrow[l,shift left=.5ex,"I"]
\end{tikzcd}\caption{General representation of quiver \ref{fig:quiver_sphere}.}\label{fig:quiver_sphere_math}
\end{figure}
On $\mathbb X$ we have a natural action of $\mathcal G=GL(V_0)\times\cdots\times GL(V_{s-1})$, which preserves the subscheme of those points satisfying the relations \eqref{rels-hilb}. Then, given a framed representation $(W,V_0,\dots,V_{s-1},X)$, $X\in\mathbb X_0$ of the quiver \ref{fig:quiver_sphere_math}, one can prove that there is a suitable definition of stability such that, in a particular chamber of the parameters at play, semi-stability is equivalent to stability (also as a GIT quotient), so that it makes sense to talk about the moduli space of stable framed representations of the quiver \ref{fig:quiver_sphere_math} without any further specification. This space will be denoted by $\mathcal N(r,n_0,\dots,n_{s-1}):=\mathbb X_0\git_\chi\mathcal G$, for some suitable choice of an algebraic character $\chi$ of $\mathcal G$.

By means of this construction one can show that there is a sum decomposition $V_0=V_i\oplus\tilde V_i$ and $V_i=V_{i+1}\oplus\hat V_{i+1}$, such that $\tilde V_i=\hat V_i\oplus\tilde V_{i-1}$. This splitting also induces the following block matrix decomposition of the morphisms $B_{1,2}^0$, $I$ and $J$ in \eqref{split_morph},
\begin{equation}\label{split_morph}
B_1^0=\begin{pmatrix}
B_1^i & B_1^{'i}\\
0 & \tilde B_1^i
\end{pmatrix},\qquad
B_2^0=\begin{pmatrix}
B_2^i & B_2^{'i}\\
0 & \tilde B_2^i
\end{pmatrix},\qquad
I=\begin{pmatrix}
I^{'i}\\ \tilde I^i
\end{pmatrix},\qquad
J=\begin{pmatrix}
0 & \tilde J^i
\end{pmatrix}.
\end{equation}
such that $(W,\tilde V_i,\tilde B_1^i,\tilde B_2^i,\tilde I^i,\tilde J^i)$ is a stable ADHM datum.

Once an equivariant action of a torus $T\curvearrowright\mathcal N(r,n_0,\dots,n_{s-1})$ is introduced in the natural way suggested by the SUSY construction of the quiver \eqref{fig:quiver_sphere}, the previous observations makes it possible to characterize the $T-$fixed locus of $\mathcal N(r,n_0,\dots,n_{s-1})$ in terms of those of some moduli spaces of stable ADHM data. This is all summarized in the following proposition 
\begin{proposition}\label{prop:fixed_locus}
The $T-$fixed locus of $\mathcal N(r,n_0,\dots,n_{s-1})$ is described by $s-$tuples of nested coloured partitions $\boldsymbol\mu_1\subseteq\cdots\subseteq\boldsymbol\mu_{s-1}\subseteq\boldsymbol\mu_0$, with $|\boldsymbol\mu_0|=n_0$ and $|\boldsymbol\mu_{i>0}|=n_0-n_i$.
\end{proposition}

\begin{example}
As an example, consider the moduli space $\mathcal N(2,3,2,1)$. Its fixed point locus will be described by the following couples of nested partitions,
\begin{displaymath}
\left[\mathcal N(2,3,2,1)\right]^T\longleftrightarrow\left\{
\begin{aligned}
&\left(1^1,2^1,3^1;\emptyset\right),\left(1^1,2^1,2^1;1^1\right),\left(1^1,2^1;1^1,1^1,1^1\right),\left(1^1,1^1,2^1;1^1,1^1\right),\\
&\left(1^1;1^1,2^1,2^1\right),\left(1^1,1^1,1^1;1^1,2^1\right),\left(1^1,1^1;1^1,1^1,2^1\right),\left(\emptyset,1^1,2^1,3^1\right),\\
&\left(1^1,2^1,2^11^1;\emptyset\right),\left(\emptyset;1^1,2^1,2^11^1\right),\left(1^1,1^2,2^11^1;\emptyset\right),\left(\emptyset;1^1,1^2,2^11^1\right),\\
&\left(1^1,1^2,1^3;\emptyset\right),\left(1^1,1^2,1^2;1^1\right),\left(1^1,1^2;1^1,1^1\right),\left(1^1,1^1,1^2;1^1,1^1\right),\\
&\left(1^1;1^1,1^2,1^2\right),\left(1^1,1^1,1^1;1^1,1^2\right),\left(1^1,1^1;1^1,1^1,1^2\right),\left(\emptyset;1^11^21^3\right),
\end{aligned}\right.
\end{displaymath}
where each term on the r.h.s. has to be interpreted as a couple of nested partitions, \eg
\begin{displaymath}
\left(1^2,2^11^1,3^12^1;\emptyset\right)\longleftrightarrow\vcenter{\hbox{\ShadedTableauR[(1,0),(1,-1),(2,-1)]{(1,0),(1,-1)}{{\ ,\ },{\ ,\ ,\ }}}}\oplus\emptyset\longleftrightarrow\left(\ShadedTableau[]{\ ,\ }\hookrightarrow\ShadedTableau[]{{\ },{\ ,\ }}\hookrightarrow\ShadedTableau[]{{\ ,\ },{\ ,\ ,\ }}\right)\oplus\emptyset.
\end{displaymath}

The notation we use for a partition $\mu\in\mathcal P$ is descriptive of its corresponding Young diagram in the following sense: $[\mu_1^{i_1}\cdots\mu_j^{i_j}\cdots]$ denotes the partition
\begin{displaymath}
\mathcal P\ni[\mu_1^{i_1}\cdots\mu_j^{i_j}\cdots]=(\underbrace{\mu_1,\dots,\mu_1}_{i_1},\dots,\underbrace{\mu_j,\dots,\mu_j}_{i_j},\dots),
\end{displaymath}
or, in other words, $i_j$ counts the number of rows of length $\mu_j$ stacked one over the other.
\end{example}

\subsubsection{Character computation}

The super determinant weighting the contribution of each fixed point can be computed 
from the character decomposition of the torus action on the (virtual) tangent space:
\begin{equation}\label{tan_rep_long}
\begin{split}
T_Z^{\textrm{vir}}\mathcal N(r,n_0,\ldots,n_{s-1})&=\End(V_0)\otimes(Q-1-\Lambda^2Q)+\Hom(W,V_0)+\Hom(V_0,W)\otimes\Lambda^2Q\\
&\quad-\Hom(V_1,W)\otimes\Lambda^2Q+\\
&\quad+\sum_{\ell=1}^{s-1}\left[\End(V_{\ell})-\Hom(V_{\ell},V_{\ell-1})\right]\otimes(Q-1-\Lambda^2Q)
\end{split}
\end{equation}
where the first line accounts for the standard ADHM quiver $(B^0_1,B^0_2, I,J)$ and their constraints, 
the second line for the constraint $JF^1=0$ and
the third line for the maps in the tail, their constraints and the relations among them.

By decomposing the vector spaces $V_i$ in terms of characters of the torus action $T\curvearrowright\mathcal N_{r,[r^1],n,\mu}$ we can then study the character decomposition of the virtual tangent space to the moduli space of nested instantons
and obtain
\begin{equation}\label{tan_rep_long_char}
\begin{split}
T_Z^{\textrm{vir}}\mathcal N_{r,[r^1],n,\mu}=&T_{\tilde Z}\mathcal M_{r,n_0}+\sum_{a,b=1}^r\sum_{i=1}^{M_0^{(a)}}\sum_{j=1}^{N_0^{(b)}}R_bR_a^{-1}\left(T_1^{i-\mu_{1,j}^{(b)}}-T_1^{i}\right)\left(T_2^{-j+\mu_{1,i}^{(a)'}+1}+\right.\\ & \left.-T_2^{-j+\mu_{0,i}^{(a)'}+1}\right)-\sum_{i=1}^{M_0^{(a)}}\sum_{j=1}^{\mu_{0,i}^{(a)'}-\mu_{1,i}^{(a)'}}T_1^iT_2^{j+\mu_{1,i}^{(a)'}}+\\
&+\sum_{k=2}^{s-1}\left[\sum_{a,b=1}^r\sum_{i=1}^{M_0^{(a)}}\sum_{j=1}^{N_0^{(b)}}R_bR_a^{-1}\left(T_1^{i-\mu_{k,j}^{(b)}}-T_1^{i-\mu_{k-1,j}^{(b)}}\right)\right.\\
&\left.\left(T_2^{-j+\mu_{k,i}^{(a)'}+1}-T_2^{-j+\mu_{0,i}^{(a)'}+1}\right)\right]+(s-1)(T_1T_2),
\end{split}
\end{equation}
where the fixed point $Z$ is to be identified with a choice of a sequence of coloured nested partitions $\boldsymbol\mu_1\subseteq\boldsymbol\mu_{N-1}\subseteq\cdots\subseteq\boldsymbol\mu_{s-1}\subseteq\boldsymbol\mu_0$, as in proposition \ref{prop:fixed_locus}, $\tilde Z\leftrightarrow\boldsymbol\mu_0$ and the last term, namely $(s-1)(T_1T_2)$, has been added in order to take into account the over-counting in the relations $[B_1^{i},B_2^{i}]=0$ due to the commutator being automatically traceless.

\subsubsection{Determinants}\label{sec:determinants}
Having the character decomposition of the virtual tangent space to the moduli space of nested instantons enables us to easily compute the $2d$ $\mathcal N=(0,2)$ partition functions of the low energy GLSM of subsection \ref{sec:2d_GLSM} in terms of the eigenvalues of the torus action, which we will do in the particular case of $r=1$ for the sake of simplicity. The partition function we want to compute on the sphere $\mathcal C_0=S^2$ with $1$ marked point will take the form
\begin{equation}\label{ZT2-0}
\mathcal Z^{\rm ell}_1(S^2;q_0,\dots,q_{s-1})=\sum_{\mu_1\subseteq\cdots\subseteq\mu_0}q_0^{|\mu_0|}q_1^{|\mu_0\setminus\mu_1|}\cdots q_{s-1}^{|\mu_0\setminus\mu_{s-1}|}Z_{(\mu_0,\mu_1,\dots,\mu_{s-1})}^{\rm ell},
\end{equation}
with $|\mu_i\setminus\mu_j|=|\mu_i|-|\mu_j|$ denoting the number of boxes in the skew Young diagram $Y_{\mu_i\setminus\mu_j}$, while $Z^{\rm ell}_{(\mu_0,\dots,\mu_{s-1})}$ is the contribution at a fixed instanton profile.

In particular, once we fix an instanton configuration by choosing a sequence of nested partitions $\mu_1\subseteq\cdots\subseteq\mu_{s-1}\subseteq\mu_0$ we can write the torus partition function as
\begin{equation}
Z_{(\mu_0,\mu_1,\dots,\mu_{s-1})}^{\rm ell}=\mathcal L_{\mu_0}^{\rm ell}\mathcal N_{\mu_0}^{\rm ell}\overline{\mathcal N}_{\mu_0}^{\rm ell}\mathcal T_{\mu_0,\mu_1}^{\rm ell}\overline{\mathcal T}_{\mu_0,\mu_1}^{\rm ell}\mathcal W_{\mu_0,\dots,\mu_{s-1}}^{\rm ell},
\end{equation}
where
\begin{align}
\mathcal L^{\rm ell}_{\mu_0}&=\prod_{s\in Y_{\mu_0}}\exp\left[-\vol(T^2)\left(\phi(s)-\xi\right)\right], \label{CS-T2}\\
\mathcal N^{\rm ell}_{\mu_0}&=\prod_{s\in Y_{\mu_0}}\frac{1}{\theta_1(\tau|E(s))\theta_1(\tau|E(s)-\epsilon)},\\
\overline{\mathcal N}_{\mu_0}^{\rm ell}&=\prod_{s\in Y_{\mu_0\setminus\Box}}\theta_1(\tau|\phi(s)-\tilde a)\theta_1(\tau|\phi(s)-\tilde a+\epsilon),\\
\mathcal T^{\rm ell}_{\mu_0,\mu_1}&=\prod_{i=1}^{M_0}\prod_{j=1}^{\mu_{0,i}-\mu_{1,i}}\theta_1(\tau|\epsilon_1i+\epsilon_2(j+\mu_{1,i}')),\\
\overline{\mathcal T}^{\rm ell}_{\mu_0,\mu_1}&=\prod_{s\in Y_{\mu_0\setminus\mu_1}}\frac{1}{\theta_1(\tau|\phi(s)-\tilde a+\epsilon)},\\ \nonumber
\mathcal W_{\mu_0,\dots,\mu_{s-1}}&=\prod_{k=0}^{s-2}\left[\prod_{i=1}^{M_0}\prod_{j=1}^{N_0}\frac{\theta_1(\tau|\epsilon_1(i+\mu_{k,j})+\epsilon_2(\mu_{k+1,i}'-j+1))}{\theta_1(\tau|\epsilon_1(i-\mu_{k+1,j})+\epsilon_2(\mu_{k+1,i}'-j+1))}\right.\\
&\left.\qquad\quad\prod_{i=1}^{M_0}\prod_{j=1}^{N_0}\frac{\theta_1(\tau|\epsilon_1(i-\mu_{k+1,j})+\epsilon_2(\mu_{0,i}'-j+1))}{\theta_1(\tau|\epsilon_1(i-\mu_{k,j})+\epsilon_2(\mu_{0,i}'-j+1))}\right],\label{WT2}
\end{align}
and for any box $s$ in a Young diagram $Y_\mu$ we defined $\phi(s)$ to be the quantity \eqref{phi_s}
\begin{equation}\label{phi_s}
\phi(s)=a+(i-1)\epsilon_1+(j-1)\epsilon_2,
\end{equation}
and
\begin{equation}\label{E_s}
E(s)=-\epsilon_1a(s)+\epsilon_2(l(s)+1),
\end{equation}
with $a(s)$ and $l(s)$ being respectively the arm and leg length of $s$ in $Y_\mu$.

Notice that $\mathcal N_{\mu_0}^{\rm ell}$ is the elliptic analogue of the Nekrasov partition function, while $\overline{\mathcal N}_{\mu_0}^{\rm ell}$ is its $\overline{\mbox{ADHM}}$ analogue due to the $\overline{\rm D7}$ coupling. Moreover the contributions from the functions $\mathcal T^{\rm ell}_{\mu_0,\mu_1}$, $\overline{\mathcal T}_{\mu_0,\mu_1}^{\rm ell}$ and $\mathcal W_{\mu_0,\dots,\mu_{s-1}}$ altogether encodes the contribution of the surface defect insertion. Finally, $\mathcal L^{\rm ell}_{\mu_0}$ encodes the CS-like term we discussed in section \ref{sec:comet} and \ref{sec:QM_reduction}. This is interpreted as a CS-term contribution when the limit to QM is taken, and a $5d$ partition function on $\mathbb R^4\times S^1$ is retrieved. In any case, it comes from the coupling to a background connection on the determinant line bundle $\Det\slashed D$ encoding fermionic zero modes. This background connection is mirrored by the presence of $\xi$ in \eqref{CS-T2}, which is intended to be later specialized to $\xi\to a$.

Because of the previous observations it is instructive to perform the summation over all the sequences of $s$ nested partitions in two steps. First we sum over all the smaller partitions $\mu_1\subseteq\cdots\subseteq\mu_{s-1}\subseteq\mu_0$ at fixed $\mu_0\in\mathcal P$. It will prove useful for what we will do later to define the rational function $P^{\rm ell}_{\mu_0}$ as in \eqref{Pell}.
\begin{equation}\label{Pell}
P^{\rm ell}_{\mu_0}=\sum_{\mu_1\subseteq\cdots\subseteq\mu_{s-1}}\mathcal T_{\mu_0,\mu_1}^{\rm ell}\overline{\mathcal T}_{\mu_0,\mu_1}^{\rm ell}\mathcal W_{\mu_0,\dots,\mu_{s-1}}^{\rm ell}q_1^{|\mu_0\setminus\mu_1|}\cdots q_{s-1}^{|\mu_0\setminus\mu_{s-1}|}.
\end{equation}

Finally, by summing also over the $\mu_0$ partitions we can rewrite the full partition function as in \eqref{ZT2-0-full},
\begin{equation}\label{ZT2-0-full}
\mathcal Z^{\rm ell}_1(S^2;q_0,\dots,q_{s-1})=\sum_{\mu_0}q_0^{|\mu_0|}\mathcal Y^{\rm ell}_{\mu_0}P^{\rm ell}_{\mu_0},
\end{equation}
where we defined
\begin{equation}\label{Yell}
\mathcal Y^{\rm ell}_{\mu_0}=\mathcal L_{\mu_0}^{\rm ell}\mathcal N_{\mu_0}^{\rm ell}\overline{\mathcal N}_{\mu_0}^{\rm ell},
\end{equation}
and 
$P^{\rm ell}_{\mu_0}$ are particular elliptic functions which can be regarded as an elliptic virtual uplift of modified Macdonald polynomials. The first few examples are listed in
\eqref{duepall},\eqref{trepall},\eqref{quattropall}.
As a useful remark, we want to point out that by taking the limit $q_{i>0}\to 0$ we can effectively switch off the tail of the quiver, since $P^{\rm ell}_{\mu_0}\xrightarrow{q_{i>0}\to 0}1$, and we recover the partition function on the sphere with one puncture of trivial holonomy, $\mathcal Z_0^{\rm ell}(S^2;q_0)$.

\subsection{An alternative derivation: contour integral formulae}
\label{appendix:integral}
In this section we will be explicitly computing the partition functions of the low energy theory coming from the D3/D7 system described in subsection \ref{sec:2d_GLSM} by reducing the supersymmetric path integral to a contour integral 
via supersymmetric localization \cite{Benini:2013nda,Benini:2013xpa}.
The model we are interested in gives rise to a $2d$ $\mathcal N=(0,2)$ GLSM on $T^2$. The mechanism of supersymmetry breaking from the maximal amount to $\mathcal N=(0,2)$ in the reduction to the low energy theory leaves us with a matter content comprised of chiral fields corresponding to the morphisms in the representation theory of quiver \ref{fig:quiver_nhs} in the category of vector spaces, and Fermi fields implementing the Lagrange multipliers in the superpotential. Let us first study the partition function for the quiver GLSM of figure \ref{fig:quiver_sphere}, having fixed the numerical type of the quiver to $(1,n_0,\dots,n_{s-1})$. In this case the localization formula is given by 
\begin{equation}\label{loc_ZT2_(0,2)}
Z_{T^2}=\frac{1}{(2\pi i)^N} \oint_C Z_{T^2,1-\textrm{loop}}(\tau,z,\mathsf x)
\end{equation}
where C is a real $N$-dimensional cycle in the moduli space of flat connections on $T^2$,
$\mathsf x$ denotes the collection of the coordinates we are integrating over and
\begin{equation}
\begin{split}
\hat Z_{T^2,1-\textrm{loop}}(\tau,z,\mathsf x)=\tilde{\mathcal Z}\left(\prod_{i\neq j}^{n_0}\frac{\theta_1(\tau|u_{ij}^0)\theta_1(\tau|u_{ij}^0-zq+\epsilon)}{\theta_1(\tau|u_{ij}^0+zq/2-\epsilon_1)\theta_1(\tau|u_{ij}^0+zq/2-\epsilon_2)}\right.\\
\left.\prod_{i=1}^{n_0}\frac{1}{\theta_1(\tau|u_i^0+z(q+p)/2-a)\theta_1(\tau|u_i^0-z(q-p)/2-a+\epsilon)}\right)\\
\prod_{k=1}^{s-1}\left(\prod_{i\neq j}^{n_k}\frac{\theta_1(\tau|u_{ij}^k)\theta_1(\tau|u_{ij}^k-zq+\epsilon)}{\theta_1(\tau|u_{ij}^k+zq/2-\epsilon_1)\theta_1(\tau|u_{ij}^k+zq/2-\epsilon_2)}\right.\\
\left.\prod_{i=1}^{n_k}\prod_{j=1}^{n_{k-1}}\frac{\theta_1(\tau|u_i^k-u_j^{k-1}-zq/2+\epsilon_1)\theta_1(\tau|u_i^k-u_j^{k-1}-zq/2+\epsilon_2)}{\theta_1(\tau|u_{j}^{k-1}-u_i^{k})\theta_1(\tau|u_{j}^{k-1}-u_i^{k}+zq-\epsilon)}\right)\\
\prod_{i=1}^{n_1}\theta_1(\tau|u_i^1+z(p-q)/2-a+\epsilon),
\end{split}
\end{equation}
with
\begin{equation}
\begin{split}
\tilde{\mathcal Z}=\prod_{i=0}^{s-1}\left[\frac{1}{n_i!}\left(\frac{2\pi\eta^2(\tau)\theta_1(\tau|-zq+\epsilon)}{\theta_1(\tau|zq/2-\epsilon_1)\theta_1(\tau|zq/2-\epsilon_2)}\right)\right]\frac{(\eta^2(\tau))^{n_0}}{(\iu\eta(\tau))^{n_1}}.
\end{split}
\end{equation}
As was already pointed out in subsection \ref{sec:2d_bulk}, the coupling of the D3-branes to the D7-branes makes the theory anomalous. This chiral anomaly is encoded in the contributions dependent on the fields coupled to the framing, namely $I$ and $J$, which break a chiral half of the original $\mathcal N=(2,2)$ supersymmetry. From the point of view of the localization formula this is most easily made manifest by studying the transformation properties of the integrand under shifts along the generators of the torus. Let us the recall that the Jacobi $\theta_1(\tau|z)$ function is defined in terms of the exponentiated modular parameter $q=\eu^{2\pi\iu\tau}$, $\Im\tau\ge 0$, and $y=\eu^{2\pi\iu z}$ as
$$ \theta_1(\tau|z)=q^{1/8}y^{-1/2}(q,q)_\infty\theta(\tau|z), $$
where $\theta(\tau|z)=(y,p)_\infty (py^{-1},p)_\infty$ and $(a,q)_\infty=\prod_{k=0}^\infty(1-aq^k)$ is the $q-$Pochhammer symbol. By this definition it is easy to see that the Jacobi function $\theta_1(\tau|z)$ is odd in $z$, \ie $\theta_1(\tau|-z)=-\theta_1(\tau|z)$, and that it is quasi-periodic  under shifts $z\to z+a+b\tau$, $a,b\in\mathbb Z$:
$$ \theta_1(\tau|z+a+b\tau)=(-1)^{a+b}\eu^{-2\pi\iu bz}\eu^{-\iu\pi b^2\tau}\theta_1(\tau|z),\qquad\forall a,b\in\mathbb Z. $$
The anomaly then comes from the fact the integrand is unbalanced in terms of the theta functions, exactly due to the presence of $I$ and $J$. The part of the $1-$loop determinant coming from adjoint and bifundamental fields does not contribute to the gauge anomaly, as it comes from an $\mathcal N=(2,2)$ multiplet. As we already explained in subsection \ref{sec:2d_bulk}, we take care of this anomaly by introducing extra Fermi fields $\overline I$ and $\overline J$, which we think can be interpreted as accounting for interactions with $\overline{\rm D7}-$branes. In this way we get that the $T^2$ partition function is corrected by the presence of the $\overline{\rm D7}$ as
\begin{equation}\label{integrand_sphere}
\begin{split}
\hat Z_{T^2,1-\textrm{loop}}^{D3/D7/\overline{D7}}(\tau,z,\mathsf x)=\hat{\mathcal Z}\left(\prod_{i\neq j}^{n_0}\frac{\theta_1(\tau|u_{ij}^0)\theta_1(\tau|u_{ij}^0-zq+\epsilon)}{\theta_1(\tau|u_{ij}^0+zq/2-\epsilon_1)\theta_1(\tau|u_{ij}^0+zq/2-\epsilon_2)}\right.\\
\left.\prod_{i=1}^{n_0}\frac{\theta_1(\tau|u_i^0+zR_{\overline I}/2-\overline a)\theta_1(\tau|u_i^0+zR_{\overline J}/2-\overline a+\epsilon)}{\theta_1(\tau|u_i^0+z(q+p)/2-a)\theta_1(\tau|u_i^0-z(q-p)/2-a+\epsilon)}\right)\\
\prod_{k=1}^{s-1}\left(\prod_{i\neq j}^{n_k}\frac{\theta_1(\tau|u_{ij}^k)\theta_1(\tau|u_{ij}^k-zq+\epsilon)}{\theta_1(\tau|u_{ij}^k+zq/2-\epsilon_1)\theta_1(\tau|u_{ij}^k+zq/2-\epsilon_2)}\right.\\
\left.\prod_{i=1}^{n_k}\prod_{j=1}^{n_{k-1}}\frac{\theta_1(\tau|u_i^k-u_j^{k-1}-zq/2+\epsilon_1)\theta_1(\tau|u_i^k-u_j^{k-1}-zq/2+\epsilon_2)}{\theta_1(\tau|u_{j}^{k-1}-u_i^{k})\theta_1(\tau|u_{j}^{k-1}-u_i^{k}+zq-\epsilon)}\right)\\
\prod_{i=1}^{n_1}\frac{\theta_1(\tau|u_i^1+z(p-q)/2-a+\epsilon)}{\theta_1(\tau|u_i^1+zR_{\overline{J}F}/2-\overline a+\epsilon)},
\end{split}
\end{equation}
with
\begin{equation}
\begin{split}
\hat{\mathcal Z}=(-1)^{n_1}\prod_{i=0}^{s-1}\left[\frac{1}{n_i!}\left(\frac{2\pi\eta^2(\tau)\theta_1(\tau|-zq+\epsilon)}{\theta_1(\tau|zq/2-\epsilon_1)\theta_1(\tau|zq/2-\epsilon_2)}\right)^{n_i}\right].
\end{split}
\end{equation}
Two observations are due here.
\begin{enumerate}
\item An appropriate choice of the $R-$charges of $\overline I$ and $\overline J$ makes it possible to overcome completely the anomaly issue in the integration variables and in the $U(1)_R$ fugacity. However, asking for the double periodicity of the integrand forces us also to impose a constraint on the twisted masses $a$ and $\overline a$, namely $a-\overline a\in\mathbb Z$. This condition is responsible for the fact that introducing the extra fields needed to cure the anomaly doesn't change the fixed point structure of the localization computation. The procedure we adopted has one additional beneficial side-effect. In fact, even though the theory involving the $\overline{\rm D7}$ branes is different from the one we started with, however it is still an interesting quantity, as it should compute a generating functions for insertions of observables, as it was proposed in the D8/$\overline{\rm D8}$ case by Nekrasov in \cite{Nekrasov:2017cih}.
\item As for the second remark, it is interesting to study the QM limit ($\tau\to\iu\infty$) of the partition function at hand. In fact when we shrink one $S^1$ in $T^2$ to a point, we can decouple the contribution of the $\overline{D7}$ branes by taking very large values of $\overline a$ and by then rescaling the relevant gauge coupling. By doing this we recover the $5d$ partition function one can independently compute on $\mathbb R^4\times S^1$, apart from an overall normalization factor. This will give us the equivariant Euler number of the nested Hilbert scheme of points on $\mathbb C^2$, possibly twisted by a power of the determinant line bundle of the Dirac $\slashed D$ operator.
\end{enumerate}

Now, in order to explicitly compute the partition function we need to remember that the Jacobi $\theta_1(\tau|z)$ function does not have any pole, however it has simple zeros on the lattice $z\in\mathbb Z+\tau\mathbb Z$. Moreover it is simple to verify that $\theta(\tau|z)^{-1}$ has residue in $z=\alpha+\beta\tau$ given by the following formula
\begin{equation}
\frac{1}{2\pi\iu}\oint_{z=\alpha+\beta\tau}\frac{1}{\theta_1(\tau|z)}=\frac{(-1)^{\alpha+\beta}\eu^{\iu\beta^2\tau}}{2\pi\eta^3(\tau)}.
\end{equation}
In general a careful analysis of singularities would be needed in order to understand which poles are giving a non-vanishing contribution to the computation of the partition function on $T^2$.  In our particular case the poles contributing to the residue computation will be classified in terms of nested partitions $\mu_1\subseteq\cdots\subseteq\mu_{s-1}\subseteq\mu_0$. 
In principle this result could be obtained via the systematic approach of Jeffrey-Kirwan. Here we follow an alternative procedure
by giving a suitable imaginary part to the twisted masses (for example through the $R-$charges via a redefinition of the relevant parameters, as in \cite{Bonelli:2013rja}) and by closing the integration contour in the lower-half plane. In this particular setting we take care of redefining $a,\,\epsilon_i$ in such a way that $\Im a,\,\Im\epsilon_i<0$ and $\Im a>\Im\epsilon$. By the requirement on the Cartan parameters of the D7-branes, namely $a-\overline a\in\mathbb Z$, we also have $\Im\overline a=\Im a<0$. It is sufficient to study the pole structure of the first two integrations (namely $\{u_j^0\}$ and $\{u_j^1\}$) in \eqref{integrand_sphere}, whose poles and zeros are schematically shown in table \ref{table:integrand_sphere_poles}.
\begin{table}[H]
\centering
\begin{tabular}{|c|c|}\hline\hline
Poles & Zeros \\ \hline\hline
$u_j^0=a$ & $u_j^0=\overline a$\\
$u_j^0=a-\epsilon$ & $u_j^0=\overline a-\epsilon$\\
$u_{ij}^0=\epsilon_1$ & $u_{ij}^0=0$\\
$u_{ij}^0=\epsilon_2$ & $u_{ij}^0=-\epsilon$\\
\hline\hline
$u_j^1=u_i^0$ & $u_j^1=u_i^0-\epsilon_1$\\
$u_j^1=u_i^0-\epsilon$ & $u_j^1=u_i^0-\epsilon_2$\\
$u_j^1=\overline a-\epsilon$ & $u_j^1=a-\epsilon$\\
$u_{ij}^1=\epsilon_1$ & $u_{ij}^1=0$\\
$u_{ij}^1=\epsilon_2$ & $u_{ij}^1=-\epsilon$\\
\hline\hline
\end{tabular}\caption{Poles and zeros for $\{u_j^0\}$ and $\{u_j^1\}$ in \eqref{integrand_sphere}.}\label{table:integrand_sphere_poles}
\end{table}
The integration over the $\{u_j^0\}$ is standard, as it is has the same pole structure of the standard Nekrasov partition function \cite{Bonelli:2013rja,Nekrasov:2002qd}, and the poles contributing to the residue computation will be described by partitions $\mu_0$. Each box in $\mu_0$ will then encode the position of a pole for the first $n_0$ integrations. As for the integrations over the $\{u_j^1\}$ variables we first point out that the 1-loop determinant due to the $\overline{\rm D7}-$brane, as $a-\overline a\in\mathbb Z$ and the corresponding pole falls out of the integration contour. In the same way also poles of the 1-loop determinant of $Q_i$ give a vanishing contributions, because of one out of two different reasons: either the singularity falls out of the integration contour or its contribution is annihilated by a zero coming from the determinants of $\chi_1^{B_i}$. Any pole that might fall outside the Young diagram associated to $\mu_0$ must also be excluded from the computations, because of the flag structure of the quiver in figure \ref{fig:quiver_sphere}. These considerations leads us to the classification of poles of the $\{u_j^1\}$ integrations in terms of partitions as follows: by choosing the order of the integration to be $u_1^1,u_2^1,\dots,u_{n_1}^1$ poles are chosen by successively picking outer corners of $Y_{\mu_0}$ so that the complement in $Y_{\mu_0}$ is still a Young diagram corresponding to a partition $\mu_1$, with $|\mu_1|=n_0-n_1$. The procedure we just described is depicted in figure \ref{fig:poles_picking}.

\begin{figure}[H]
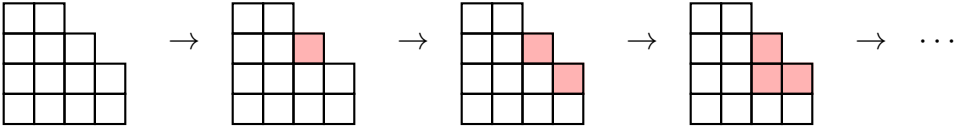

\centering
\ShadedTableauRed{{\ ,\ },{\ ,\ ,\ },{\ ,\ ,\ ,\ },{\ ,\ ,\ ,\ }}
\quad$\to$\quad
\ShadedTableauRed[(3,-1)]{{\ ,\ },{\ ,\ ,\ },{\ ,\ ,\ ,\ },{\ ,\ ,\ ,\ }}
\quad$\to$\quad
\ShadedTableauRed[(3,-1),(4,-2)]{{\ ,\ },{\ ,\ ,\ },{\ ,\ ,\ ,\ },{\ ,\ ,\ ,\ }}
\quad$\to$\quad
\ShadedTableauRed[(3,-1),(4,-2),(3,-2)]{{\ ,\ },{\ ,\ ,\ },{\ ,\ ,\ ,\ },{\ ,\ ,\ ,\ }}
\quad$\to$\quad$\cdots$
\caption{Procedure for picking poles of $\{u_j^1\}$ from $Y_{\mu_0}$.}\label{fig:poles_picking}
\end{figure}
Any successive integration is done in the same way, and the poles contributing to the integration are classified by sequences of nested partitions, as we discussed in \ref{sec:fixed}.

Boxes in the skew partitions $\mu_0\setminus\mu_j$ will denote positions for poles in the $j-$th integration, according to the following rule: a box of $Y_{\mu_0\setminus\mu_k}$ located at position $(i,j)$ inside $Y_{\mu_0}$ (this is required by the nesting phenomenon) corresponds to the coordinate $u_l^{(k)}=a+(i-1)\epsilon_1+(j-1)\epsilon_2$. One thing to be pointed out is that the assignment of a certain Young diagram configuration do in fact specify a particular pole only up to Weyl permutations of the coordinates: because of this we choose a particular ordering of the coordinates and neglect the counting factor $(n_0!\cdots n_{s-1}!)^{-1}$ in $\hat{\mathcal Z}$.

The partition function $Z_{T^2}^{D3/D7/\overline{D7}}$ will then take the following form
\begin{equation}
\begin{split}
Z_{T^2}^{D3/D7/\overline{D7}}=\hat{\mathcal Z}_{\textrm{res}}\sum_{\mu_1\subseteq\cdots\mu_{s-1}\subseteq\mu_0}\left(Z_{\mu_0}(\epsilon_1,\epsilon_2,\overline a)Z^{JF}_{\mu_1,\mu_0}(\epsilon_1,\epsilon_2,\overline a)\prod_{i=0}^{s-2}Z_{\mu_{i+1},\mu_i}(\epsilon_1,\epsilon_2,\overline a)\right),
\end{split}
\end{equation}
with
\begin{align}
\hat{\mathcal Z}_{\textrm{res}}&=(-1)^{n_1}\prod_{i=0}^{s-1}\left[\frac{1}{n_i!}\left(\frac{\theta_1(\tau|-zq+\epsilon)}{\theta_1(\tau|zq/2-\epsilon_1)\theta_1(\tau|zq/2-\epsilon_2)\eta(\tau)}\right)^{n_i}\right]\\
Z_{\mu_0}(\epsilon_1,\epsilon_2,\overline a)&=\prod_{s\in\mu_0\setminus\Box}\frac{\theta_1(\tau|\phi(s)-\overline a)\theta_1(\tau|\phi(s)-\overline a+\epsilon)}{\theta_1(\tau|\phi(s))\theta_1(\tau|\phi(s)+\epsilon)}\prod_{\substack{s\neq s'\\ s,s'\in\mu_0}}\left(\frac{\theta_1(\tau|\phi(s)-\phi(s'))}{\theta_1(\tau|\phi(s)-\phi(s')-\epsilon_1)}\right.\nonumber\\
&\qquad\qquad\left.\frac{\theta_1(\tau|\phi(s)-\phi(s')+\epsilon)}{\theta_1(\tau|\phi(s)-\phi(s')-\epsilon_2)}\right)\\
Z_{\mu_{k+1},\mu_k}(\epsilon_1,\epsilon_2,\overline{a})&=\prod_{\substack{s\neq s'\\ s,s'\in\mu_0\setminus\mu_{k+1}}}\frac{\theta_1(\tau|\phi(s)-\phi(s'))\theta_1(\tau|\phi(s)-\phi(s')+\epsilon)}{\theta_1(\tau|\phi(s)-\phi(s')-\epsilon_1)\theta_1(\tau|\phi(s)-\phi(s')-\epsilon_2)}\cdot\nonumber\\
&\quad\ 	\prod_{\substack{s\in\mu_0\setminus\mu_{k+1}\\ s'\in\mu_0\setminus\mu_k}}\frac{\theta_1(\tau|\phi(s)-\phi(s')+\epsilon_1)\theta_1(\tau|\phi(s)-\phi(s')+\epsilon_2)}{\theta_1(\tau|\phi(s')-\phi(s))\theta_1(\tau|\phi(s')-\phi(s)-\epsilon)}\\
Z^{JF}_{\mu_1,\mu_0}(\epsilon_1,\epsilon_2,\overline{a})&=\prod_{\substack{s\in\mu_0\setminus\mu_1\\ s'\in\mu_0}}\frac{\theta_1(\tau|\phi(s)+\epsilon)}{\theta_1(\tau|\phi(s)-\overline a+\epsilon)}
\end{align}
These formulae are to be compared with the contribution of a quiver with fixed numerical type to $\mathcal Z_1^{\rm ell}(S^2;q_0,\dots,q_{s-1})$, in particular the contribution at each fixed point will be the same as $Z_{(\mu_0,\dots,\mu_{s-1})}^{\rm ell}$, which was defined in section \ref{sec:determinants}, but in principle one could use the same technique in order to compute partition functions in the more general case of a genus $g$ Riemann surface $\mathcal C_g$.

\subsection{General Riemann Surfaces}\label{determinants-g}
When we switch from the genus $0$ case to a generic Riemann surface $\mathcal C_g$ with $1$ puncture, we are effectively turning on a matter bundle corresponding to the contribution of $g$ adjoint hypermultiplets, and the quiver in figure \ref{fig:quiver_sphere} describing the GLSM we are studying gets modified into quiver \ref{fig:quiver_riemann_g}.

\begin{figure}[H]
\centering\vspace{-5mm}
\begin{tikzpicture}
\node[FrameNode](F0) at (2.5+10,-2){$r$};
\node[GaugeNode](Gk0) at (2.5+10,0){$n_0$};
\node[GaugeNode](Gk1) at (5+10,0){$n_1$};
\node[](label) at (7+10,0){$\cdots$};
\node[](dots2) at (2.5+10,1.5){$\vdots$};
\node[](dots2) at (1.3+10,1.3){$2g+2\ \bigg\{$};
\node[GaugeNode](Gkn) at (9+10,0){$n_{s-1}$};
\draw[->-](F0.15+90) to (Gk0.165+90);
\draw[->-](Gk0.195+90) to (F0.345+90);
\draw[->-](Gk1) to (Gk0);
\draw[->-](label) to (Gk1);
\draw[->-](Gkn) to (label);
\draw[-](Gk0) to[out=65,in=115,looseness=12] (Gk0);
\draw[-](Gk0) to[out=70,in=110,looseness=7] (Gk0);
\draw[-](Gk1) to[out=65,in=115,looseness=9] (Gk1);
\draw[-](Gk1) to[out=65+180,in=115+180,looseness=9] (Gk1);
\draw[-](Gkn) to[out=65,in=115,looseness=9] (Gkn);
\draw[-](Gkn) to[out=65+180,in=115+180,looseness=9] (Gkn);
\end{tikzpicture}\caption{Low energy GLSM quiver for a general $\mathcal C_{g,1}$.}\label{fig:quiver_riemann_g}
\end{figure}
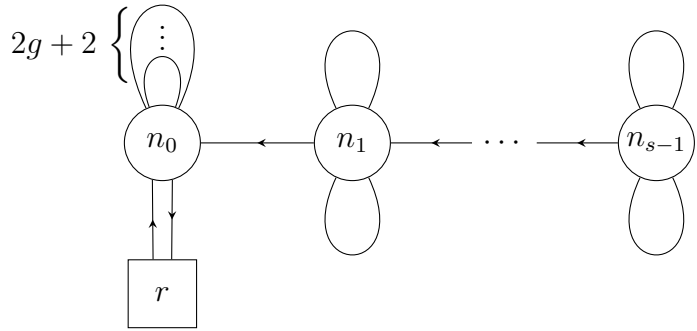

This GLSM encodes the ADHM construction of $\mathcal N_{r,[r^1],n,\mu}$ with additional $g$ hypermultiplets in the adjoint representation, all of them with twisted mass $m$, which reproduces an $\mathcal N=(0,2)^*$ theory.\footnote{Strictly speaking we are dealing with an $\mathcal N=(0,2)^*$ theory only in the case in which $\mathcal C$ is a $g=1$ Riemann surface. In the same spirit we might want to point out that the $5d$ partition function to which the elliptic index is reduced in the QM limit is not really computing the equivariant virtual $\chi_{y}-$genus of the vector bundle $\mathcal V_g$, but rather the equivariant virtual Euler characteristic of an antisymmetric power of $\mathcal V_g$. This also means that the torus partition function is not, strictly speaking, the (equivariant virtual) elliptic genus as it is defined in \cite{fantechi2010}, as it is instead an elliptic generalization of the virtual Euler characteristic.} In the same spirit as in \cite{Bruzzo:2002xf}, from the point of view of the matter fields this consists in introducing $2g$ adjoint chirals and $2g$ fundamental chirals, with appropriate relations dictated by the brane system. As it was the case for the theory without any adjoint hypermultiplet, each of the fundamentals we introduce makes the theory anomalous by breaking a chiral half of the supersymmetry, and this phenomenon can be cured by insertion of observables, encoded in $\overline{\rm D7}$ contributions. The additional field content to \ref{table:ADHM_field_cont} is summarized in table \ref{table:ADHM_field_cont_addition},
\begin{table}[H]
\centering
\begin{tabular}{|r|c|c|c|c|}\hline\hline
 & gauge $G$ & flavour $U(1)\times U(1)^2$ & twisted mass & $R-$charge\\ \hline\hline
$B_3^i$ & $\overline{\Box}_0\otimes\Box_0$ & $\mathbbm 1_{(0,0)}$ & $m$ & $-q+t$ \\
$B_4^i$ & $\overline{\Box}_0\otimes\Box_0$ & $\mathbbm 1_{(-1,-1)}$ & $\epsilon-m$ & $-q-t$ \\
$K_i$ & $\overline\Box_0$ & $\overline\Box_{(0,0)}$ & $a-m$ & $-p-t$ \\
$L_i$ & ${\Box}_0$ & $\Box_{(1,1)}$ & $-a+\epsilon-m$ & $p-t$ \\
$\chi_0^{(3),i}$ & $\overline{\Box}_0\otimes\Box_0$ & $\mathbbm 1_{(0,1)}$ & $\epsilon_1-m$ & $-t$\\
$\chi_0^{(4),i}$ & $\overline{\Box}_0\otimes\Box_0$ & $\mathbbm 1_{(-1,0)}$ & $m-\epsilon_2$  & $t$\\ \hline\hline
\end{tabular}\caption{Hypermultiplet additional fields for quiver \ref{fig:quiver_riemann_g}}\label{table:ADHM_field_cont_addition}
\end{table}
while the ADHM relation on the $n_0$ node must be modified
\begin{equation}
[B_1^0,B_2^0]+\sum_{i=1}^g[B_3^{i\dag},B_4^{i\dag}]+IJ=0
\end{equation}
and the relations \eqref{additional_rels_i}-\eqref{additional_rels_f} must be enforced through $\chi_0^{(3,4),i}$, $K_i$ and $L_i$.
\begin{align}\label{additional_rels_i}
\mathcal E_{3,i}^{\rm adj}&=[B_1,B_3^i]-[B_2^\dag,B_4^{i\dag}]\\
\mathcal E_{4,i}^{\rm adj}&=[B_1,B_4^i]-[B_2^\dag,B_3^{i\dag}]\\
\mathcal E_{K_i}^{\rm fun}&=B_3^iI-B_4^{i\dag} J^\dag\\
\mathcal E_{L_i^\dag}^{\rm fun}&=B_4^iI+B_3^{i\dag} J^\dag \label{additional_rels_f}
\end{align}

The partition function for a general genus $g$ riemann surface $\mathcal C_g$ with one puncture will now read (we take the $r=1$ case for the sake of simplicity)
\begin{equation}\label{ZT2-0-g}
\mathcal Z^{\rm ell}_1(\mathcal C_g;q_0,\dots,q_{s-1})=\sum_{\mu_1\subseteq\cdots\subseteq\mu_0}q_0^{|\mu_0|}q_1^{|\mu_0\setminus\mu_1|}\cdots q_{s-1}^{|\mu_0\setminus\mu_{s-1}|}Z_{(\mu_0,\mu_1,\dots,\mu_{s-1})}^{{\rm ell},g},
\end{equation}
with
\begin{equation}
Z_{(\mu_0,\mu_1,\dots,\mu_{s-1})}^{{\rm ell},g}=\mathcal L_{\mu_0}^{\rm ell}\mathcal N_{\mu_0}^{\rm ell}\overline{\mathcal N}_{\mu_0}^{\rm ell}\mathcal E_{g,\mu_0}^{\rm ell}\overline{\mathcal E}_{g,\mu_0}^{\rm ell}\mathcal T_{\mu_0,\mu_1}^{\rm ell}\overline{\mathcal T}_{\mu_0,\mu_1}^{\rm ell}\mathcal W_{\mu_0,\dots,\mu_{s-1}}^{\rm ell},
\end{equation}
where we defined
\begin{align}
\mathcal E_{g,\mu_0}^{\rm ell}&=\prod_{s\in Y_{\mu_0}}\theta_1^g(\tau|E(s)-m)\theta_1^g(\tau|E(s)-\epsilon+m),\\
\overline{\mathcal E}_{g,\mu_0}^{\rm ell}&=\prod_{s\in\mu_0\setminus\Box}\frac{1}{\theta_1^g(\tau|\phi(s)-\tilde a-m)\theta_1^g(\tau|\phi(s)-\tilde a+\epsilon-m)}.
\end{align}
We remark that by setting $g=0$ we readily recover the function $Z_{(\mu_0,\dots,\mu_{s-1})}^{\rm ell}$ which is needed in order to compute the partition function $\mathcal Z_1^{\rm ell}(S^2;q_0,\dots,q_{s-1})$.

In the same way as we did in section \ref{sec:determinants}, we can compute the full partition function $\mathcal Z_1^{\rm ell}(\mathcal C_g;q_0,\dots,q_{s-1})$ by first summing over the nested partitions $\mu_1\subseteq\cdots\subseteq\mu_{s-1}$ and use the definition \eqref{Pell} of $P^{\rm ell}$ in order to get
\begin{equation}\label{ZT2-g-full}
\mathcal Z^{\rm ell}_1(\mathcal C_g;q_0,\dots,q_{s-1})=\sum_{\mu_0}q_0^{|\mu_0|}\mathcal Y^{\rm ell}_{g,\mu_0}P^{\rm ell}_{\mu_0},
\end{equation}
with the following definition of $\mathcal Y^{\rm ell}_{g,\mu_0}$
\begin{equation}\label{Yell-g}
\mathcal Y^{\rm ell}_{g,\mu_0}=\mathcal L_{\mu_0}^{\rm ell}\mathcal N_{\mu_0}^{\rm ell}\overline{\mathcal N}_{\mu_0}^{\rm ell}\mathcal E^{\rm ell}_{g,\mu_0}\overline{\mathcal E}^{\rm ell}_{g,\mu_0}.
\end{equation}
Again we remark that $P_{\mu_0}^{\rm ell}\xrightarrow{q_{i>0}\to 0}1$ so that $\mathcal Z_{1}^{\rm ell}(\mathcal C_g;q_0,\dots,q_{s-1})\xrightarrow{q_{i>0}\to 0}\mathcal Z_1^{\rm ell}(\mathcal C_g;q_0)$.

\subsubsection{Comet shaped quiver}
Finally, we are interested in computing the partition function on a Riemann surface $\mathcal C_g$ with $k$ punctures of generic holonomy, whose low energy GLSM is in general described by the quiver in figure \ref{fig:quiver_comet}.
\begin{figure}[H]
\centering
\vspace{-5mm}
\begin{tikzpicture}
\node[FrameNode](F) at (2.5+10,-2){$r$};
\node[GaugeNode](Gk0) at (2.5+10,0){$n_0$};
\node[GaugeNode](Gk1) at (5+10,2){$n_1^{(1)}$};
\node[GaugeNode](Gk12) at (5+10,-2){$n_1^{(k)}$};
\node[](label) at (7+10,2){$\cdots$};
\node[](label2) at (7+10,-2){$\cdots$};
\node[](dots1) at (11,0){$\cdots$};
\node[](g) at (11.2,1.05){$2g+2$};
\node[rotate=90](bracket) at (11.2,.75){$\Bigg\}$};
\node[](vdots1) at (5+10,0){$\vdots$};
\node[](vdots2) at (9+10,0){$\vdots$};
\node[GaugeNode](Gkn) at (9+10,2){$n_{s-1}^{(1)}$};
\node[GaugeNode](Gkn2) at (9+10,-2){$n_{s-1}^{(k)}$};
\draw[->-](F.15+90) to (Gk0.165+90);
\draw[->-](Gk0.195+90) to (F.345+90);
\draw[->-](Gk1) to (Gk0);
\draw[->-](Gk12) to (Gk0);
\draw[->-](label) to (Gk1);
\draw[->-](label2) to (Gk12);
\draw[->-](Gkn) to (label);
\draw[->-](Gkn2) to (label2);
\draw[-](Gk0) to[out=60+90,in=120+90,looseness=12] (Gk0);
\draw[-](Gk0) to[out=70+90,in=110+90,looseness=7] (Gk0);
\draw[-](Gk1) to[out=65+180,in=115+180,looseness=8] (Gk1);
\draw[-](Gk1) to[out=65,in=115,looseness=8] (Gk1);
\draw[-](Gkn) to[out=65+180,in=115+180,looseness=8] (Gkn);
\draw[-](Gkn) to[out=65,in=115,looseness=8] (Gkn);
\draw[-](Gk12) to[out=65+180,in=115+180,looseness=8] (Gk12);
\draw[-](Gk12) to[out=65,in=115,looseness=8] (Gk12);
\draw[-](Gkn2) to[out=65+180,in=115+180,looseness=8] (Gkn2);
\draw[-](Gkn2) to[out=65,in=115,looseness=8] (Gkn2);
\end{tikzpicture}\caption{Comet-shaped quiver.}\label{fig:quiver_comet}
\end{figure}

We will start from the case of $\mathcal C_0=S^2$, which will take the form \eqref{ZT2-0-k}
\begin{equation}\label{ZT2-0-k}
\begin{split}
\mathcal Z_k^{\rm ell}(S^2;q_0,\{q_1^i,\dots,q_{s-1}^i\})=\sum_{\mu_0}q_0^{|\mu_0|}\sum_{\{\mu_1^i\subseteq\cdots\subseteq\mu_{s-1}^i\}_{i=1}^k}\prod_{j=1}^k\left(q_1^{|\mu_0\setminus\mu_1^j|}\cdots\right.\\
\left.\cdots q_{s-1}^{|\mu_0\setminus\mu_{s-1}^j|}\right)Z_{(\mu_0,\{\mu^i_1,\dots,\mu^i_{s-1}\})}^{\rm ell}.
\end{split}
\end{equation}
In this case the virtual tangent space to $\mathcal N_{r,[r^1],n,\mu}$ in \eqref{tan_rep_long} will be modified to be of the form \eqref{tan_rep_comet}.
\begin{equation}\label{tan_rep_comet}
\begin{split}
T_Z^{\textrm{vir}}\mathcal N(r,\{n_0^i,\ldots,n_{s-1}^i\})&=\End(V_0)\otimes(Q-1-\Lambda^2Q)+\Hom(W,V_0)+\Hom(V_0,W)\otimes\Lambda^2Q\\
&\quad-\sum_{i=1}^k\Hom(V_1^{(k)},W)\otimes\Lambda^2Q+\\
&\quad+\sum_{i=1}^k\sum_{\ell=1}^{s-1}\left[\End(V_{\ell}^{(k)})-\Hom(V_{\ell}^{(k)},V_{\ell-1}^{(k)})\right]\otimes(Q-1-\Lambda^2Q).
\end{split}
\end{equation}

By a simple generalization of the computations leading to \eqref{tan_rep_long_char} it is possible to see that $\mathcal Z_k^{\rm ell}(S^2;q_0,\{q_1^i,\dots,q_{s-1}^i\})$ takes a form similar to \eqref{ZT2-0-full}, as is shown in \eqref{ZT2-0-k-full}
\begin{equation}\label{ZT2-0-k-full}
\mathcal Z_k^{\rm ell}(S^2;q_0,\{q_1^i,\dots,q_{s-1}^i\})=\sum_{\mu_0}q_0^{|\mu_0|}\mathcal Y^{\rm ell}_{\mu_0}\prod_{i=1}^kP^{{\rm ell},i}_{\mu_0},
\end{equation}
with
\begin{equation}\label{Pell-i}
P^{{\rm ell},i}_{\mu_0}=\sum_{\mu_1^i\subseteq\cdots\subseteq\mu_{s-1}^i}\mathcal T_{\mu_0,\mu_1^i}^{\rm ell}\overline{\mathcal T}_{\mu_0,\mu_1^i}^{\rm ell}\mathcal W_{\mu_0,\dots,\mu_{s-1}^i}^{\rm ell}\left(q_1^i\right)^{|\mu_0\setminus\mu_1^i|}\cdots \left(q_{s-1}^i\right)^{|\mu_0\setminus\mu_{s-1}^i|},
\end{equation}
and the functions $\mathcal T_{\mu_0,\mu_1^i}^{\rm ell}$, $\overline{\mathcal T}_{\mu_0,\mu_1^i}^{\rm ell}$ and $\mathcal W_{\mu_0,\dots,\mu_{s-1}^i}$ take the same form as in equations \eqref{CS-T2}-\eqref{WT2}.

By a completely analogous procedure we can get that partition function of the low energy theory relative to a general Riemann surface of genus $g$, possibly $g=0$. By using the results of section \ref{determinants-g}, we easily see that
\begin{equation}\label{ZT2-g-k}
\begin{split}
\mathcal Z_k^{\rm ell}(\mathcal C_g;q_0,\{q_1^i,\dots,q_{s-1}^i\})=\sum_{\mu_0}q_0^{|\mu_0|}\sum_{\{\mu_1^i\subseteq\cdots\subseteq\mu_{s-1}^i\}_{i=1}^k}\prod_{j=1}^k\left(q_1^{|\mu_0\setminus\mu_1^j|}\cdots\right.\\
\left.\cdots q_{s-1}^{|\mu_0\setminus\mu_{s-1}^j|}\right)Z_{(\mu_0,\{\mu^i_1,\dots,\mu^i_{s-1}\})}^{{\rm ell},g}.
\end{split}
\end{equation}
By turning on the matter bundle described in section \ref{determinants-g} on the moduli space of nested instantons $\mathcal N(r,n_0,\{n_1^i,\dots,n_{s-1}^i\})$, whose virtual tangent space is given in equation \eqref{tan_rep_comet} as an element of the representation ring of the torus $R(T)$, the suspersymmetric localization theorem (or equivalently the equivariant one) gives us \eqref{ZT2-g-k-full},
\begin{equation}\label{ZT2-g-k-full}
\mathcal Z_k^{\rm ell}(\mathcal C_g;q_0,\{q_1^i,\dots,q_{s-1}^i\})=\sum_{\mu_0}q_0^{|\mu_0|}\mathcal Y^{\rm ell}_{g,\mu_0}\prod_{i=1}^kP^{{\rm ell},i}_{\mu_0},
\end{equation}
where $P^{{\rm ell},i}_{\mu_0}$ is defined in \eqref{Pell-i} and $\mathcal Y_{g,\mu_0}^{\rm ell}$ is the same one as in equation \eqref{Yell-g}.

A couple of final remarks are due here. First of all we notice that we can switch off any number of the contributions of the tails of the comet shaped quiver \ref{fig:quiver_comet} by taking the limit to $0$ of the respective instanton counting parameters. Then, given any $k'<k$ we have that
\begin{equation}
\mathcal Z_k^{\rm ell}(\mathcal C_g;q_0,\{q_1^i,\dots,q_{s-1}^i\})\xrightarrow[\substack{i=1,\dots,s-1\\ j=k'+1,\dots,k}]{q_i^j\to 0}\mathcal Z_{k'}^{\rm ell}(\mathcal C_g;q_0,\{q_1^i,\dots,q_{s-1}^i\}).
\end{equation}

Moreover, we expect our partitions functions to be computing the equivariant elliptic cohomology of the moduli spaces of stable representations of quivers \ref{fig:quiver_sphere}-\ref{fig:quiver_comet}, as in \cite{Rimanyi:2019ubu}.

\subsection{Limit to supersymmetric quantum mechanics}\label{sec:partition_functions_QM}
We now want to study a particular dimensional reduction of the $2d$ $\mathcal N=(0,2)$ system we studied on $T^2$ in the previous subsections. By reducing on a circle we get the Witten index of an $\mathcal N=2$ SQM. This dimensional reduction can be obtained from the elliptic case we just studied by taking the limit $\eu^{2\pi\iu\tau}\to 0$. In this scaling limit we can use the fact that $\theta_1(\tau|z)\to 2q^{1/8}\sin(\pi z)$ as $q=\eu^{2\pi\iu\tau}\to 0$. In the resulting theory on $S^1$ we can decouple the $\overline{\rm D7}$ branes by taking very large values of the Cartan parameter $\overline a$ and then rescaling the gauge coupling. As we already anticipated, we will see how the results we obtain by this procedure compute particular equivariant virtual invariants of the bundle $\mathcal V_g$ over the moduli space of nested instantons $\mathcal N_{r,[r^1],n,\mu}$, which is described by the stable representations of the quiver in figure \ref{fig:quiver_riemann_g}. A bit of care is required in order to take the correct scaling limit, and in particular one has to require that $q\to 0$ while $\vol(T^2)\to\beta=r_{S^1}$. Moreover, one should take into account that in the $S^1$ theory twisted masses are also rescaled by $\beta$, so that the result may be expressed in terms of $q_1=\eu^{\beta\epsilon_1/2}$, $q_2=\eu^{\beta\epsilon_2/2}$ and $y=\eu^{-\beta m}$.

The geometric interpretation of the Witten index of the quiver gauge theories described in the previous section is the equivariant (virtual) Euler characteristic of a given bundle over the moduli space of nested instantons. Then, computing the Witten index geometrically amounts to studying the stable representations in the category of vector spaces of the quiver \ref{fig:quiver_local} under suitable stability conditions. This procedure has the advantage of letting us compute the weight decomposition of the virtual tangent space $T^{\rm vir}_Z\mathcal N_{r,[r^1],n,\mu}$ at the fixed points $Z$ in the representation ring of the torus. The way in which this is done is very briefly described in section \ref{sec:fixed}. As it is shown in subsection \ref{sec:fixed}, the fixed locus of the torus action consists only of isolated points, which are characterized in terms of $s-$tuples of nested colored partitions $\boldsymbol\mu_1\subseteq\cdots\boldsymbol\mu_{s-1}\subseteq\boldsymbol\mu_0$, such that $|\boldsymbol\mu_0|=n_0=n$, while $|\boldsymbol\mu_j|=n_0-n_j$.

Once the fixed point locus has been completely characterized and a weight decomposition of the virtual tangent space is at hand, one can in full generality define an $s-$parameter family of partition functions on $\mathcal N_{r,[r^1],n,\mu}$, with parameters $\bold p=(p_0,p_1,\dots,p_{s-1})\in\mathbb Z^{s}$. In terms of the quiver vector spaces $(W,V_0,\dots,V_{s-1})$ one can introduce $(s+1)-$tautological bundles $\mathcal W$ and $\mathcal V_i$, $i=0,\dots,s-1$, with $\mathcal W=\mathcal O_{\mathcal N_{r,[r^1],n,\mu}}$. We can then define $\mathcal L_i=\det\mathcal V_i$, $\mathcal L_{\bold p}=\bigotimes_i\mathcal L_i^{\otimes p_i}$ and compute the virtual Euler characteristic of the bundle $S\otimes\mathcal L_{\bold p}$ over $\mathcal N_{r,[r^1],n,\mu}$, with $S$ an arbitrary irreducible representation of $T$. The generating function of the virtual Euler characteristics of the moduli space of nested instantons in \eqref{gen_euler} will then reproduce the QM partition function \ref{S1_partitionfunction}, when $\bold p=(1-g,0,\dots,0)$.\footnote{It is interesting to compare the role of this line bundle $\mathcal L$ to the way in which the Cern-Simons term was introduced in section \ref{sec:QM_reduction}. In particular it turns out that the vector space $V_0$ can be recognized to be the space of fermionic zero modes, \cite{Kim:2002qm,Tachikawa:2004ur,Nakajima:2003uh}, so that the identification of $\mathcal L_{(1,0,\dots,0)}=\det\mathcal V_0$ with $\Det\slashed D$ is in fact quite natural.}
\begin{equation}\label{gen_euler}
Z_{\bold p}^{\textrm{vir}}(q_1,q_2,\bold x)=\sum_{\bold n\in\mathbb Z^{s}_{\ge 0}}\ch_{T}\chi^{\textrm{vir}}_{T}\left(\mathcal N_{r,[r^1],n,\mu},\mathcal L_{\bold p}\right)\prod_{i=1}^{s}x_i^{n_i}.
\end{equation}
In the following we use the notation $\ch_T$ to denote the $T-$equivariant Chern character of a vector bundle, which has a very convenient representation in the representation ring $R(T)$. The usual Chern character is defined as follows: if $E$ is rank $r$ vector bundle over $X$, with Chern roots $x_1,\dots,x_r$, then one defines
\begin{equation}
\ch(E)=\sum_{i=1}^r\eu^{x_i},
\end{equation}
which can be equivariantly extended to a ring homomorphism $\ch_G:K^i_G(X)\to H^i_G(\hat X,\mathbb C)$, where $\hat X=\{(x,g)\in X\times G|xg=x\}=\coprod_gX^g$ and $H^i_G(\hat X,\mathbb C)\simeq\left[\bigoplus_gH^i(X^g,\mathbb C)\right]^G$. The effect of $\ch_G$ can be concretely characterized as follows: if $E$ is a $G-$equivariant vector bundle on $X$, for each $x\in X^g$, we can compute the eigenvalues (supposed to be distinct) $\lambda_1,\dots,\lambda_r$ of the $G-$action, and the corresponding eigenspaces $E_x^1,\dots,E_x^r$, so that $E|_{X^g}$ can be represented as the direct sum of vector bundles
\begin{equation}
E_{X^g}=E^1\oplus\cdots\oplus E^r.
\end{equation}
Finally one defines $\ch_g(E)=\sum_i\lambda_i\ch(E^i)$, so that
\begin{equation}
\ch_G(E)=\bigoplus_{g\in G}\ch_g(E)\in\left[\bigoplus_{g\in G}H^{\rm ev}(X^g,\mathbb C)\right]^G.
\end{equation}

The Chern character, and also the equivariant Chern character, satisfies some important properties which we will use extensively in the following:
\begin{equation}
\ch(E\oplus F)=\ch E+\ch F,\qquad\ch(E\otimes F)=\ch E\ch F.
\end{equation}

If we restrict to the case $\bold p=(p_0,0,\dots,0)$, the fiber of $\mathcal L_{\bold p}$ at a fixed point $Z\leftrightarrow\boldsymbol\mu_1\subseteq\cdots\subseteq\boldsymbol\mu_{s-1}\subseteq\boldsymbol\mu_0$ will be given by \eqref{line_bdl_fix},
\begin{equation}\label{line_bdl_fix}
\mathcal L_Z=\mathcal L_{\boldsymbol\mu_0}=\left(\prod_{\alpha=1}^r\prod_{i=1}^{M_0^{(a)}}\prod_{j=1}^{\mu_{0,i}^{(a)\prime}}T_{a_\alpha}T_1^{-i+1}T_2^{-j+1}\right)^{p_0},
\end{equation}
where $(T_1,T_2,T_{a_1},\dots,T_{a_r})$ denote the fundamental characters of $T\times\left(\mathbb C^*\right)^r\curvearrowright\mathcal N_{r,[r^1],n,\mu}$ and $M_0^{(a)}=\mu_{0,1}^{(a)\prime}$, $N_0^{(a)}=\mu_{0,1}^{(a)}$.

Then supersymmetric localization (equiv. equivariant localization) can be exploited in order to compute partition functions (equiv. virtual equivariant Euler characteristics). If we start from the case $g=0$ we get
\begin{equation}
\begin{split}
\ch_{T}\left[\chi_{T}^{\textrm{vir}}\left(\mathcal N(r,\bold n),\mathcal L_{\bold p}\right)\right]&=\sum_{Z\in\mathcal N_{r,[r^1],n,\mu}^{T}}\frac{\ch_{T}\mathcal L_Z}{\Lambda_{-1}\left[T_Z^{\textrm{vir}}\mathcal N_{r,[r^1],n,\mu}^\vee\right]}\\
&=\sum_{\boldsymbol\mu_1\subseteq\cdots\subseteq\boldsymbol\mu_{s-1}\subseteq\boldsymbol\mu_0}\left(\frac{\mathcal L_{\boldsymbol\mu_0}(q_1,q_2)}{\Lambda_{-1}\left[T_{\tilde Z}\mathcal M_{r,n_0}^\vee\right]}\mathcal T_{\boldsymbol\mu_0,\boldsymbol\mu_1}(q_1,q_2)\cdot\right.\\
&\qquad\qquad\qquad\qquad\qquad\qquad\qquad\cdot\mathcal W_{\boldsymbol\mu_{0},\dots,\boldsymbol\mu_{s-1}}(q_1,q_2)\Bigg),
\end{split}
\end{equation}
with $\Lambda_t(E)=\sum_{i\ge 0}t^i\Lambda^i E$ for any (equivariant) vector bundle $E$ on $\mathcal N_{r,[r^1],n,\mu}$, $\mathcal L_{\boldsymbol\mu_0}(q_1,q_2)$, $\mathcal W_{\boldsymbol\mu_0,\dots,\boldsymbol\mu_{s-1}}(q_1,q_2)$ and $\mathcal T_{\boldsymbol\mu_{0},\boldsymbol\mu_1}(q_1,q_2)$ given by equations \eqref{part_f_defs_i}-\eqref{part_f_defs_f},
\begin{align}\label{part_f_defs_i}
\mathcal L_{\boldsymbol\mu_0}(q_1,q_2)&=\left(\prod_{a=1}^r\prod_{i=1}^{M_0^{(a)}}\prod_{j=1}^{\mu_{0,i}^{(a)\prime}}\rho_a q_1^{i-1}q_2^{j-1}\right)^{p_0}\\
\mathcal T_{\boldsymbol\mu_0,\boldsymbol\mu_1}(q_1,q_2)&=\prod_{a=1}^r\prod_{i=1}^{M_0^{(a)}}\prod_{j=1}^{\mu_{0,i}^{(a)}-\mu_{1,i}^{(a)}}\left(1-\rho_a q_1^{-i}q_2^{-j-\mu_{1,i}^{(a)\prime}}\right)\\ \nonumber
\mathcal W_{\boldsymbol\mu_{0},\dots,\boldsymbol\mu_{s-1}}(q_1,q_2)&=\prod_{k=0}^{s-2}\prod_{a,b=1}^r\prod_{i=1}^{M_0^{(a)}}\prod_{j=1}^{N_0^{(b)}}\left(\frac{(1-\rho_a\rho_b^{-1}q_1^{\mu_{k,j}^{(b)}-i}q_2^{j-\mu_{k+1,i}^{(a)\prime}-1})}{(1-\rho_a\rho_b^{-1}q_1^{\mu_{k+1,j}^{(b)}-i}q_2^{j-\mu_{k+1,i}^{(a)\prime}-1})}\right.\\ \label{part_f_defs_f}
&\qquad\qquad\qquad\qquad\qquad\left.\frac{(1-\rho_a\rho_b^{-1}q_1^{\mu_{k+1,j}^{(b)}-i}q_2^{j-\mu_{0,i}^{(a)\prime}-1})}{(1-\rho_a\rho_b^{-1}q_1^{\mu_{k,j}^{(b)}-i}q_2^{j-\mu_{0,i}^{(a)\prime}-1})}\right)
\end{align}
with $\rho_i=\ch T_{a_i}$ and similarly $q_i=\ch T_i$.

The generalization to the case of a general Riemann surface $\mathcal C_g$ of genus $g$ is immediate, as it only amounts to computing the ``virtual Hirzebruch $\chi_y-$genus'' of the bundle $\pi^*\mathcal V_g\to\mathcal N_{r,[r^1],n,\mu}$. This is obviously the same as turning on a matter bundle relative to additional $g$ adjoint hypermultiplets, whose twisted mass $m$ is naturally identified with $y$ in the Hirzebruch genus by exponentiation.
\begin{equation}\label{chiy_tail}
\begin{split}
\ch_{T}\chi_y^{T,{\rm vir}}\left(\pi^*\mathcal V_g,\mathcal N_{r,[r^1],n,\mu}\right)=\sum_{\boldsymbol\mu_1\subseteq\cdots\subseteq\boldsymbol\mu_{s-1}\subseteq\boldsymbol\mu_0}\left(\frac{\ch_{T}(\mathcal L_{\boldsymbol\mu_0})\ch_{T}\Lambda_{y}[(T_{\boldsymbol\mu_0}\mathcal M_{r,n_0}^\vee)^{\oplus g}]}{\ch_{T}\Lambda_{-1}[T_{\boldsymbol\mu_0}\mathcal M_{r,n_0}^\vee]}\cdot\right.\\
\left.\cdot\mathcal T_{\boldsymbol\mu_1,\boldsymbol\mu_0}(q_1,q_2)\prod_{i=0}^{s-2}\mathcal W_{\boldsymbol\mu_{i+1},\boldsymbol\mu_i}^{\boldsymbol\mu_0}(q_1,q_2)\right).
\end{split}
\end{equation}
Surprisingly enough, explicit computations suggest that the partition function of each choice of numerical type for the nested instantons quiver should consists of a usual Nekrasov partition function multiplied by a polynomial in the torus characters. This observation is summarized in the following conjecture.
\begin{conjecture}
The function $\sum_{\boldsymbol\mu_{i>0}}\mathcal T_{\boldsymbol\mu_0,\boldsymbol\mu_1}(q_1,q_2)\mathcal W_{\boldsymbol\mu_0,\dots,\boldsymbol\mu_{s-1}}(q_1,q_2)$ is a polynomial in $q=q_1^{-1}$ and $t=q_2^{-1}$ with rational coefficients in the $\{\rho_i\}_{1\le i\le r}$, while it is a polynomial with integer coefficients when $r=1$.
\end{conjecture}

\subsection{Comparison to LHRV formulae}
The Nekrasov partition function on $\mathbb R^4\times S^1$ is known to have the following form
\begin{equation}\label{Nek_R4S1}
Z_{k,N}^{\mathbb R^4\times S^1}=\sum_{\bold Y_k}\prod_{\lambda,\tilde\lambda=1}^N\prod_{s\in Y_\lambda}\frac{\sinh\left[\frac{\beta}{2}(E(s)-m))\right]\sinh\left[\frac{\beta}{2}(E(s)-\epsilon+m)\right]}{\sinh\left[\frac{\beta}{2}E(s)\right]\sinh\left[\frac{\beta}{2}(E(s)-\epsilon)\right]},
\end{equation}
where $E(s)=a_{\lambda\tilde\lambda}-\epsilon_1h(s)+\epsilon_2(v(s)+1)$, and given two Young diagrams $Y_\lambda,Y_{\tilde\lambda}\in\bold Y_k$ the quantities $h(s)$ and $v(s)$ are defined to be $h(s)=\nu_{i_\lambda}-j_\lambda$ and $v(s)=\tilde\nu'_{j_\lambda}-i_\lambda$. We will be interested in the specialization of the Nekrasov partition function to the case $N=1$, so that $h(s)$ and $v(s)$ will become respectively the arm length $a(s)$ and leg length $l(s)$ for the box $s$ in the Young tableaux classifying a given pole configuration.
Now, following the conventions of \cite{HLRV}, let $\bold x_1=\{x_{1,1},x_{1,2},\dots\}$ and $\bold x_k=\{x_{k,1},x_{k,2},\dots\}$ be $k$ infinite sets of variables and let moreover $\Lambda(\bold x_1),\dots,\Lambda(\bold x_k)$ be the corresponding rings of symmetric functions. Given a partition $\lambda$, $\tilde H_\lambda(\bold x;q,t)\in\Lambda(\bold x)\otimes_{\mathbb Z}\mathbb Q[q,t]$ will denote the modified Macdonald symmetric function. The $k-$point genus $g$ Cauchy function $\Omega(z,w)$, with coefficients in $\mathbb Q[z,w]\otimes_{\mathbb Z}\Lambda(\bold x_1,\dots,\bold x_k)$, is defined as follows
\begin{equation}
\Omega(z,w)=\sum_{\lambda\in\mathcal P}\mathcal H_\lambda(z,w)\prod_{i=1}^k\tilde H_\lambda(\bold x_i;z^2,w^2),
\end{equation}
with
\begin{equation}
\mathcal H_\lambda(z,w)=\prod_{s\in\lambda}\frac{(z^{2a(s)+1}-w^{2l(s)+1})^{2g}}{(z^{2a(s)+2-w^{2l(s)}})(z^{2a(s)}-w^{2l(s)+2})}.
\end{equation}
The modified Macdonald polynomials $\tilde H_\lambda(\bold x;q,t)$ are defined as
\begin{equation}
\tilde H_\lambda(\bold x;q,t)=\sum_{\mu}\tilde K_{\mu\lambda}(q,t)s_\mu(\bold x),
\end{equation}
where $s_\lambda(\bold x)$ are the usual Schur functions, while $\tilde K_{\lambda\mu}(q,t)$ denotes the modified Kostka polynomials, which are expressed in terms of the usual Kostka polynomials as
\begin{equation}
\tilde K_{\lambda\mu}(q,t)=t^{n(\mu)}K_{\lambda\mu}(q,t^{-1}),
\end{equation}
with $n(\mu)=\sum_{i=1}^{l(\mu)}\mu_i(i-1)$, and $K_{\lambda\mu}(q,t)$ can be interpreted as being a deformation of the Kostka coefficients $K_{\lambda\mu}$ appearing in the expansion of the Schur polynomials in terms of the monomial symmetric functions:
\begin{equation}
s_\lambda(\bold x)=\sum_{\mu}K_{\lambda\mu}m_\mu(\bold x).
\end{equation}

The modified Macdonald polynomials by themselves can be viewed as a $q-$deformation of the standard Hall-Littlewood polynomials, and are related in a non trivial way to the Macdonald polynomials $P_\mu(\bold x;q,t)$, which are eigenfunctions of the trigonometric Ruijsenaars-Schneider Hamiltonian \cite{Koroteev:2015dja,Garsia4313}:
\begin{equation}
\tilde H_\lambda[X;q,t]=t^{n(\lambda)}J_\lambda\left[\frac{X}{1-1/t};q,1/t\right],
\end{equation}
where $X$ denotes the plethystic substitution $X=x_1+x_2+x_3+\cdots$, the square brackets are to be intended as a plethystic insertion and
\begin{equation}
J_\lambda(\bold x;q,t)=\prod_{s\in\lambda}\left(1-q^{a_\lambda(s)}t^{l_\lambda(s)+1}\right)P_\lambda(\bold x;q,t).
\end{equation}
The modified Macdonald polynomials are also eigenfunctions of a linear operator $\Delta$, \cite{haiman2002}, which acts on a symmetric function $f$ as
\begin{equation}
\Delta f=\left.f\left[X+\frac{(1-q)(1-t)}{z}\right]\Omega[-zX]\right|_{z^0},
\end{equation}
where $\Omega[X]=\sum_{n=0}^\infty h_n(X)$.

We will think to $\Omega(z,w)$ as being a function associated to a genus $g$ Riemann surface with $k$ punctures. Moreover, if we are give $\boldsymbol\mu=(\mu^1,\dots,\mu^k)\in\mathcal P^k$ we can define the following function
\begin{equation}
\mathbb H_{\boldsymbol\mu}(z,w)=(z^2-1)(1-w^2)\langle\PL\Omega(z,w),h_{\boldsymbol\mu}\rangle,
\end{equation}
where $h_{\boldsymbol\mu}=h_{\mu^1}(\bold x_1)\cdots h_{\mu^k}(\bold x_k)\in\Lambda(\bold x_1,\dots,\bold x_k)$ are the complete symmetric functions, and $\langle\cdot,\cdot\rangle$ is an extension of the Hall pairing. The interest in $\mathbb H_{\boldsymbol\mu}(z,w)$ lays in the fact that it encodes information both about $GL_n(\mathbb C)$ character varieties $\mathcal M_{\boldsymbol\mu}$ of $k-$punctured genus $g$ Riemann surfaces with generic semisimple conjugacy classes of type $\boldsymbol\mu$ at the punctures and about comet-shaped quivers $\mathcal Q_{\boldsymbol\mu}$ with $g$ loops and $k$ tails of length defined by $\boldsymbol\mu$. It is in fact conjectured that through the knowledge of $\mathbb H_{\boldsymbol\mu}(z,w)$ we can get the mixed Hodge polynomial and the $E-$polynomial (and thus the Euler characteristic) of both these character varieties and quiver varieties.

If we now study the particular case of comet-shaped quivers with $k=1$, $l(\mu)=1$ and $g=1$, whose corresponding quiver is the Jordan quiver, we can specialize $\bold x=(T,0,\dots)$ for some variable $T$ and $\tilde H_\lambda(T,0,\dots;z,w)=T^{|\lambda|}$, so that
\begin{equation}\label{Omega_Jordan}
\Omega(z,w)=\sum_{k}\sum_{|\lambda|=k}\prod\frac{\left(z^{2a(s)+1}-w^{2l(s)+1}\right)^2}{\left(z^{2a(s)+2}-w^{2l(s)}\right)\left(z^{2a(s)}-w^{2l(s)+2}\right)}T^{|\lambda|}.
\end{equation}
If we now compare \eqref{Omega_Jordan} to \eqref{Nek_R4S1} in the case $N=1$, with $m=\epsilon/2$, we can immediately see how closely $\Omega(z,w)$ resembles to $\sum_{k}Z_{k,1}^{\mathbb R^4\times S^1}q^k$ as long as we make the identifications $z^2=\eu^{\beta\epsilon_1}$, $w^2=\eu^{\beta\epsilon_2}$ and $T=q$, $q$ being the instanton counting parameter. 

If we next take $g$ to be arbitrary, but still take $k=1$ and $l(\mu)=1$ a generalization of our previous observations is straightforward. In fact, as we already pointed out in the previous sections, adding loops to the Jordan quiver has the net effect of introducing $2g+2$ matter fields $B_1$, $B_2$, $B_3^{(i)}$, $B_4^{(i)}$ (with $i=1,\dots,g$) transforming in the adjoint representation of the gauge group $U(k)$. The role played by each of the $B_3^{(i)}$, $B_4^{(i)}$ fields is analogous to the one of $B_3$ and $B_4$ in the ADHM linear sigma model with adjoint matter. Since all of these fields do not contribute with poles to the residue computation of the localization formula, if we choose their twisted masses and $R-$charges to be the same as the ones for $B_3$ and $B_4$ their net effect will be that of introducing a $g-$th power to the numerator of \eqref{Nek_R4S1} (which really is the meaning of turning on a matter bundle for $g$ adjoint hypermultiplets twisted by their mass $m$).

Actually one needs to turn on a Chern-Simons coupling in order to exactly reproduce $\Omega(z,w)$ starting from a gauge theory. In fact we can rewrite \eqref{Omega_Jordan} as
\begin{equation}
\begin{split}
\Omega(z,w)=&\sum_k\sum_{|\lambda|=k}\prod_{s\in\lambda}\left[(-1)^{g-1}\frac{(z^{2a(s)+1}w^{2l(s)+1})^g}{z^{2a(s)+2}w^{2l(s)+2}}\cdot\right.\\
&\left.\cdot\frac{(1-z^{-2a(s)-1}w^{2l(s)+1})^g(1-z^{2a(s)+1}w^{-2l(s)+1})^g}{(1-z^{-2a(s)-2}w^{2l(s)})(1-z^{2a(s)}w^{-2l(s)-2})}T^{|\lambda|}\right]
\end{split}
\end{equation}
and we can easily see that
\begin{displaymath}
\begin{split}
\prod_{s\in\lambda}\frac{(z^{2a(s)+1}w^{2l(s)+1})^g}{z^{2a(s)+2}w^{2l(s)+2}}&=\frac{1}{(zw)^{|\lambda|}}\prod_{s\in\lambda}(z^{2a(s)+2}w^{2l(s)+2})^{g-1}\\
&=\frac{1}{(zw)^{|\lambda|}}\left(z^{2\sum_s(a(s)+1)}w^{2\sum_s(l(s)+1)}\right)^{g-1}\\
&=\frac{1}{(zw)^{|\lambda|}}\left(z^{2\sum_s i(s)}w^{2\sum_s j(s)}\right)^{g-1}\\
&=\frac{(zw)^{|\lambda|(2g-2)}}{\eu^{a(g-1)|\lambda|}(zw)^{|\lambda|}}\prod_{s\in\lambda}\left(\eu^az^{2(i(s)-1)}w^{2(j(s)-1)}\right)^{g-1},
\end{split}
\end{displaymath}
which, apart from a harmless overall normalization, is the contribution of a Chern-Simons interaction at level $1-g$, \cite{Tachikawa:2004ur}. Thus we conclude that the partition function for the $5d$ $\mathcal N=1^*$ ADHM quiver theory with $g$ adjoint hypermultiplets and a Chern-Simons term at level $1-g$ reproduces the Cauchy function \eqref{Omega_Jordang} when resummed over all the instanton sectors (see also \cite{Chuang:2012dv}).
\begin{equation}\label{Omega_Jordang}
\Omega(z,w)=\sum_{k}\sum_{|\lambda|=k}\prod\frac{\left(z^{2a(s)+1}-w^{2l(s)+1}\right)^{2g}}{\left(z^{2a(s)+2}-w^{2l(s)}\right)\left(z^{2a(s)}-w^{2l(s)+2}\right)}T^{|\lambda|}.
\end{equation}

As it was shown in \cite{bruzzo2011,Chuang:2013wpa}, one interesting thing to point out in equation \eqref{Omega_Jordang} is that it computes a generating function for a geometric index. It is actually known that the moduli space of stable representations for the ADHM data \eqref{ADHM_data} is isomorphic to the Hilbert scheme of $\dim(V)=n$ points in $\mathbb C^2$ when $\dim(W)=1$.
\begin{equation}\label{ADHM_data}
\begin{tikzcd}
V \arrow[out=70,in=110,loop,swap,"B_1"] \arrow[out=250,in=290,loop,swap,"B_2"] \arrow[r,shift left=.5ex,"J"] & W \arrow[l,shift left=.5ex,"I"]
\end{tikzcd},\qquad [B_1,B_2]+IJ=0
\end{equation}
Then $\Omega_\lambda(q_1,q_2)$ such that $\Omega(z,w)=\sum_k\Omega_\lambda(z^2,w^2)T^{|\lambda|}$ is computing the Hirzebruch $\chi_y-$genus of a vector bundle over $(\mathbb C^2)^{[n]}$. In particular we have \cite{Chuang:2010ii,Chuang:2013wpa}
\begin{equation}
\begin{split}
\sum_{\lambda\in\mathcal P(n)}\Omega_\lambda(q_1,q_2,y)&=\ch_{ T}\chi_y\left[\left(T^\vee(\mathbb C^2)^{[n]}\right)^{\oplus g}\otimes\left(\det\mathcal T\right)^{1-g},(\mathbb C^2)^{[n]}\right]\\
&=\sum_{\lambda\in\mathcal P(n)}\frac{\ch_{ T}\left(\det\mathcal T\right)^{1-g}\ch_{ T}\Lambda_{y}\left[(T_\lambda^\vee(\mathbb C^2)^{[n]})^{\oplus g}\right]}{\ch_{ T}\Lambda_{-1}\left[T^\vee_\lambda(\mathbb C^2)^{[n]}\right]},
\end{split}
\end{equation}
where $\det\mathcal T$ denotes the determinant line bundle on $(\mathbb C^2)^{[n]}$ and $y=\eu^{-m}$.

It was proved in \cite{Chuang:2013wpa} that a similar result holds true also for the genus $g$ Cauchy function relative to punctured Riemann surfaces with non-trivial holonomy around the punctures. In the case of a single puncture (assumed to be generic) of type $\mu$, the Cauchy function at fixed $|\lambda|=n$ computes the residual equivariant Hirzebruch genus of a vector bundle over a nested Hilbert scheme of $n$ points $\mathcal N_{1,[1^1],n,\boldsymbol\mu}$ on $\mathbb C^2$:
\begin{equation}\label{char_vecbun_nested}
\sum_{\lambda\in\mathcal P(n)}\mathcal H_\lambda(z,w)\tilde H_\lambda(\bold x;z^2,w^2)=\ch_{ T}\chi_y\left[\pi^*\mathcal V_{g},\mathcal N_{1,[1^1],n,\boldsymbol\mu}\right],
\end{equation}
where $\pi:\mathcal N_{1,[1^1],n,\boldsymbol\mu}\to(\mathbb C^2)^{[n]}$ is the natural projection of the nested Hilbert scheme to the underlying Hilbert scheme of $n$ points on $\mathbb C^2$, and $\mathcal V_g=\left(T^\vee(\mathbb C^2)^{[n]}\right)^{\oplus g}\otimes(\det\mathcal T)^{1-g}$. Moreover the rhs of \eqref{char_vecbun_nested} can be computed in terms only of characters of vector bundles over $(\mathbb C^2)^{[n]}=\Hilb^n(\mathbb C^2)$ due to a result by Haiman, \cite{Chuang:2013wpa,haiman}, and we have that
\begin{equation}
\ch_{T}\chi_y\left[\pi^*\mathcal V_{g},\mathcal N_{1,[1^1],n,\boldsymbol\mu}\right]=\sum_{\lambda\in\mathcal P(n)}\frac{\ch_{ T}\left(\det\mathcal T\right)^{1-g}\ch_{ T}\Lambda_{y}\left[(T_\lambda^\vee(\mathbb C^2)^{[n]})^{\oplus g}\right]}{\ch_{ T}\Lambda_{-1}\left[T^\vee_\lambda(\mathbb C^2)^{[n]}\right]}\ch_{ T}(\mathcal P^\gamma_\mu),
\end{equation}
where $\mathcal P^\gamma$ is a vector bundle over $(\mathbb C^2)^{[n]}$ whose fibers over closed points $[I]\in(\mathbb C^2)^{[n]}$ are isomorphic to permutation representations of $\mathcal S_n$.

By virtue of what we showed in subsection
\ref{sec:partition_functions_QM}, we expect our results to give a virtual refinement of the formulae found in \cite{Chuang:2013wpa,HLRV}. For the sake of simplicity, let us start from studying the case of a quiver consisting of only two gauge nodes and $r=1$, corresponding to a complex curve $\mathcal C$ of genus $g=0$. We already computed in subsection \ref{sec:partition_functions_QM} the partition function relative to any generic quiver of the type shown in figure \ref{fig:quiver_local}, with $(r_0,r_1,\dots,r_{s-1})=(r,0,\dots,0)$.
We will then be computing the generating function
\begin{equation}
Z^{(p_0,p_1)}_{\textrm{vir}}=\sum_{\bold n\in\mathbb Z^2_{\ge 0}}Z^{(p_0,p_1)}_{\bold n}\prod_{i=0}^1x_i^{n_i}=\sum_{\bold n\in\mathbb Z_{\ge 0}^2}\ch_{T}\chi_{T}^{\textrm{vir}}\left(\mathcal N_{1,[1^1],n,\gamma(\bold n)},\mathcal L_{(p_0,p_1)}\right)\prod_{i=0}^1x_i^{n_i},
\end{equation}
where $\gamma(\bold n)$ is the ordered sequence determined by $n_i$ determining the relevant quiver variety of numerical type $(1,\hat n_0,\hat n_1)$.

We will restrict our attention to $\bold p=(p_0,0)$, in which case the restriction $\mathcal L_Z$ of $\mathcal L_{(p_0,0)}$ to the fixed point under $T\curvearrowright\mathcal N_{1,[1^1],n,\gamma(\bold n)}$ is
\begin{equation}
\mathcal L_Z=\left(\prod_{i=1}^{M_1}\prod_{j=1}^{\nu_i'}T_1^{-i+1}T_2^{-j+1}\right)^{p_0}.
\end{equation}
The result obtained in subsection \ref{sec:partition_functions_QM} by means of SUSY localization then specializes in this case to the form \eqref{vir_tang_char}:
\begin{equation}\label{vir_tang_char}
\begin{split}
Z^{(p_0,0)}_{\bold n}=\sum_{\substack{Z=(\nu,\mu)\\ (|\nu|,|\mu|)=\gamma(\bold n)}}\frac{\ch_{T}\mathcal L_Z}{\Lambda_{-1}\left[T_Z^{\textrm{vir}}\mathcal N_{1,[1^1],n,\gamma(\bold n)}^\vee\right]}=\sum_{\substack{Z=(\nu,\mu)\\ (|\nu|,|\mu|)=\gamma(\bold n)}}\frac{\mathcal L_{\nu}(q_1,q_2)\tilde{\mathcal W}_{(\nu,\mu)}(q_1,q_2)}{\Lambda_{-1}\left[T_{\tilde Z}\mathcal M_{1,n_0}^\vee\right]},
\end{split}
\end{equation}
with
\begin{equation}
\mathcal L_{\nu}(q_1,q_2)=\left(\prod_{i=1}^{M_1}\prod_{j=1}^{\nu_i'}q_1^{i-1}q_2^{j-1}\right)^{p_0},
\end{equation}
and
\begin{equation}
\tilde{\mathcal W}_{(\nu,\mu)}=\prod_{i=1}^{M_1}\prod_{j=1}^{N_1}\frac{(1-q_1^{\mu_j-i}q_2^{j-\nu'_i-1})(1-q_1^{-i}q_2^{j-\mu'_i-1})}{(1-q_1^{\mu_j-i}q_2^{(j-\mu'_i-1})(1-q_1^{-i}q_2^{j-\nu'_i-1})}\prod_{i=1}^{M_1}\prod_{j=1}^{\nu'_i-\mu'_i}\frac{(1-q_1^{-i}q_2^{-j-\mu'_i})}{(1-q_1^{-1}q_2^{-1})},
\end{equation}
where, as usual, $q_1=\ch_{ T}T_1$ and $q_2=\ch_{ T}T_2$.


In order to support our conjecture that the quiver we studied so far do indeed provide an ADHM-type construction for the nested Hilbert scheme of points on $\mathbb C^2$ we will show some relevant examples in the following. In the two-steps quiver case this is true by a result of \cite{flach_jardim}, which moreover implies that the non-abelian quiver provides an ADHM description for the moduli space of framed torsion-free flags of sheaves on $\mathbb P^2$. A very brief review of the result of \cite{flach_jardim} which are useful for what follows can be found in appendix \ref{appendix:fixed_pts}. Even in the two-steps case we can still compare the results coming from direct localization computations to the formulae in \cite{HLRV,Chuang:2013wpa}. In particular, since the nested Hilbert scheme of points is known to be non smooth except for the case $(n_0,n_1)=(n,1)$ or $(n_0,n_1)=(n,0)$, the polynomials we get multiplied by the Nekrasov partition function order by order are expected to reproduce the modified Macdonald polynomials $\tilde H_\lambda(\bold x;q,t)$ when $n_1=1$. For the sake of ease of comparison, in what follows we will use the notation $\mathcal N(r,n_0,\dots,n_{s-1})$, which is found in \cite{flach_jardim,Chuang:2013wpa}, instead of $\mathcal N_{r,[r^1],n,\mu}$.

\begin{example}
If $\bold n=(n,0)$ we need to compute the partition function for $\mathcal N(1,n,0)$, and obviously the partition function reproduces the result in equation \eqref{Omega_Jordang}, for $g=0$.
\end{example}

\begin{example}
Take $\bold n=(1,1)$, so that $\mathcal F(1,1,1)\simeq\mathcal N(1,2,1)\simeq\Hilb^{(1,2)}(\mathbb C^2)$, \cite{flach_jardim}. We have two different choices for the fixed points:
\begin{equation}
(\nu,\mu)=\vcenter{\hbox{\ShadedTableau[(1,0)]{{\ ,\ }}}}=(2^1,1^1)\qquad\qquad\emph{or}\qquad\qquad (\nu,\mu)=\vcenter{\hbox{\ShadedTableau[(1,-1)]{\ ,\ }}}=(1^2,1^1)
\end{equation}
and we have for the partition function
\begin{equation}
Z_{\bold n}^{(1-g,0)}(\bold x;q,t)=\sum_{\nu}Z_{\bold n,\nu}^{(1-g,0)}(\bold x;q,t)=x_0x_1\left\{\sum_{(\nu,\mu)}\frac{\mathcal L_{\nu}(q^{-1},t^{-1})\tilde{\mathcal W}_{(\nu,\mu)}(q^{-1},t^{-1})}{\Lambda_{-1}\left[T_{\tilde Z}\mathcal M_{1,n_0}^\vee\right]}\right\}
\end{equation}
with
\begin{equation}
\left\{
\begin{aligned}
&Z_{\bold n,2^1}^{(1-g,0)}(\bold x;q,t)=\frac{\mathcal L_{2^1}(q^{-1},t^{-1})}{\Lambda_{-1}\left[T_{2^1}\mathcal M_{1,n_0}^\vee\right]}(1+q)x_0x_1\\
&Z_{\bold n,1^2}^{(1-g,0)}(\bold x;q,t)=\frac{\mathcal L_{1^2}(q^{-1},t^{-1})}{\Lambda_{-1}\left[T_{1^2}\mathcal M_{1,n_0}^\vee\right]}(1+t)x_0x_1
\end{aligned}\right.
\end{equation}
By putting together with the previous example, we have that
\begin{equation}
\begin{split}
Z_{|\bold n|=2}^{1-g,0}&=\sum_{\nu\in\mathcal P(2)}\frac{\mathcal L_{\nu}(q^{-1},t^{-1})}{\Lambda_{-1}\left[T_{\nu}\mathcal M_{1,n_0}^\vee\right]}\tilde H_{\nu}(x_0,x_1;q,t)\\
&=\sum_{\nu\in\mathcal P(2)}\mathcal H_\nu^{g=0}(z,w)\tilde H_\nu(x_0,x_1;z^2,w^2)
\end{split}
\end{equation}
We want to point out that the elliptic counterpart to the polynomials determined by $\tilde{\mathcal W}_{(\nu,\mu)}$ are the following
\begin{equation}\label{duepall}
\left\{
\begin{aligned}
&\left.P^{\rm ell}_{\ShadedTableauS[]{{\ ,\ }}}(\bold x;\epsilon_1,\epsilon_2)\right|_{x_0x_1}=\frac{\theta_1(\tau|2\epsilon_1)}{\theta_1(\tau|\epsilon_1)},\\
&\left.P^{\rm ell}_{\ShadedTableauS[]{\ ,\ }}(\bold x;\epsilon_1,\epsilon_2)\right|_{x_0x_1}=\frac{\theta_1(\tau|2\epsilon_2)}{\theta_1(\tau|\epsilon_2)},
\end{aligned}
\right.
\end{equation}
which obviously reduce to the corresponding modified Macdonald polynomials coefficients when $\tau\to\iu\infty$.
\end{example}

\begin{example}
Let's now consider $\bold n$ to be such that $n_0+n_1=3$. The only quantity we need to compute is related to $\bold n=(2,1)$, which corresponds to $\mathcal N(1,3,1)$. We have the following possibilities for the fixed points:
\begin{equation}
\{(\nu,\mu)\}=\left\{\vcenter{\hbox{\ShadedTableau[(1,0),(2,0)]{{\ ,\ ,\ }}}}\ ,\ \vcenter{\hbox{\ShadedTableau[(1,-1),(2,-1)]{{\ },{\ ,\ }}}}\ ,\ \vcenter{\hbox{\ShadedTableau[(1,-1),(1,0)]{{\ },{\ ,\ }}}}\ ,\ \vcenter{\hbox{\ShadedTableau[(1,-2),(1,-1)]{\ ,\ ,\ }}}\right\}
\end{equation}
and
\begin{equation}
\left\{
\begin{aligned}
&\tilde{\mathcal W}_{\vcenter{\hbox{\ShadedTableauS[(1,0),(2,0)]{{\ ,\ ,\ }}}}}(q^{-1},t^{-1})=(1+q+q^2)\\
&\tilde{\mathcal W}_{\vcenter{\hbox{\ShadedTableauS[(1,-1),(2,-1)]{{\ },{\ ,\ }}}}}(q^{-1},t^{-1})+\mathcal W_{\vcenter{\hbox{\ShadedTableauS[(1,-1),(1,0)]{{\ },{\ ,\ }}}}}(q^{-1},t^{-1})=(1+q+t)\\
&\tilde{\mathcal W}_{\vcenter{\hbox{\ShadedTableauS[(1,-2),(1,-1)]{\ ,\ ,\ }}}}(q^{-1},t^{-1})=(1+t+t^2)
\end{aligned}
\right.
\end{equation}
As in the previous example, we can exhibit explicitly the elliptic counterparts to these modified Macdonald polynomials, which read:
\begin{equation}\label{trepall}
\left\{
\begin{aligned}
&\left.P^{\rm ell}_{\ShadedTableauS[]{{\ ,\ ,\ }}}(\bold x;\epsilon_1,\epsilon_2)\right|_{x_0^2x_1}=\frac{\theta_1(\tau|3\epsilon_1)}{\theta_1(\tau|\epsilon_1)},\\
&\left.P^{\rm ell}_{\ShadedTableauS[]{{\ },{\ ,\ }}}(\bold x;\epsilon_1,\epsilon_2)\right|_{x_0^2x_1}=\left(\frac{\theta_1(\tau|2\epsilon_1-\epsilon_2)}{\theta_1(\tau|\epsilon_1-\epsilon_2)}+\frac{\theta_1(\tau|2\epsilon_2-\epsilon_1)}{\theta_1(\tau|\epsilon_2-\epsilon_1)}\right),\\
&\left.P^{\rm ell}_{\ShadedTableauS[]{\ ,\ ,\ }}(\bold x;\epsilon_1,\epsilon_2)\right|_{x_0^2x_1}=\frac{\theta_1(\tau|3\epsilon_2)}{\theta_1(\tau|\epsilon_2)}.
\end{aligned}
\right.
\end{equation}
\end{example}

\begin{example}
As a final example of a smooth nested Hilbert scheme of points we will take $\mathcal N(1,4,1)$, so that the fixed points will be
\begin{equation}
\{(\nu,\mu)\}=\left\{
\vcenter{\hbox{\ShadedTableau[(1,0),(2,0),(3,0)]{{\ ,\ ,\ ,\ }}}}\ ,\ 
\vcenter{\hbox{\ShadedTableau[(1,-1),(2,-1),(3,-1)]{{\ },{\ ,\ ,\ }}}}\ ,\ 
\vcenter{\hbox{\ShadedTableau[(1,0),(1,-1),(2,-1)]{{\ },{\ ,\ ,\ }}}}\ ,\ 
\vcenter{\hbox{\ShadedTableau[(1,0),(1,-1),(2,-1)]{{\ ,\ },{\ ,\ }}}}\ ,\ 
\vcenter{\hbox{\ShadedTableau[(1,0),(1,-1),(1,-2)]{{\ },{\ },{\ ,\ }}}}\ ,\ 
\vcenter{\hbox{\ShadedTableau[(1,-1),(1,-2),(2,-2)]{{\ },{\ },{\ ,\ }}}}\ ,\ 
\vcenter{\hbox{\ShadedTableau[(1,-1),(1,-2),(1,-3)]{\ ,\ ,\ ,\ }}}
\right\}
\end{equation}
by which we get
\begin{equation}
\left\{
\begin{aligned}
&\tilde{\mathcal W}_{\vcenter{\hbox{\ShadedTableauS[(1,0),(2,0),(3,0)]{{\ ,\ ,\ ,\ }}}}}(q^{-1},t^{-1})=(1+q+q^2+q^3)\\
&\tilde{\mathcal W}_{\vcenter{\hbox{\ShadedTableauS[(1,-1),(2,-1),(3,-1)]{{\ },{\ ,\ ,\ }}}}}(q^{-1},t^{-1})+\tilde{\mathcal W}_{\vcenter{\hbox{\ShadedTableauS[(1,0),(1,-1),(2,-1)]{{\ },{\ ,\ ,\ }}}}}(q^{-1},t^{-1})=(1+q+q^2+t)\\
&\tilde{\mathcal W}_{\vcenter{\hbox{\ShadedTableauS[(1,0),(1,-1),(2,-1)]{{\ ,\ },{\ ,\ }}}}}(q^{-1},t^{-1})=(1+q+t+qt)\\
&\tilde{\mathcal W}_{\vcenter{\hbox{\ShadedTableauS[(1,0),(1,-1),(1,-2)]{{\ },{\ },{\ ,\ }}}}}(q^{-1},t^{-1})+\tilde{\mathcal W}_{\vcenter{\hbox{\ShadedTableauS[(1,-1),(1,-2),(2,-2)]{{\ },{\ },{\ ,\ }}}}}(q^{-1},t^{-1})=(1+t+t^2+q)\\
&\tilde{\mathcal W}_{\vcenter{\hbox{\ShadedTableauS[(1,-1),(1,-2),(1,-3)]{\ ,\ ,\ ,\ }}}}(q^{-1},t^{-1})=(1+t+t^2+t^3)\\
\end{aligned}
\right.
\end{equation}
which again reproduce modified Macdonald polynomials which can be found tabulated in the mathematical literature. Their elliptic counterpart is now given by:
\begin{equation}\label{quattropall}
\left\{
\begin{aligned}
&\left.P^{\rm ell}_{\ShadedTableauS[]{{\ ,\ ,\ ,\ }}}(\bold x;\epsilon_1,\epsilon_2)\right|_{x_0^3x_1}=\frac{\theta_1(\tau|4\epsilon_1)}{\theta_1(\tau|\epsilon_1)},\\
&\left.P^{\rm ell}_{\ShadedTableauS[]{{\ },{\ ,\ ,\ }}}(\bold x;\epsilon_1,\epsilon_2)\right|_{x_0^3x_1}=\left(\frac{\theta_1(\tau|2\epsilon_1)}{\theta_1(\tau|\epsilon_1)}\frac{\theta_1(\tau|3\epsilon_1-\epsilon_2)}{\theta_1(\tau|2\epsilon_1-\epsilon_2)}+\frac{\theta_1(\tau|2\epsilon_2-2\epsilon_1)}{\theta_1(\tau|\epsilon_2-2\epsilon_1)}\right),\\
&\left.P^{\rm ell}_{\ShadedTableauS[]{{\ ,\ },{\ ,\ }}}(\bold x;\epsilon_1,\epsilon_2)\right|_{x_0^3x_1}=\left(\frac{\theta_1(\tau|2\epsilon_1)}{\theta_1(\tau|\epsilon_1)}\frac{\theta_1(\tau|2\epsilon_2)}{\theta_1(\tau|\epsilon_2)}\right),\\
&\left.P^{\rm ell}_{\ShadedTableauS[]{{\ },{\ },{\ ,\ }}}(\bold x;\epsilon_1,\epsilon_2)\right|_{x_0^3x_1}=\left(\frac{\theta_1(\tau|2\epsilon_2)}{\theta_1(\tau|\epsilon_2)}\frac{\theta_1(\tau|3\epsilon_2-\epsilon_1)}{\theta_1(\tau|2\epsilon_2-\epsilon_1)}+\frac{\theta_1(\tau|2\epsilon_1-2\epsilon_2)}{\theta_1(\tau|\epsilon_1-2\epsilon_2)}\right),\\
&\left.P^{\rm ell}_{\ShadedTableauS[]{\ ,\ ,\ ,\ }}(\bold x;\epsilon_1,\epsilon_2)\right|_{x_0^3x_1}=\frac{\theta_1(\tau|4\epsilon_2)}{\theta_1(\tau|\epsilon_2)}.
\end{aligned}
\right.
\end{equation}
\end{example}

The following is the easiest example of a non smooth nested Hilbert scheme, namely $\mathcal N(1,4,2)$, and we can see how in this case our computation doesn't reproduce the $\chi_y$ genus of \cite{Chuang:2013wpa}, hence the formulae of \cite{HLRV}, giving instead their virtual generalization.

\begin{example}
Take $(n_0,n_1)=(4,2)$. The prescription for the fixed points gives us
\begin{equation}
\{(\nu,\mu)\}=\left\{
\vcenter{\hbox{\ShadedTableau[(1,0),(2,0)]{{\ ,\ ,\ ,\ }}}}\ ,\ 
\vcenter{\hbox{\ShadedTableau[(1,-1),(2,-1)]{{\ },{\ ,\ ,\ }}}}\ ,\ 
\vcenter{\hbox{\ShadedTableau[(1,0),(1,-1)]{{\ },{\ ,\ ,\ }}}}\ ,\ 
\vcenter{\hbox{\ShadedTableau[(1,0),(1,-1)]{{\ ,\ },{\ ,\ }}}}\ ,\ 
\vcenter{\hbox{\ShadedTableau[(1,-1),(2,-1)]{{\ ,\ },{\ ,\ }}}}\ ,\
\vcenter{\hbox{\ShadedTableau[(1,-1),(1,-2)]{{\ },{\ },{\ ,\ }}}}\ ,\ 
\vcenter{\hbox{\ShadedTableau[(1,-2),(2,-2)]{{\ },{\ },{\ ,\ }}}}\ ,\
\vcenter{\hbox{\ShadedTableau[(1,-2),(1,-3)]{\ ,\ ,\ ,\ }}}
\right\}
\end{equation}

by which we get
\begin{equation}
\left\{
\begin{aligned}
&\tilde{\mathcal W}_{\vcenter{\hbox{\ShadedTableauS[(1,0),(2,0)]{{\ ,\ ,\ ,\ }}}}}(q^{-1},t^{-1})=1+q+2q^2+q^3+q^4-q^2t-q^3t-2q^4t-q^5t-q^6t\\
&\tilde{\mathcal W}_{\vcenter{\hbox{\ShadedTableauS[(1,-1),(2,-1)]{{\ },{\ ,\ ,\ }}}}}(q^{-1},t^{-1})+\tilde{\mathcal W}_{\vcenter{\hbox{\ShadedTableauS[(1,0),(1,-1)]{{\ },{\ ,\ ,\ }}}}}(q^{-1},t^{-1})=1+q+2q^2+t+qt-q^2t-q^3t\\&\qquad\qquad\qquad\qquad\qquad\qquad\qquad\quad-q^4t-qt^2-q^2t^2-q^3t^2\\
&\tilde{\mathcal W}_{\vcenter{\hbox{\ShadedTableauS[(1,0),(1,-1)]{{\ ,\ },{\ ,\ }}}}}(q^{-1},t^{-1})+\tilde{\mathcal W}_{\vcenter{\hbox{\ShadedTableauS[(1,-1),(2,-1)]{{\ ,\ },{\ ,\ }}}}}(q^{-1},t^{-1})=1+q+q^2+t+qt+t^2-q^2t-qt^2\\&\qquad\qquad\qquad\qquad\qquad\qquad\qquad\quad-2q^2t^2-q^3t^2-q^2t^3\\
&\tilde{\mathcal W}_{\vcenter{\hbox{\ShadedTableauS[(1,-1),(1,-2)]{{\ },{\ },{\ ,\ }}}}}(q^{-1},t^{-1})+\tilde{\mathcal W}_{\vcenter{\hbox{\ShadedTableauS[(1,-2),(2,-2)]{{\ },{\ },{\ ,\ }}}}}(q^{-1},t^{-1})=1+q+t+qt+2t^2-q^2t-qt^2\\&\qquad\qquad\qquad\qquad\qquad\qquad\qquad\quad-q^2t^2-qt^3-q^2t^3-qt^4\\
&\tilde{\mathcal W}_{\vcenter{\hbox{\ShadedTableauS[(1,-2),(1,-3)]{\ ,\ ,\ ,\ }}}}(q^{-1},t^{-1})=1+t+2t^2+t^3+t^4-qt^2-qt^3-2qt^4-qt^5-qt^6\\
\end{aligned}.
\right.
\end{equation}
The polynomials above contain the coefficients for the modified Macdonald polynomials which in this case read
\begin{equation}
\left\{
\begin{aligned}
\tilde H_{\vcenter{\hbox{\ShadedTableauS[]{{\ ,\ ,\ ,\ }}}}}(q,t)|_{x_0^2x_1^2}&=1+q+2q^2+q^3+q^4\\
\tilde H_{\vcenter{\hbox{\ShadedTableauS[]{{\ },{\ ,\ ,\ }}}}}(q,t)|_{x_0^2x_1^2}&=1+q+2q^2+t+qt\\
\tilde H_{\vcenter{\hbox{\ShadedTableauS[]{{\ ,\ },{\ ,\ }}}}}(q,t)|_{x_0^2x_1^2}&=1+q+q^2+t+qt+t^2\\
\tilde H_{\vcenter{\hbox{\ShadedTableauS[]{{\ },{\ },{\ ,\ }}}}}(q,t)|_{x_0^2x_1^2}&=1+q+t+qt+2t^2\\
\tilde H_{\vcenter{\hbox{\ShadedTableauS[]{\ ,\ ,\ ,\ }}}}(q,t)|_{x_0^2x_1^2}&=1+t+2t^2+t^3+t^4
\end{aligned}
\right.
\end{equation}
\end{example}

As a final remark let us point out that, even though the GLSM partition function is naturally computing virtual invariants, as the moduli space $\mathcal N(r,n_0,n_1)$ is in general a singular quasi-projective variety, \cite{cheah}, however one should be able to use equivariant localization to compute usual topological invariants also for singular varieties \cite{2013arXiv1308.0787W,2013arXiv1308.0788W}. 
\begin{appendices}
\section{Low energy theory for D3/D7}
Let us here sketch a derivation of the low energy effective theory of the D3-D7 system 
at an orbifold point
by
studying the equations of motion  
reduced on $T^2\times{\mathcal C}$. This amounts to solve the BPS equations
\begin{eqnarray}\label{forse}
F^{(2,0)}=0,\quad \partial_A\Phi_S=0,\quad \partial_AB_i=0, \quad \partial_AI=0,\quad \partial_AJ=0\\
\omega\cdot F+[B_i,B_i^\dagger]+[\Phi_S,\Phi_S^\dagger]+I^\dagger I-JJ^\dagger=\zeta {\mathbf 1}_N
\end{eqnarray}
while we minimise the super potential
\begin{equation}\label{si}
{\mathcal W}={\rm Tr}\left\{\Phi_S\left(\left[B_1,B_2\right]+IJ\right)\right\}\, .
\end{equation}
Let us now focus in the vicinity of the orbifold point, where the local geometry of 
${\mathcal C}$ is ${\mathbb C}/{\mathbb Z}_s$ and that of $T^*{\mathcal C}$
is the ALE quotient ${\mathbb C}^2/{\mathbb Z}_s$. There the Chan-Paton bundle
of the open string modes decomposes in ${\mathbb Z}_s$-representations as already 
discussed in Section 2.
\eqref{forse} admit vortex solutions centered at the orbifold point, whose vorticity is fixed by the order of the cyclic group. On the vortex background, the gauge field along 
${\mathbb C}/{\mathbb Z}_s$ becomes massive due to the Higgs mechanism
and decouples from the low energy spectrum.

Unpacking the open strings moduli in the $V_j$ twisted sectors one gets the 
degrees of freedom in \eqref{campi37} and the relations \eqref{rels}. 
Let us now discuss how these arise.
The modes $B_1^j$ and $B_2^j$ come from the ${\mathbb Z}_s$ representation of the $B_1$ and $B_2$ fields and analogously $I^j$ and $J^j$ from $I$ and $J$.
The further degrees of freedom arise from $\Phi_S$, that is the one-form in the adjoint. 
 Since these are describing open string modes in twisted directions under the 
 ${\mathbb Z}_s$ group, the 
fields which arise from $\Phi_S$ are homomorphisms between nearby twisted sectors. 
Explicitly from the reduction of $\Phi_S$ one gets the bifundamental modes $F^j\in {\rm Hom}\left(V_j,V_{j+1}\right)$.

The BPS vacua equations of this system therefore are obtained from the reduction to the 
constant modes of \eqref{forse} and the minimization of the super potential
\begin{equation}
[B_1^j,B_2^j]+I^jJ^j=0\, , \quad B_1^jF^j-F^jB_1^{j+1}=0\, \quad B_2^jF^j-F^jB_2^{j+1}=0\, , \quad J^jF^j=0.\nonumber
\end{equation}
\section{Flags of framed torsion-free sheaves on $\mathbb P^2$}
As we already pointed out the QM partition function obtained as the trigonometric limit of our D3/D7 system computes virtual invariants of  a certain $T-$equivariant bundle over the moduli space $\mathcal N(r,n_0,\dots,n_{s-1})$ of stable representations of the quiver in figure \ref{fig:quiver_sphere}. When the quiver is two-steps it is called in the mathematical literature the {\it enhanced ADHM quiver} and the moduli space of its stable representations of type $(r,n_0,n_1)$ has been identified in \cite{flach_jardim} with the moduli space of flags of framed torsion-free sheaves on $\mathbb P^2$, $\mathcal F(r,n_0-n_1,n_1)$. These are defined as follows. Once a line $\ell_\infty\subset\mathbb P^2$ is fixed, a framed flag of sheaves consists of a triple $(E,F,\varphi)$, where $F$ is a rank$-r$ torsion free sheaf on $\mathbb P^2$, framed at $\ell_\infty$ via $\varphi:F_{\ell_\infty}\xrightarrow{\simeq}\mathcal O_{\ell_\infty}^{\oplus r}$, while $E$ is a subsheaf of $F$ such that the quotient $F/E$ is supported away from $\ell_\infty$. This triple is characterized by three numerical invariants: $r=\rk E=\rk F$, $n=c_2(F)$ and $l$ such that $c_2(E)=n+l$. The moduli space of flags of framed torsion-free sheaves on $\mathbb P^2$ is thus parametrized by these three numerical invariants, and it is denoted by $\mathcal F(r,n,l)$. Moreover, if $\mathcal M(r,n)$ denotes the moduli space of framed torsion-free sheaves on $\mathbb P^2$, $\mathcal M(r,n)\simeq\mathcal M_{r,n}$, one has that $\mathcal F(r,n,l)\hookrightarrow\mathcal M(r,n)\times\mathcal M(r,n+l)$ as an incidence variety.

These moduli spaces are of particular interest to us because of the following theorem \ref{thm:jardim}.

\begin{theorem}[von Flach-Jardim, \cite{flach_jardim}]\label{thm:jardim}
The moduli space $\mathcal N(r,n_0,n_1)\simeq\mathcal F(r,n_0-n_1,n_1)$ of stable representations of the enhanced ADHM quiver is a quasi-projective variety equipped with a perfect obstruction theory. The following $T-$equivariant complex $\mathcal C(X)$
\begin{equation}
{\footnotesize 
\begin{tikzcd}[row sep=-1mm]
& & Q\otimes\End(V_0)\\
& &\oplus & \Lambda^2Q\otimes\End(V_0)\\
& &\Hom(W,V_0) & \oplus\\
&\End(V_0) &\oplus & Q\otimes\Hom(V_1,V_0)\\
& \quad\oplus\quad \arrow[r,"d_0"] &\Lambda^2Q\otimes\Hom(V_0,W) \arrow[r,"d_1"] & \quad\qquad\oplus\qquad\quad \arrow[r,"d_2"] & \Lambda^2Q\otimes\Hom(V_1,V_0)\\
&\End(V_1) &\oplus & \Lambda^2Q\otimes\Hom(V_1,W)\\
& & Q\otimes\End(V_1) & \oplus\\
& &\oplus & \Lambda^2Q\otimes\End(V_1)\\
& &\Hom(V_1,V_0)
\end{tikzcd}
}
\end{equation}
with
\begin{displaymath}
\left\{
\begin{aligned}
&d_0(h_0,h_1)=\left([h_0,B_1^0],[h_0,B_2^0],h_0I,-Jh_0,[h_1,B_1^1],[h_1,B_2^1],h_0F-Fh_1\right) \\
&d_1(b_1^0,b_2^0,i,j,b_1^1,b_2^1,f)=\left([b_1^0,B_2^1]+[B_1^0,b_2^0]+iJ+Ij,B_1^0f+b_1^0F-Fb_1^1-fB_1^1,\right. \\
&\qquad\qquad\qquad\qquad\qquad\quad \left.B_2^0f+b_2^0F-Fb_2^1-fB_2^1,jF+Jf,[b_1^1,B_2^1]+[B_1^1,b_2^1]\right)\\
&d_2(c_1,c_2,c_3,c_4,c_5)=c_1F+B_2^0c_2-c_2B_2^1+c_3B_1^0-B_1^1c_3-Ic_4-Fc_5
\end{aligned}\right.
\end{displaymath}
encodes the structure of the perfect obstruction theory for $\mathcal N(r,n_0,n_1)$. The infinitesimal deformation space and the obstruction space at any $X$ will be isomorphic to $H^1[\mathcal C(X)]$ and $H^2[\mathcal C(X)]$, respectively. $\mathcal N(r,n_1,n_2)$ is smooth iff $n_1=1$, \cite{cheah}.
\end{theorem}

Moreover, it is shown in \cite{flach_jardim} that there exists a surjective morphism 
$$\mathfrak q:(W,\{V_i,B_1^i,B_2^i\},I,J,F)\mapsto (W,V,B_1',B_2',I',J')$$
mapping the enhanced ADHM data of type $(r,n_0,n_1)$ to the ADHM data of numerical type $(r,n_0-n_1)$. This morphism is moreover compatible with the natural forgetting morphism $\eta:\mathcal N(r,n_0,n_1)\to\mathcal M(r,n_0)$, so that we have two different maps sending the moduli space of stable representations of the enhanced ADHM quiver to the moduli space of stable representations of ADHM data. The situation is depicted by the following commutative diagram
\begin{figure}[H]
\centering
\begin{tikzcd}
\mathcal N(r,n_0,n_1)\arrow[r,"\eta"]\arrow[dr,"\mathfrak q"] & \mathcal M(r,n_0)\\
 & \mathcal M(r,n_0-n_1)\arrow[u,"\tilde f"]
\end{tikzcd}
\end{figure}
which enables us to characterize $T-$fixed points of $\mathcal N(r,n_0,n_1)$ in terms of fixed points of $\mathcal M(r,n_0)$ and $\mathcal M(r,n_0-n_1)$. Consistently with what we found in the more general case of a quiver with an arbitrary number of nodes, the fixed point locus consists of isolated non-degenerate points which can be described be couples of nested partitions $\mathcal P(n_0-n_1)\ni\mu\subseteq\nu\in\mathcal P(n_0)$.

\section{Fixed points and virtual dimension}\label{appendix:fixed_pts}

The characterization of the fixed points we described in section \ref{sec:fixed} makes it clear that the $T-$fixed locus in $\mathcal N_{r,[r^1],n,\mu}$ consists only of isolated non-degenerate points. Moreover through a simple computation it's now very easy to compute the virtual dimension of $\mathcal N(r,n_0,\dots,n_{s-1})$. Altogether these facts get summarized by the following proposition, which for the sake of simplicity we state in the simple case of the two-steps quiver.

\begin{proposition}
The $T-$fixed locus of the moduli space of nested instantons $\mathcal N(r,n_0,n_1)$ consists only of isolated non-degenerate points, which are into $1-1$ correspondence with $r-$tuples of colored nested partitions. Moreover ${\rm vd}_{\mathcal N(r,n_0,n_1)}=2rn_0-rn_1+1$.
\end{proposition}
\begin{proof}
A very brief sketch of how to prove the statement about the fixed points was previously given in section \ref{sec:fixed}, so now we will only focus on computing the virtual dimension of $\mathcal N(r,n_0,n_1)$. Using the description provided by quiver \ref{fig:quiver_sphere_math} we see that the number of variables involved in the computation is $\#\textrm{var}=2n_0^2+2n_1^2+2n_0r+n_0n_1$, with $r=\dim W$. Moreover, the number of constraints we need to implement is $\#\textrm{constr}=n_0^2+n_1^2-1+n_0n_1+n_1r$, where we also took into account that the constraints are not independent. Finally we account for the fact that we take the GIT quotient by the action of $GL(n_0)\times GL(n_1)$, which contributes by $\#\textrm{symm}=n_0^2+n_1^2$. Then
\begin{equation}
`` \dim\mathcal N(r,n_1,n_2) \text{''}=\#\textrm{var}-\#\textrm{constr}-\#\textrm{symm}=2n_0r-n_1r+1.
\end{equation}
In order to directly compute the virtual dimension of the nested Hilbert scheme of points on $\mathbb C^2$, we use the character decomposition of $T_Z^{\rm vir}\mathcal N(1,n_0,n_1)$ at a generic fixed point under the torus action. Then
\begin{equation}
\begin{split}
{\rm vd}_{\mathcal N(1,n_0,n_1)}&=\lim_{T_i\to 1}\left[T_{\tilde Z}\mathcal M(1,n_0)+\sum_{i=1}^{M_1}\sum_{j=1}^{N_1}(T_1^{i-\mu_j}-T_1^{i})(T_2^{-j+\mu_i'+1}-T_2^{-j+\nu_i'+1})\right.\\
&\qquad\qquad\left.-\sum_{i=1}^{M_1}\sum_{j=1}^{\nu_i'-\mu_i'}T_1^iT_2^{j+\mu_i'}+T_1T_2\right]\\
&=2n_0-n_1+1,
\end{split}
\end{equation}
which, in the case of a smooth nested Hilbert scheme of points, coincides with the computation of \cite{flach_jardim}. A completely analogous computation can be carried out in the generic (non necessarily smooth) case, by using the character decomposition we computed for $T_Z^{\rm vir}\mathcal N(r,n_0,n_1)$, which in turn coincides with the representation of the virtual tangent space to the nested Hilbert scheme of points (when $r=1$) given in \cite{2017arXiv170108899G}.
\end{proof}
\end{appendices}
\bibliographystyle{biblio}
\bibliography{refs.bib}
\end{document}